%% file: beyond_kappa.tex
\documentclass[11pt]{article}

\usepackage{amssymb,amsmath,amsthm,amsfonts,dsfont}
\usepackage[numbers,comma,sort&compress]{natbib}
\usepackage[colorlinks,hypertexnames=false]{hyperref}
\hypersetup{
	pdfstartview={FitH},
	pdfnewwindow=true,
	colorlinks=true,
	linkcolor=blue,
	citecolor=blue,
	filecolor=blue,
	urlcolor=blue}
\usepackage[capitalise]{cleveref}
\usepackage[all]{hypcap}
\usepackage{url}
\usepackage{graphicx}
\usepackage[font=small]{caption}
\usepackage[subrefformat=parens,labelformat=parens]{subcaption}
\usepackage{float}
\usepackage{footnote}
\usepackage{enumitem}
\usepackage[margin=1in]{geometry}

\newcommand{\abs}[1]{\left\lvert#1\right\rvert}
\newcommand{\norm}[1]{\left\lVert#1\right\rVert}
\DeclareMathOperator{\poly}{poly}
\DeclareMathOperator{\polylog}{polylog}

\newtheorem{theorem}{Theorem}
\newtheorem{definition}{Definition}
\newtheorem{lemma}{Lemma}
\newtheorem{proposition}[lemma]{Proposition}
\newtheorem{corollary}[lemma]{Corollary}
\theoremstyle{remark}

\theoremstyle{plain}
\newtheorem*{definition*}{Definition} 

\newcommand{\eq}[1]{\cref{eq:#1}}
\newcommand{\thm}[1]{\hyperref[thm:#1]{Theorem~\ref*{thm:#1}}}
\newcommand{\defn}[1]{\hyperref[defn:#1]{Definition~\ref*{defn:#1}}}
\newcommand{\lem}[1]{\hyperref[lem:#1]{Lemma~\ref*{lem:#1}}}
\newcommand{\prop}[1]{\hyperref[prop:#1]{Proposition~\ref*{prop:#1}}}
\newcommand{\fig}[1]{\hyperref[fig:#1]{Figure~\ref*{fig:#1}}}
\newcommand{\tab}[1]{\hyperref[tab:#1]{Table~\ref*{tab:#1}}}
\renewcommand{\sec}[1]{\hyperref[sec:#1]{Section~\ref*{sec:#1}}}
\newcommand{\append}[1]{\hyperref[append:#1]{Appendix~\ref*{append:#1}}}
\newcommand{\cor}[1]{\hyperref[cor:#1]{Corollary~\ref*{cor:#1}}}
\newcommand{\obs}[1]{\hyperref[obs:#1]{Observation~\ref*{obs:#1}}}

\newcommand{\ket}[1]{|#1\rangle}
\newcommand{\bra}[1]{\langle#1|}
\newcommand{\ketbra}[2]{\ket{#1}\!\bra{#2}}
\newcommand{\braket}[2]{\langle #1|#2\rangle}

\usepackage{mathtools}

\DeclareFontFamily{U}{matha}{\hyphenchar\font45}
\DeclareFontShape{U}{matha}{m}{n}{
	<5> <6> <7> <8> <9> <10> gen * matha
	<10.95> matha10 <12> <14.4> <17.28> <20.74> <24.88> matha12
}{}
\DeclareSymbolFont{matha}{U}{matha}{m}{n}
\DeclareFontFamily{U}{mathx}{\hyphenchar\font45}
\DeclareFontShape{U}{mathx}{m}{n}{
	<5> <6> <7> <8> <9> <10>
	<10.95> <12> <14.4> <17.28> <20.74> <24.88>
	mathx10
}{}
\DeclareSymbolFont{mathx}{U}{mathx}{m}{n}
\DeclareMathSymbol{\obot}         {2}{matha}{"6B}
\DeclareMathSymbol{\bigobot}       {1}{mathx}{"CB}

\usepackage{nicematrix}
\usepackage{tikz}

\makeatletter
\def\newmaketag{%
  \def\maketag@@@##1{\hbox{\m@th\normalfont\normalsize##1}}%
  }
\makeatother

\usepackage[hyperpageref]{backref}
\renewcommand*{\backrefalt}[4]{%
\ifcase #1 %
No citations.%
\or
(Cited on page #2).%
\else
(Cited on pages #2).%
\fi
}

\usepackage{etoolbox}
\makeatletter
\patchcmd\NAT@citexnum{\let\NAT@last@num\NAT@num}{\MakeLinkTarget[cite]{}\Hy@backout{\@citeb\@extra@b@citeb}\let\NAT@last@num\NAT@num}{}{\fail}
\makeatother

\usepackage{etoolbox}
\apptocmd{\sloppy}{\hbadness 10000\relax}{}{}

\usepackage{authblk}

\newcommand{\keff}{\kappa_{\text{eff}}}

\usepackage{multirow} 
\usepackage{makecell}

\makeatletter
\newcommand{\thickhline}{%
    \noalign {\ifnum 0=`}\fi \hrule height 1.5pt
    \futurelet \reserved@a \@xhline
}
\newcolumntype{"}{@{\hskip\tabcolsep\vrule width 1pt\hskip\tabcolsep}}
\makeatother

\usepackage{xcolor}

\newif\ifrevision
\revisionfalse

\ifrevision
  \newenvironment{revised}{\color{blue}\ignorespaces}{\ifhmode\unskip\fi}
\else
  \newenvironment{revised}{\ignorespaces}{\ifhmode\unskip\fi}
\fi

\title{Faster quantum linear system solver beyond the condition number}
\author[1]{Alexander M.\ Dalzell}
\author[2,3]{Jianqiang Li}
\author[1]{Yuan Su}
\affil[1]{AWS Center for Quantum Computing, Pasadena, CA 91106, USA}
\affil[2]{Department of Computer Science, Rice University, Houston, TX 77005, USA}
\affil[3]{Department of Computer Science, The University of Texas at Dallas,\break Richardson, TX 75080, USA}

\date{}

\begin{document}
\maketitle

\begin{abstract}
The spectral condition number is a widely adopted measure of worst-case cost for quantum linear system solvers. Yet it can significantly overestimate the actual runtime for a typical problem instance.

We present two quantum algorithms that produce the normalized solution of linear system $Ax=\ket{b}$ to accuracy $\epsilon$ with complexity independent of the condition number $\kappa=\norm{A}\norm{A^{-1}}$. We focus on the standard input model where the coefficient matrix $A$ is accessed through a block encoding oracle and the initial vector $\ket{b}$ is prepared by a unitary. But we also introduce an affine dilation model that encodes $A$ and $\ket{b}$ jointly, allowing further refinements of the query complexity.

Our truncation-based solver makes an optimal number of queries to the preparation of $\ket{b}$ and $\operatorname{\mathbf{O}}\left(\keff\polylog\left(\frac{\keff}{\epsilon}\right)\right)$ queries to the block encoding of $A$. We prove a family of upper bounds on the effective condition number $\keff$, including $\keff\leq\frac{\norm{(A^\dagger A)^{-t/2}x}^{1/t}}{\norm{x}^{1/t}\epsilon^{1/t}}$ for positive even integer $t$ and $\keff\leq\frac{\norm{A^{-1\dagger}(A^\dagger A)^{-(t-1)/2}x}^{1/t}}{\norm{x}^{1/t}\epsilon^{1/t}}$ for positive odd $t$, going beyond the naive estimate in terms of $\kappa$.

Our filtering-based solver is extremely simple with a favorable runtime prefactor. In particular, the solver has query complexity
\begin{revised}
$3\frac{\norm{A^{-1\dagger}x}}{\norm{x}\epsilon}$
\end{revised}
to leading order when the solution norm $\norm{x}$ is known. We then present a similarly simple solution norm estimator with the same asymptotic cost up to logarithmic factors.

Our quantum linear system solvers thus substantially improve a recent algorithm of Li, which has runtime $\operatorname{\mathbf{O}}\left(\frac{\norm{A^{-1\dagger}x}}{\norm{x}\epsilon^2}\right)$ if $\norm{x}$ is known and $\operatorname{\mathbf{O}}\left(\frac{\norm{A^{-1\dagger}x}\norm{x}^2}{\epsilon^2}\right)$ otherwise. This enables faster quantum linear system solving beyond the condition number, bringing prohibitively ill-conditioned systems closer to feasibility.
\end{abstract}

\newpage
{
	\thispagestyle{empty}
	\clearpage\tableofcontents
	\thispagestyle{empty}
}
\newpage

\section{Introduction}
\label{sec:intro}
\input{intro.tex}

\section{Quantum linear system problem and input models}
\label{sec:input}
\input{input.tex}

\section{Beyond-\texorpdfstring{$\kappa$}{k} solver based on effective truncation of linear system}
\label{sec:trunc}
\input{trunc.tex}

\section{Beyond-\texorpdfstring{$\kappa$}{k} solver based on filtering with effective gap}
\label{sec:gap}
\input{gap.tex}

\section{Discussion}
\label{sec:discuss}
\input{discuss.tex}

\section*{Acknowledgements}

A.M.D.~and Y.S.~thank Fernando Brand\~ao, Oskar Painter, James Hamilton, Nafea Bshara, Peter DeSantis, and Andy Jassy for their involvement and support of the research activities at the AWS Center for Quantum Computing, and Michael Kastoryano for reviewing an earlier draft. J.L.~was supported by NSF Award FET-2243659 and NSF Career Award FET-2339116, Welch Foundation Grant no. A22-0307, a Microsoft Research Award, an Amazon Research Award, and a Ken Kennedy Research Cluster award (Rice K2I), in part through seed funding from the Ken Kennedy Institute at Rice University, the Rice CS Department, and the George R. Brown School of Engineering and Computing, by Rice University and the Department of Computer Science at Rice University, by the Ken Kennedy Institute and Rice Quantum Initiative, which is part of the Smalley-Curl Institute.

\newpage
\appendix
\section{Analysis of Li's beyond-\texorpdfstring{$\kappa$}{k} solver in the standard input model}
\label{append:li}
\input{li.tex}

\section{Analysis of beyond-\texorpdfstring{$\kappa$}{k} solvers in the affine dilation model}
\label{append:affine}
\input{affine.tex}

\section{Generalized truncation property}
\label{append:trunc_general}
\input{trunc_general.tex}


\clearpage
\bibliographystyle{myhamsplain2}
\bibliography{beyond_kappa.bib}

\end{document}

%% file: intro.tex
\subsection{Quantum linear system solvers}
\label{sec:intro_solver}
Solving systems of linear equations constitutes a foundational task across scientific and engineering disciplines. Classical solvers quickly become computationally prohibitive as the number of constraints and unknown variables increases. This intractability led Harrow, Hassidim, and Lloyd to develop the first quantum linear system algorithm in the seminal work~\cite{Harrow2009}, where the BQP-completeness of matrix inversion was established. Remarkably, the query complexity of quantum linear system solvers does not depend explicitly on the problem dimensionality, enabling a potential exponential speedup over classical solvers. It remains challenging to identify end-to-end applications that fully realize this exponential quantum speedup; however, recent studies point to physics-inspired problems such as estimating many-body Green's functions~\cite{2021Yupreconditioned} as promising avenues.

Beyond direct applications, quantum linear system solvers also serve as a key primitive for constructing more advanced quantum algorithms. This was demonstrated in previous work on solving differential equations~\cite{Berry2017Differential,Krovi2023improvedquantum,BerryCosta22}, where a system of linear equations was introduced whose inverse encodes the matrix exponential via a truncated Taylor series. More generally, quantum linear system solvers have found applications in transforming eigenvalues of non-normal operators, achieved through rational generating functions for Faber polynomials~\cite{QEVP}, contour integration of resolvents~\cite{AlaseKaruvade24,Takahira2020QuantumCauchy,takahira2021contour}, and rational approximation~\cite{wang2026sign}.

A critical parameter determining the runtime of quantum linear system solvers is the \emph{spectral condition number}. Specifically, consider the linear system $Ax=\ket{b}$, where $A$ is the \emph{coefficient matrix} with spectral norm $\norm{A}\leq1$ accessed through a \emph{block encoding} oracle $O_A$, and $\ket{b}$ is the \emph{initial vector/state} with unit Euclidean norm $\norm{\ket{b}}=1$ prepared by a unitary $O_b$ from a standard reference state. If $A$ is invertible, then the solution to the linear system is given by
\begin{equation}
    x=A^{-1}\ket{b},
\end{equation}
and we let the condition number be the largest norm of the solution vector maximized over all possible initial states
\begin{equation}
    \kappa=\max_{\norm{\ket{b}}=1}\norm{A^{-1}\ket{b}}=\norm{A^{-1}}.
\end{equation}
\begin{revised}
The condition number is commonly defined as $\norm{A}\norm{A^{-1}}$, which is invariant under rescaling of $A$. Under the normalization $\norm{A}\leq1$, the quantity $\norm{A^{-1}}$ upper bounds $\norm{A}\norm{A^{-1}}$ with equality when $\norm{A}=1$, and we follow previous quantum linear system solvers in referring to $\norm{A^{-1}}$ as the condition number. In practice, the condition number is rarely known exactly; we instead work with an upper bound $\kappa \geq \norm{A^{-1}}$.
\end{revised}
To produce the normalized state $\ket{x}=\frac{x}{\norm{x}}$ with accuracy $\epsilon$, the Harrow--Hassidim--Lloyd solver has query complexity $\operatorname{\mathbf{O}}\left(\poly\left(\kappa,\frac{1}{\epsilon}\right)\right)$. This was improved through a series of works, ultimately yielding solvers with query complexity $\operatorname{\mathbf{O}}\left(\kappa\log\left(\frac{1}{\epsilon}\right)\right)$~\cite{Costa2021linearsystems,Dalzell2024shortcut,cunningham_et_al,Low2026quantumlinearsystem}. This tradeoff is optimal, which is established by an unpublished lower bound of Harrow and Kothari and confirmed recently in~\cite[Appendix A]{Costa2023constant} and \cite{mori2026sparsitydependentcomplexitylowerbound}.

All known solvers attaining the optimal scaling treat the coefficient matrix and the initial vector on equal footing, and thus make the same asymptotic number of queries to $O_A$ and $O_b$:
\begin{equation}
    \operatorname{\mathbf{O}}\left(\kappa\log\left(\frac{1}{\epsilon}\right)\operatorname{\mathbf{Cost}}\left(O_A\right)
    +\kappa\log\left(\frac{1}{\epsilon}\right)\operatorname{\mathbf{Cost}}\left(O_b\right)\right).
\end{equation}
Alternatively, one may solve the quantum linear system problem using \emph{Variable Time Amplitude Amplification} (VTAA)~\cite{Ambainis2012VTAA} incurring a cost of~\cite{Low2026quantumlinearsystem}
\begin{equation}
    \operatorname{\mathbf{O}}\left(\kappa\log\left(\frac{\kappa}{\norm{x}}\right)
    \left(\log\log\left(\frac{\kappa}{\norm{x}}\right)+\log\left(\frac{1}{\epsilon}\right)\right)\operatorname{\mathbf{Cost}}\left(O_A\right)
    +\frac{\kappa}{\norm{x}}\operatorname{\mathbf{Cost}}\left(O_b\right)\right).
\end{equation}
Here, the solution norm satisfies $1\leq\norm{x}\leq\kappa$, so the scaling for the state preparation cost $\frac{\kappa}{\norm{x}}$ can be much better than the worst-case scaling $\kappa\log\left(\frac{1}{\epsilon}\right)$ in terms of the condition number. Moreover, this cost of $O_b$ can be shown to be optimal through a reduction from the quantum search problem. 
\begin{revised}
That said, the number of queries to $O_A$ still scales with the condition number $\kappa$, and it exceeds the optimal scaling $\kappa\log\left(\frac{1}{\epsilon}\right)$ by logarithmic factors. Moreover, even with the deterministic amplification schedule of~\cite{Low2026quantumlinearsystem}, VTAA-based solvers remain considerably more complex to implement than those of~\cite{Costa2021linearsystems,Dalzell2024shortcut,cunningham_et_al,Low2026quantumlinearsystem}.
\end{revised}

It is therefore natural to ask whether the condition number barrier can be lifted, especially with respect to the block encoding cost. This question was raised by Harrow, Hassidim, and Lloyd in~\cite[Section IV]{Harrow2009} and answered affirmatively in the recent work of Li~\cite{li2025new}. Li's original solver handles only sparse linear systems, though it is fairly straightforward to generalize to the block encoding model---we sketch this generalization in~\append{li}. The generalized solver then has the query cost
\begin{equation}
    \operatorname{\mathbf{O}}\left(\frac{\norm{A^{-1\dagger}x}}{\norm{x}\epsilon^2}\operatorname{\mathbf{Cost}}\left(O_A\right)
    +\frac{\norm{A^{-1\dagger}x}}{\norm{x}\epsilon^2}\operatorname{\mathbf{Cost}}\left(O_b\right)\right).
\end{equation}
Achieving the above scaling requires the nontrivial assumption that the solution norm $\norm{x}$ is known up to a constant multiplicative factor. If $\norm{x}$ is unknown, then the complexity worsens to
\begin{equation}
    \operatorname{\mathbf{O}}\left(\frac{\norm{A^{-1\dagger}x}\norm{x}^2}{\epsilon^2}\operatorname{\mathbf{Cost}}\left(O_A\right)
    +\frac{\norm{A^{-1\dagger}x}\norm{x}^2}{\epsilon^2}\operatorname{\mathbf{Cost}}\left(O_b\right)\right).
\end{equation}
Regardless, the cost now depends on vector norms as opposed to the condition number $\kappa$, satisfying the inequality $\norm{x}^2\leq\norm{A^{-1\dagger}x}\leq\kappa\norm{x}$. We shall henceforth call these the \emph{beyond-$\kappa$ solvers}.

Li demonstrated that for applications such as solving multivariate polynomial systems, these vector norm inequalities are far from saturated, and beyond-$\kappa$ solvers remain efficient even in the presence of exponentially large condition numbers. This opens the door to solving ill-conditioned linear systems on quantum computers, a task out of reach for all prior approaches. However, at the opposite extreme where $\norm{A^{-1\dagger}x}\approx\kappa\norm{x}\approx\kappa^2$, the query complexity of Li's solver is $\operatorname{\mathbf{O}}\left(\frac{\kappa}{\epsilon^2}\right)$ with a known $\norm{x}$ and degrades to $\operatorname{\mathbf{O}}\left(\frac{\kappa^4}{\epsilon^2}\right)$ otherwise, falling far short of optimality. Even in the intermediate regime, these unfavorable exponents could offset any practical benefit of the beyond-$\kappa$ solvers, undermining their intended advantage over conventional methods.

\subsection{Main result}
\label{sec:intro_result}
We present two quantum linear system solvers whose complexity goes beyond the condition number barrier: one based on an effective truncation of the linear system, and one based on eigenstate filtering over the unit circle with an effective gap.
We summarize our result and compare it against previous results in~\tab{compare}.

\begin{table}[t]
    \centering
\resizebox{\textwidth}{!}{
\renewcommand{\arraystretch}{1.2}
    \begin{tabular}{cccccc}
\thickhline 
\multirow{1}{*}{Year} & \multirow{1}{*}{Algorithm} & \multirow{1}{*}{Primary innovation} & \multicolumn{1}{c}{Queries to $O_{A}$} & \multirow{1}{*}{Queries to $O_{b}$} \tabularnewline
\thickhline 
2008 & \cite{Harrow2009} & \makecell{Phase estimation\\+ Hamiltonian simulation} & $\operatorname{\mathbf{O}}\left(\frac{\kappa^2}{\norm{x}\epsilon^{2}}\right)$  & $\operatorname{\mathbf{O}}(\frac{\kappa}{\norm{x}})$ \tabularnewline
\hline 
2012 & \cite{Ambainis2012VTAA} & VTAA & $\operatorname{\mathbf{O}}\left(\frac{\kappa}{\epsilon^{3}}\log^{3}\left(\frac{\kappa}{\epsilon}\right)\log^{2}\left(\frac{1}{\epsilon}\right)\right)$ & $\operatorname{\mathbf{O}}\left(\kappa\polylog\left(\frac{\kappa}{\epsilon}\right)\right)$  \tabularnewline
\hline 
2017 & \cite{Childs2015LinearSystems} & Linear combination of quantum walks & $\operatorname{\mathbf{O}}\left(\frac{\kappa^2}{\norm{x}}\operatorname{polylog}\left(\frac{\kappa}{\epsilon}\right)\right)$ & $\operatorname{\mathbf{O}}(\frac{\kappa}{\norm{x}}\log\left(\frac{\kappa}{\epsilon}\right))$ \tabularnewline
\hline
2017 & \cite{Childs2015LinearSystems} & \makecell{Gapped phase estimation\\+ Quantum walk + VTAA} & $\operatorname{\mathbf{O}}\left(\kappa\polylog\left(\frac{\kappa}{\epsilon}\right)\right)$ & $\operatorname{\mathbf{O}}(\kappa\polylog\left(\frac{\kappa}{\epsilon}\right))$  \tabularnewline
\hline 
2018 & \cite{Subasi2019QLSPadiabatic} & Randomized abiabatic evolution & $\operatorname{\mathbf{O}}\left(\frac{\kappa}{\epsilon}\log(\kappa)\right)$ & $\operatorname{\mathbf{O}}\left(\frac{\kappa}{\epsilon}\log(\kappa)\right)$  \tabularnewline
\hline 
2018 & \cite{Chakraborty2018BlockEncoding} & Detailed analysis of VTAA & $\operatorname{\mathbf{O}}\left(\kappa\log(\kappa)\log^{2}\left(\frac{\kappa}{\epsilon}\right)\right)$ & \multicolumn{1}{c}{$\operatorname{\mathbf{O}}\left(\frac{\kappa}{\norm{x}}\log(\kappa)\right)$} \tabularnewline
\hline 
2019 & \cite{An2022QSLP} & Continuous time adiabatic evolution & $\operatorname{\mathbf{O}}\left(\kappa\operatorname{polylog}\left(\frac{\kappa}{\epsilon}\right)\right)$ & $\operatorname{\mathbf{O}}\left(\kappa\operatorname{polylog}\left(\frac{\kappa}{\epsilon}\right)\right)$ \tabularnewline
\hline 
2019 & \cite{Lin2020QLSPfiltering} & Eigenstate filtering & $\operatorname{\mathbf{O}}\left(\kappa\log\left(\frac{\kappa}{\epsilon}\right)\right)$ & $\operatorname{\mathbf{O}}\left(\kappa\log\left(\frac{\kappa}{\epsilon}\right)\right)$ \tabularnewline
\hline 
2022 & \cite{Costa2021linearsystems} & Discrete time adiabatic evolution & $\operatorname{\mathbf{O}}\left(\kappa\log\left(\frac{1}{\epsilon}\right)\right)$ & $\operatorname{\mathbf{O}}\left(\kappa\log\left(\frac{1}{\epsilon}\right)\right)$ \tabularnewline
\hline 
2023 & \cite{Chakraborty2023VTAAQLSP} & QSVT-based GPE + VTAA & $\operatorname{\mathbf{O}}\left(\kappa\log(\kappa)\log\left(\frac{\kappa}{\epsilon}\right)\right)$  & $\operatorname{\mathbf{O}}\left(\frac{\kappa}{\norm{x}}\log(\kappa)\right)$ \tabularnewline
\hline 
2024 & \cite{Dalzell2024shortcut} & Kernel reflection  & $\operatorname{\mathbf{O}}\left(\kappa\log\left(\frac{1}{\epsilon}\right)\right)$ &  $\operatorname{\mathbf{O}}\left(\kappa\log\left(\frac{1}{\epsilon}\right)\right)$ \tabularnewline
\hline
2024 & \cite{cunningham_et_al} & Poisson-distributed phase randomization & $\operatorname{\mathbf{O}}\left(\kappa\log\left(\frac{1}{\epsilon}\right)\right)$ &  $\operatorname{\mathbf{O}}\left(\kappa\log\left(\frac{1}{\epsilon}\right)\right)$ \tabularnewline
\hline
2024 & \cite{Low2026quantumlinearsystem} & \makecell{Tunable VTAA\\+ Discretized inverse state} & $\operatorname{\mathbf{O}}\left(\kappa\log\left(\frac{\kappa}{\norm{x}}\right)\log\left(\frac{\log(\kappa/\norm{x})}{\epsilon}\right)\right)$ & $\operatorname{\mathbf{O}}\left(\frac{\kappa}{\norm{x}}\right)$ \tabularnewline
\hline 
2024 & \cite{Low2026quantumlinearsystem} & Block preconditioning  & $\operatorname{\mathbf{O}}\left(\kappa\log\left(\frac{1}{\epsilon}\right)\right)$ & $\operatorname{\mathbf{O}}\left(\kappa\log\left(\frac{1}{\epsilon}\right)\right)$ \tabularnewline
\hline 
2025 & \cite{li2025new} & Filtering with effective gap  & $\operatorname{\mathbf{O}}\left(\frac{\norm{A^{-1\dagger}x}}{\norm{x}\epsilon^2}\right)$ & $\operatorname{\mathbf{O}}\left(\frac{\norm{A^{-1\dagger}x}}{\norm{x}\epsilon^2}\right)$ \tabularnewline
\thickhline 
2026 & This work & Effective truncation of linear system  & $\operatorname{\mathbf{O}}\left(\keff\log\left(\frac{\keff}{\norm{x}}\right)\log\left(\frac{\log(\keff/\norm{x})}{\epsilon}\right)\right)$ & $\operatorname{\mathbf{O}}\left(\frac{\keff}{\norm{x}}\right)$
\tabularnewline
\hline 
2026 & This work & \begin{revised}Dirichlet filtering with effective gap\end{revised} & \begin{revised}$\operatorname{\mathbf{O}}\left(\frac{\norm{A^{-1\dagger}x}}{\norm{x}\epsilon}\right)$\end{revised} & \begin{revised}$\operatorname{\mathbf{O}}\left(\frac{\norm{A^{-1\dagger}x}}{\norm{x}\epsilon}\right)$\end{revised}
\tabularnewline
\thickhline 
    \end{tabular}
    \renewcommand{\arraystretch}{1}
    }
    \caption{Complexity comparison of the new and previous solvers for the quantum linear system problem.
    Our effective condition number $\keff$ admits a family of upper bounds indexed by a positive integer $t$, including $\keff\leq\frac{\norm{(A^\dagger A)^{-t/2}x}^{1/t}}
{\norm{x}^{1/t}\epsilon^{1/t}}$ for even $t$ and $\keff\leq\frac{\norm{A^{-1\dagger}(A^\dagger A)^{-(t-1)/2}x}^{1/t}}
{\norm{x}^{1/t}\epsilon^{1/t}}$ for odd $t$.
    \begin{revised}Our filtering-based solver achieves a constant prefactor of $3$ in its leading-order query complexity.\end{revised}
    The comparison assumes $\norm{x}$ is known. Without this assumption, the result of~\cite{li2025new} degrades to $\operatorname{\mathbf{O}}\left(\frac{\norm{A^{-1\dagger}x}\norm{x}^2}{\epsilon^2}\right)$, whereas ours remains the same up to logarithmic factors.
    }
    \label{tab:compare}
\end{table}

Our truncation-based solver has query complexity
\begin{equation}
\label{eq:vtaa_trunc_intro}
    \operatorname{\mathbf{O}}\left(\keff\log\left(\frac{\keff}{\norm{x}}\right)
    \left(\log\log\left(\frac{\keff}{\norm{x}}\right)+\log\left(\frac{1}{\epsilon}\right)\right)\operatorname{\mathbf{Cost}}\left(O_A\right)
    +\frac{\keff}{\norm{x}}\operatorname{\mathbf{Cost}}\left(O_b\right)\right).
\end{equation}
Here, the \emph{effective condition number} $1\leq\keff\leq\kappa$ is chosen such that the truncated solvers still produce the normalized solution state with the target accuracy $\epsilon$. We show that the following family of bounds hold for the effective condition number
\begin{equation}
\label{eq:keff_t_intro}
    \keff\leq\left(\frac{\norm{(A)_{\mathbf{sv}}^{-t}x}}
    {\norm{x}\epsilon}\right)^{\frac{1}{t}},\qquad
    0<t<\infty,
\end{equation}
where $(A)_{\mathbf{sv}}^{-t}$ denotes the singular value transformation~\cite{Gilyen2018singular} with scalar function $\sigma\mapsto\sigma^{-t}$. For a positive integer $t$, the singular value transformation is given by the standard matrix inversion as
\begin{equation}
\label{eq:keff_integer_intro}
    \keff\leq\begin{cases}
        \left(\frac{\norm{A^{-1\dagger}(A^\dagger A)^{-(t-1)/2}x}}
{\norm{x}\epsilon}\right)^{1/t},\quad&t\text{ is odd},\\
        \left(\frac{\norm{(A^\dagger A)^{-t/2}x}}
{\norm{x}\epsilon}\right)^{1/t},\quad&t\text{ is even}.
    \end{cases}
\end{equation}
This reduces further to the succinct bound $\keff\leq\frac{\norm{A^{-1\dagger}x}}{\norm{x}\epsilon}$
when $t=1$, but it also gives alternative bounds with logarithmic scaling in the inverse accuracy when $t=\operatorname{\mathbf{\Theta}}\left(\frac{\log\left(\frac{1}{\epsilon}\right)}{\log\log\left(\frac{1}{\epsilon}\right)}\right)$. As $t$ increases, the vector norm $(\norm{(A)_{\mathbf{sv}}^{-t}x}/\norm{x})^{1/t}\nearrow\kappa$ increases monotonically to the condition number while the inverse error $(1/\epsilon)^{1/t}\searrow1$ decreases to unity, recovering the scaling of conventional VTAA.

\begin{revised}
Our filtering-based solver is remarkably simple and has a favorable runtime prefactor. It applies the Dirichlet kernel of a walk operator constructed from the linear system, which we implement by a linear combination of unitaries and which is optimal among all filters of the same degree. In particular, we show that the solver has query complexity
\begin{equation}
\label{eq:filter_intro}
    3\frac{\norm{A^{-1\dagger}x}}{\norm{x}\epsilon}\operatorname{\mathbf{Cost}}\left(O_A\right)
    +3\frac{\norm{A^{-1\dagger}x}}{\norm{x}\epsilon}\operatorname{\mathbf{Cost}}\left(O_b\right)
\end{equation}
to leading order.
\end{revised}
This is achieved assuming that relative approximations of the solution norm $\norm{x}$ are available with an approximation ratio approaching $1$. If $\norm{x}$ is unknown, then we present a comparably simple solution norm estimator achieving a constant relative precision with query cost
\begin{revised}
\begin{equation}
\begin{aligned}
    \operatorname{\mathbf{O}}\bigg(&\frac{\norm{A^{-1\dagger}x}}{\norm{x}}\log\left(\frac{\norm{A^{-1\dagger}x}}{\norm{x}}\right)\log\log\left(\frac{\norm{A^{-1\dagger}x}}{\norm{x}}\right)\operatorname{\mathbf{Cost}}\left(O_A\right)\\
    &+\frac{\norm{A^{-1\dagger}x}}{\norm{x}}\log\left(\frac{\norm{A^{-1\dagger}x}}{\norm{x}}\right)\log\log\left(\frac{\norm{A^{-1\dagger}x}}{\norm{x}}\right)\operatorname{\mathbf{Cost}}\left(O_b\right)\bigg).
\end{aligned}
\end{equation}
\end{revised}
The approximation ratio can then be refined using standard methods such as amplitude estimation.

\begin{revised}
Our beyond-$\kappa$ solvers thus substantially improve upon the prior art. The truncation-based solver is never worse than conventional VTAA, since $\keff\leq\kappa$. The filtering-based solver is simpler and has a smaller constant prefactor, and it outperforms conventional solvers when $\frac{\norm{A^{-1\dagger}x}}{\norm{x}}$ is sufficiently small compared to $\kappa\epsilon$. We compare the regimes in which each solver is preferable in~\sec{discuss}. This constitutes a concrete step toward efficiently solving ill-conditioned linear systems on a quantum computer.
\end{revised}

\subsection{Block encoding and affine dilation}
\label{sec:intro_input}
Quantum linear system solvers can be formulated in an abstract setting where the coefficient matrix $A$ has norm $\norm{A}\leq1$ and is accessed through a block encoding as
\begin{equation}
    A=G_1^\dagger O_AG_0,
\end{equation}
with unitary $O_A$ and isometries $G_0,G_1$,
while the initial state $\ket{b}$ is prepared by a unitary $O_b$ as
\begin{equation}
    \ket{b}=O_b\ket{0}.
\end{equation}
Assuming the implementation cost of $G_0$ and $G_1$ is negligible without loss of generality, 
the complexity of these solvers is then measured in terms of the number of queries to $O_A$ and $O_b$ respectively. This applies to a range of concrete settings, including sparse linear systems, and extends to cases where $A$ and $\ket{b}$ are constructed indirectly via common linear algebraic operations. We term this the \emph{standard input model} for the remainder of the paper.

The beyond-$\kappa$ solver of~\cite{li2025new} operates on an augmented matrix of the form $\begin{bmatrix}
    A & -\ket{b}
\end{bmatrix}$. Using oracles $O_A$ and $O_b$, one can construct a block encoding of this augmented matrix with normalization factor $\sqrt{2}$ and thereby reformulate the solver in the standard input model; see~\append{li} for details. However, as shown in~\cite{li2025new}, the query complexity can be further refined through a direct block encoding of $\begin{bmatrix}
    A & -\ket{b}
\end{bmatrix}$, bypassing the standard input model.

We introduce the \emph{affine dilation model} to facilitate comparison of Li's solver~\cite{li2025new} with other quantum linear system solvers while also exploiting this refined block encoding. Specifically, we consider $\left[\begin{smallmatrix}
        A & -\ket{b}\\
        0 & c
    \end{smallmatrix}\right]$,
where the augmented matrix is recovered by setting the scalar $c=0$. However, if $c\neq0$, then the affine dilated matrix is invertible and
\begin{equation}
    \begin{bmatrix}
        A & -\ket{b}\\
        0 & c
    \end{bmatrix}^{-1}
    =\begin{bmatrix}
        A^{-1} & \frac{1}{c}A^{-1}\ket{b}\\
        0 & \frac{1}{c}
    \end{bmatrix}.
\end{equation}
Consequently, the solution $x = A^{-1}\ket{b}$ can be extracted from the top-right entry of the inverse matrix, while $A$ and $\ket{b}$ are handled jointly in a single block encoding analogous to~\cite{li2025new}. We formally introduce the quantum linear system problem and its input models in~\sec{input}. While we primarily focus on the standard input model for simplicity, we also analyze the query complexity of beyond-$\kappa$ solvers in the affine dilation model in~\append{affine}.

\subsection{Solver based on effective truncation of linear system}
\label{sec:intro_trunc}
We now present our first beyond-$\kappa$ solver based on an effective truncation of the linear system. To describe this truncation, we define the projections $\Pi_{\text{left},\mathcal{S}}=\sum_{\sigma_j\in\mathcal{S}}\ketbra{u_j}{u_j}$ and $\Pi_{\text{right},\mathcal{S}}=\sum_{\sigma_j\in\mathcal{S}}\ketbra{v_j}{v_j}$, for $A=\sum_j\sigma_j\ketbra{u_j}{v_j}$ the singular value decomposition with singular values $\kappa^{-1}\leq\sigma_j\leq1$ and $\mathcal{S}\subseteq\mathbb{R}$ a subset of real numbers. Then the goal of our truncation-based solver is to produce a quantum state proportional to the truncated solution $\Pi_{\text{right},\left[\keff^{-1},1\right]}A^{-1}\ket{b}$. Here, we aim to choose an effective condition number $1\leq\keff=\keff(\epsilon)\leq\kappa$ as the smallest value satisfying
\begin{equation}
    \norm{\frac{\Pi_{\text{right},\left[\keff^{-1},1\right]}A^{-1}\ket{b}}{\norm{\Pi_{\text{right},\left[\keff^{-1},1\right]}A^{-1}\ket{b}}}
    -\frac{A^{-1}\ket{b}}{\norm{A^{-1}\ket{b}}}}
    =\operatorname{\mathbf{O}}(\epsilon),
\end{equation}
so that the solver still outputs the desired solution state up to error $\operatorname{\mathbf{O}}(\epsilon)$.

We show how to fulfill this truncation requirement by characterizing $\keff$ through the generalized inverse (quantile function) of the weighted singular value distribution, where the weights are determined by the squared overlaps of $A^{-1}\ket{b}$ with the singular vector subspaces:
\begin{equation}
    \frac{\norm{\Pi_{\text{right},\left[0,\keff^{-1}\right)}A^{-1}\ket{b}}^2}{\norm{A^{-1}\ket{b}}^2}\leq\epsilon^2,\qquad
    \frac{\norm{\Pi_{\text{right},\left[0,\keff^{-1}\right]}A^{-1}\ket{b}}^2}{\norm{A^{-1}\ket{b}}^2}>\epsilon^2.
\end{equation}
\begin{revised}
Here, the quantile function of a distribution is the generalized inverse of its cumulative distribution function, which maps a probability level to the corresponding threshold value; in our setting, $\keff^{-1}$ is the threshold at which the weight of the normalized solution on smaller singular values reaches $\epsilon^2$.
\end{revised}
This definition is illustrated in~\fig{quantile_keff}.
\begin{revised}
As a simple example, consider $A=\ketbra{0}{0}+\eta\ketbra{1}{1}$ with $0<\eta\ll1$ and $\ket{b}=\sqrt{1-\eta^4}\ket{0}+\eta^2\ket{1}$, so that $\kappa=\frac{1}{\eta}$ and $x=\sqrt{1-\eta^4}\ket{0}+\eta\ket{1}$. The overlap of $\ket{b}$ with the small singular value is nonzero but is amplified by the inversion only to a component of weight $\frac{\eta^2}{\norm{x}^2}\leq\eta^2$ in the normalized solution. Hence $\keff=1$ for all $\eta\leq\epsilon<1$. For $\eta=e^{-n}$, the condition number is exponential in $n$, whereas the effective condition number remains constant at any accuracy $\epsilon\geq e^{-n}$.
\end{revised}
In practice, determining the exact value of the effective condition number $\keff$ can be computationally prohibitive. We address this by developing a family of upper bounds on $\keff$ such as~\eq{keff_integer_intro}, which can be analyzed using standard linear algebraic techniques.
See~\thm{keff} for details.

\begin{figure}[t]
\centering
\includegraphics{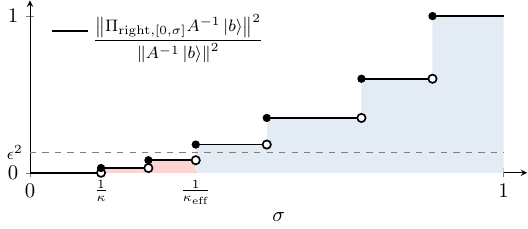}
\caption{Illustration of the definition of effective condition number. The step function represents the cumulative distribution of the squared inverse singular values of $A$ weighted by the initial state $\ket{b}$. The effective condition number $\kappa_{\text{eff}}$ is defined so that the total weight below $1/\kappa_{\text{eff}}$ (red region) is at most $\epsilon^2$.}
\label{fig:quantile_keff}
\end{figure}

Note that although the singular vector projection $\Pi_{\text{right},\left[\keff^{-1},1\right]}$ may be realized to precision $\epsilon$ using Quantum Singular Value Transformation (QSVT)~\cite{Gilyen2018singular} with cost $\operatorname{\mathbf{O}}\left(\keff\log\left(\frac{1}{\epsilon}\right)\right)$, we do \emph{not} pursue this approach, as this would incur query overhead and significantly complicate the algorithm. Instead, we invoke a quantum linear system solver with the condition number set to $\keff$, without modifying the input oracles $O_A$ and $O_b$. Since our $\kappa_{\text{eff}}$ is no longer an upper bound on the condition number, this constitutes a somewhat ``abnormal'' use of the quantum linear system solver, and the resulting behavior is generally unpredictable. To address this, we demand the solver to effectively truncate the coefficient matrix in the following sense.

\begin{definition*}[Strong truncation property]
Given a linear system $Ax=\ket{b}$ with $\norm{A}\leq1$ and $\norm{\ket{b}}=1$, let $O_A$ and $O_b$ be the oracles for block encoding and initial state preparation respectively. A quantum linear system solver is said to satisfy the \emph{strong truncation property} if, when configured with precision $\epsilon$ and condition number $\kappa_{\text{eff}}=\kappa_{\text{eff}}(\epsilon) \geq 1$, it produces the truncated solution state $\Pi_{\text{right},\left[\keff^{-1},1\right]}A^{-1}\ket{b}/\norm{\Pi_{\text{right},\left[\keff^{-1},1\right]}A^{-1}\ket{b}}$ to accuracy 
$\operatorname{\mathbf{O}}(\epsilon)$.
\end{definition*}

Under this definition, any quantum linear system solver with the strong truncation property can produce a truncated solution state, with a query complexity determined by the effective condition number $\kappa_{\text{eff}}$. Recall from~\eq{keff_t_intro} and~\eq{keff_integer_intro} that $\kappa_{\text{eff}}$ can be bounded by vector norms that are independent of the condition number; consequently, we have obtained a beyond-$\kappa$ solver.
We formally state this result in~\thm{beyondk_strong_trunc}.

It remains to identify solvers that fulfill the strong truncation requirement. One simple solver can be constructed by block encoding $A^{-1}$ using QSVT with condition number set to $\keff$, followed by a (fixed-point) amplitude amplification. This QSVT-based solver makes $\operatorname{\mathbf{O}}\left(\frac{\keff^2}{\norm{x}}\log\left(\frac{\keff}{\norm{x}\epsilon}\right)\right)$ queries to $O_A$ and $\operatorname{\mathbf{O}}\left(\frac{\keff}{\norm{x}}\right)$ queries to $O_b$. Running VTAA with $\kappa_{\text{eff}}$ yields an improved query complexity, as shown in~\eq{vtaa_trunc_intro}.
However, rigorously establishing the strong truncation property for VTAA is highly nontrivial. 
The difficulty arises because VTAA recursively amplifies all components of the state that potentially lead to success. When configuring with $\kappa_{\mathrm{eff}}$, we intend to zero out components corresponding to singular values $\sigma_j < \kappa_{\mathrm{eff}}^{-1}$. However, VTAA still treats these components as potentially successful and amplifies them, requiring an additional postprocessing step. In effect, VTAA amplifies first and then truncates, rather than truncating first and then amplifying. This nuance leads to a nonstandard amplification schedule that is difficult to analyze directly.
To circumvent this technical difficulty, we introduce the weak truncation property. Combined with an effective truncation of the initial state, this admits a significantly cleaner proof.

Specifically, with the same effective condition number $\keff=\keff(\epsilon)$ as above, we consider the truncated initial state $\ket{b_{\text{eff}}}$ proportional to $\Pi_{\text{left},\left[\keff^{-1},1\right]}\ket{b}$. 
We then define\begin{revised}, for the purpose of analysis only,\end{revised} an oracle $O_{b_{\text{eff}}}$ that prepares $\ket{b_{\text{eff}}}$ with one query to $O_b$ and approximates $O_b$, i.e.,
\begin{equation}
    O_{b_{\text{eff}}}\ket{0}=\frac{\Pi_{\text{left},\left[\keff^{-1},1\right]}\ket{b}}{\norm{\Pi_{\text{left},\left[\keff^{-1},1\right]}\ket{b}}}=\ket{b_{\text{eff}}},\qquad
    \norm{O_{b_{\text{eff}}}-O_b}=\operatorname{\mathbf{O}}\left(\frac{\norm{x}\epsilon}{\keff}\right).
\end{equation}
Once again, the singular vector projection is \emph{not} performed in the actual implementation of the algorithm. Instead, we run a quantum linear system solver configured with condition number $\kappa_{\text{eff}}$, while retaining the original input oracles $O_A$ and $O_b$.
Suppose that the solver we use has an \emph{optimal query complexity of initial state preparation}, i.e., it makes $\operatorname{\mathbf{O}}\left(\frac{\keff}{\norm{x}}\right)$ queries to $O_b$. Then replacing every occurrence of $O_b$ by $O_{b_{\text{eff}}}$ introduces an error at most $\operatorname{\mathbf{O}}(\epsilon)$. Up to this error tolerance, our solver is conceptually an algorithm with condition number $\keff$ and oracles $O_A$ and $O_{b_{\text{eff}}}$; in this sense the initial state is effectively truncated.

As before, this is a non-standard use of the quantum linear system solver, and its correct execution requires the following additional assumption on the underlying solver.

\begin{definition*}[Weak truncation property]
Given a linear system $Ax=\ket{b}$ with $\norm{A}\leq1$ and $\norm{\ket{b}}=1$, let $O_A$ and $O_b$ be the oracles for block encoding and initial state preparation respectively.
A quantum linear system solver is said to satisfy the \emph{weak truncation property} if, when configured with precision $\epsilon$ and condition number $\kappa_{\text{eff}}=\kappa_{\text{eff}}(\epsilon) \geq 1$, it produces the truncated solution state $A^{-1}\ket{b}/\norm{A^{-1}\ket{b}}$ to accuracy 
$\operatorname{\mathbf{O}}(\epsilon)$,
as long as the initial vector is supported in $\ket{b}\in\operatorname{\mathbf{Im}}\left(\Pi_{\text{left},\left[\keff^{-1},1\right]}\right)$.
\end{definition*}

Note that the strong truncation property naturally implies the weak truncation property. This follows from the fact that $A^{-1}\ket{b} = A^{-1}\Pi_{\text{left},[\kappa_{\text{eff}}^{-1},1]}\ket{b} = \Pi_{\text{right},[\kappa_{\text{eff}}^{-1},1]}A^{-1}\ket{b}$ whenever $\ket{b} \in \operatorname{\mathbf{Im}}\left(\Pi_{\text{left},[\kappa_{\text{eff}}^{-1},1]}\right)$ holds. 
Conversely, we show in~\thm{beyondk_weak_trunc} that a solver satisfying the weak truncation property with optimal query complexity of initial state preparation necessarily fulfills the strong truncation requirement.
Given the additional promise that $\ket{b} \in \operatorname{\mathbf{Im}}(\Pi_{\text{left},[\kappa_{\mathrm{eff}}^{-1},1]})$, 
\begin{revised}VTAA satisfies the weak truncation property, as we verify in~\sec{trunc}.\end{revised}
As VTAA also has an optimal query complexity of initial state preparation, it can be used to construct beyond-$\kappa$ solvers based on effective state truncation (\cor{beyondk_vtaa}). We illustrate this algorithmic template in~\fig{truncation} and present a detailed analysis of the truncation-based solver in~\sec{trunc}.

\begin{figure}[t]
	\centering
\includegraphics[width=0.65\textwidth]{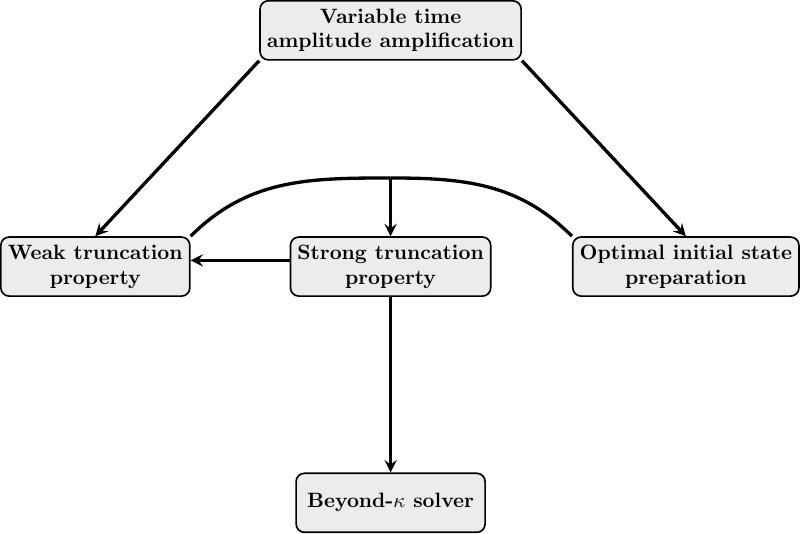}
\caption{Illustration of the algorithmic template for constructing beyond-$\kappa$ solvers based on effective truncation of the linear system.}
\label{fig:truncation}
\end{figure}

To obtain the beyond-$\kappa$ solvers, we require that VTAA produces the truncated solution state when configured with the effective condition number $\keff$. Interestingly, VTAA can be modified to produce the truncated solution state for \emph{any} $\alpha$ supplied as the condition number parameter. This is a more general truncation property of VTAA and is discussed further in~\append{trunc_general}.

\subsection{Solver based on filtering with effective gap}
\label{sec:intro_gap}
Our second beyond-$\kappa$ solver is based on an improved eigenstate filtering on the unit circle with an effective gap. Specifically, suppose that the augmented matrix $\begin{bmatrix}
    \beta A & -\ket{b}
\end{bmatrix}$ is block encoded by the overlap of isometries $G_1^\dagger G_0$ with some scalar $\beta>0$. We then start with the initial state $G_0\begin{bmatrix}
    0\\
    1
\end{bmatrix}$ partitioned conformally to the augmented matrix, which decomposes as
\begin{equation}
\begin{aligned}
    G_0\begin{bmatrix}
        0\\1
    \end{bmatrix}
    &=\frac{\beta}{\norm{x}^2+\beta^2}
    G_0\begin{bmatrix}
        x\\\beta
    \end{bmatrix}
    -\frac{\beta}{\norm{x}^2+\beta^2}G_0\begin{bmatrix}
        x\\
        -\frac{\norm{x}^2}{\beta}
    \end{bmatrix}\\
    &=\frac{\beta}{\norm{x}^2+\beta^2}G_0\begin{bmatrix}
        x\\
        \beta
    \end{bmatrix}
    -\frac{\sqrt{\beta^2+1}}{\norm{x}^2+\beta^2}
    G_0G_0^\dagger G_1A^{-1\dagger}x.
\end{aligned}
\end{equation}
Our goal is to retain the first term while discarding the second, from which a state proportional to the solution $x=A^{-1}\ket{b}$ can be extracted.

A direct calculation shows that the first term belongs to $\operatorname{\mathbf{Ker}}\left(\Pi_1 \Pi_0\right)$ while the second term is supported in $\operatorname{\mathbf{Im}}\left(\Pi_0 \Pi_1\right)$, where $\Pi_0=G_0G_0^\dagger$ and $\Pi_1=G_1G_1^\dagger$. As the smallest nonzero singular value of $\Pi_1\Pi_0$ is at least $\operatorname{\mathbf{\Omega}}(\kappa^{-1})$~\cite[Appendix C]{li2025new}, one can perform a right singular vector filtering targeting at the singular value $0$ with $\Pi_1\Pi_0$ as the underlying operator. This preserves the first term while arbitrarily suppressing the second to magnitude below $\operatorname{\mathbf{O}}(\delta)$, although the query cost would be $\operatorname{\mathbf{O}}\left(\kappa\log\left(\frac{1}{\delta}\right)\right)$ scaling linearly with the condition number.

In contrast, the beyond-$\kappa$ solver uses the operator $W=\left(2\Pi_0-I\right)\left(I-2\Pi_1\right)$.
A key observation is that the first term from the decomposition is an eigenvector of $W$ with eigenvalue $e^{i0}=+1$, while the second term corresponds to eigenvalues $e^{i\theta}$ with nonzero phase angles $\theta\neq0$.
The effective gap lemma~\cite[Lemma 4.2]{Lee_effectivegap} then allows us to reduce the magnitude of the second term below $\operatorname{\mathbf{O}}(\delta)$ with a query cost $\operatorname{\mathbf{O}}\left(\frac{1}{\delta}\right)$ independent of the condition number.
This suggests performing an eigenstate filtering over the unit circle that preserves the $\theta=0$ component while suppressing the remaining ones with the help of effective gap lemma.

This eigenstate filtering can be realized in a variety of ways. 
\begin{revised}
The original beyond-$\kappa$ solver~\cite{li2025new} uses quantum phase estimation on $W$.
\end{revised}
Alternatively, one can use generalized quantum signal processing~\cite{WangZhangYuWang23,MotlaghWiebe24} with the target polynomial chosen as the Dolph-Chebyshev polynomial over the unit circle~\cite{li2025discrimination}. Prior work~\cite[Page 30]{li2025discrimination} also proposed a filtering method based on the conventional quantum signal processing~\cite{LC17,LYC16} and QSVT. However, their method relies on $\sin(\theta)$ and appears incapable of distinguishing between eigenphases $\theta=0$ and $\theta=\pm\pi$, a critical issue left unresolved. 
\begin{revised}
This issue can be fixed by shifting the operator to $\frac{W-I}{2}$. More importantly, all these approaches bound the eigenphases inside and outside a chosen gap separately. With a Dolph-Chebyshev filter, this leads to a query cost scaling as $\frac{1}{\epsilon}\log\left(\frac{1}{\epsilon}\right)$. With phase estimation, it leads to a query cost scaling as $\frac{1}{\epsilon^2}$. We instead use the effective gap lemma to bound the entire unwanted component over the unit circle by a single quantity. For the Dirichlet kernel of $W$, which is optimal among filters of the same degree, this suppresses the unwanted component to $\operatorname{\mathbf{O}}(\delta)$ with $\operatorname{\mathbf{O}}\left(\frac{1}{\delta}\right)$ queries, giving a query cost scaling as $\frac{1}{\epsilon}$ as the required suppression $\delta$ is proportional to $\epsilon$. We then implement it by a linear combination of unitaries. Our analysis may improve the accuracy dependence of other applications where state preparation is performed by eigenstate filtering with the effective gap lemma~\cite{Lee_effectivegap}.
\end{revised}

If $\norm{x}$ is known a priori, we can run the beyond-$\kappa$ solver with $\beta=\operatorname{\mathbf{\Theta}}(\norm{x})$, successfully producing the solution state with a constant probability. This filter-then-measure strategy is exceptionally simple. In fact, we show that its query complexity has a constant prefactor \begin{revised}$3$\end{revised} to leading order; see~\eq{filter_intro}. If $\norm{x}$ is unknown, then we develop a similarly simple solution norm estimation algorithm by applying an exponential search starting from $\beta=1$ and terminating at $\beta=\operatorname{\mathbf{\Theta}}(\norm{x})$ with a high probability. This yields a constant multiplicative approximation of the solution norm, and the approximation ratio can be further refined using amplitude estimation as in~\cite{Dalzell2024shortcut}. We present the filtering-based solver in detail in~\sec{gap} with the result formally stated in~\thm{beyondk_filtering}.

We remark that although both solvers approximate the solution state $|x\rangle$ to an accuracy of $\operatorname{\mathbf{O}}(\epsilon)$, their error profiles are structurally different. The truncation-based solver produces a state approximately supported on the right singular vectors of $A$ with singular values of at least $\kappa_{\mathrm{eff}}^{-1}$, discarding all components corresponding to smaller singular values. The filtering-based solver preserves the full solution but incurs errors from incomplete suppression of an auxiliary component in a constrained orthogonal decomposition induced by the operator $W$. Despite this structural difference, the two solvers provide equivalent 
approximation guarantees: both produce a quantum state within Euclidean 
distance $\operatorname{\mathbf{O}}(\epsilon)$ of the normalized~$|x\rangle$.
Which solver to use then depends on the regime of interest: the truncation-based solver (via VTAA) achieves nearly linear scaling in the effective condition number with an optimal number of queries to the initial state preparation, while the filtering-based solver (with the effective spectral gap) provides a much simpler approach with a more favorable runtime prefactor.

We conclude in \sec{discuss} with a summary of our main contributions and a number of open questions for future investigation.

%% file: input.tex
In this section, we formally introduce the quantum linear system problem and discuss its input models. We first define the standard input model in \sec{input_standard} based on block encoding, and then introduce the affine dilation model in \sec{input_affine} which allows further fine-tuning of the query complexity. Finally, we formulate the quantum linear system problem in \sec{input_problem} and discuss its solvers.

\subsection{Standard input model}
\label{sec:input_standard}
Block encoding provides a natural input model for quantum linear algebra algorithms~\cite{Chakraborty2018BlockEncoding,Low17}.
It can be defined in full generality using unitaries and isometries. We say an operator $G:\mathcal{G}\rightarrow\mathcal{H}$ is an \emph{isometry} between spaces $\mathcal{G}$ and $\mathcal{H}$ if $G^\dagger G=I$. By definition, $G$ is injective with image $\operatorname{\mathbf{Im}}(G)=\operatorname{\mathbf{Im}}(GG^\dagger)$, whereas $G^\dagger$ is surjective with kernel $\operatorname{\mathbf{Ker}}(G^\dagger)=\operatorname{\mathbf{Ker}}(GG^\dagger)$. We thus obtain the isometric embedding $\mathcal{G}
    \xrightleftharpoons[G^\dagger]{G}
    \operatorname{\mathbf{Im}}(GG^\dagger)\subseteq\mathcal{H}$. 
Fixing any bases with respect to this embedding, the operator $G$ can be represented in matrix form as
\begin{equation}
    G=\begin{bmatrix}
        I\\
        0
    \end{bmatrix}.
\end{equation}
Examples of isometries include: (i) quantum state $\ket{\psi}$, where $\norm{\ket{\psi}}=1$; (ii) unitary operator $U$, where $UU^\dagger=U^\dagger U=I$; (iii) tensor product $G_0\otimes G_1$, where $G_0$ and $G_1$ are both isometries; (iv) operator composition $G_1G_0$, where $G_0$ and $G_1$ are both isometries and the composition is well defined; and (v) Stinespring dilation $\sum_j\ket{j}\otimes E_j$, where $\sum_jE_j^\dagger E_j=I$ and $\{\ket{j}\}$ are orthonormal.

We say that $A$ is block encoded by isometries $G_0$, $G_1$ and unitary $O_A$, if
\begin{equation}
    A=G_1^\dagger O_AG_0.
\end{equation}
In bases compatible with the isometric embeddings, we can write $O_A$ in the matrix form
\begin{equation}
    O_A=\begin{bmatrix}
        A & \cdot\\
        \cdot & \cdot
    \end{bmatrix},
\end{equation}
where $A$ appears as the top-left block; hence the name \emph{block encoding}. Note that such a block encoding is feasible only when $\norm{A}\leq\norm{G_1^\dagger}\norm{O_A}\norm{G_0}=1$. Conversely, for any $\norm{A}\leq1$, one can mathematically define a dilation $O_A=\left[\begin{smallmatrix}
        A & -\sqrt{I-AA^\dagger}\\
        \sqrt{I-A^\dagger A} & A^\dagger
    \end{smallmatrix}\right]$ as in~\cite[2.7.P2]{horn2012matrix} and \cite{halmos1950normal}, although $A$ may need to be further rescaled by some normalization factor for $O_A$ to be efficiently realizable as a quantum circuit.

For a given linear system $Ax=\ket{b}$, we assume without loss of generality that the initial vector $\ket{b}$ is a normalized quantum state, prepared by oracle $O_b$ from a standard reference state as
\begin{equation}
    \ket{b}=O_b\ket{0}.
\end{equation}
We further assume that the coefficient matrix $A$ is accessed through a block encoding, with any normalization factor absorbed into the matrix so that $\norm{A} \leq 1$.
Depending on the quantum linear system solver under consideration, we adopt one of two simplified forms of block encoding.
In the first form, we neglect the implementation cost of the isometries $G_0$, $G_1$ and simply say that $A$ is block encoded by the unitary $O_A$. This can happen for a block encoding like $A=\left(\bra{0}\otimes I\right)O_A\left(\ket{0}\otimes I\right)$ where isometries are given by trivial state preparations.
In the second form, we absorb $O_A$ into the definition of either $G_0$ or $G_1$ and say that $A$ is block encoded by the overlap of isometries $G_1^\dagger G_0$. The query complexity is then measured by the number of calls to either $G_0$ or $G_1$.

\subsection{Affine dilation model}
\label{sec:input_affine}
In the standard input model, we assumed that the coefficient matrix $A$ has spectral norm $\norm{A} \leq 1$ and is accessed through a block encoding $A = G_1^\dagger O_A G_0$. However, this often cannot be efficiently realized by a quantum circuit unless $A$ is rescaled by an additional factor $\alpha > 1$, giving $\frac{A}{\alpha} = G_1^\dagger O_A G_0$.
The normalization factor $\alpha>1$ pushes the singular values of $A$ toward $0$, making the linear system harder to solve quantumly; minimizing $\alpha$ is thus critical for reducing the overall query complexity.

To this end, we introduce the affine dilation model, where the input is given by the block encoding of
\begin{equation}
    \widetilde{A}=\frac{1}{\widetilde{\alpha}}
    \begin{bmatrix}
        A & -\ket{b}\\
        0 & c
    \end{bmatrix}
\end{equation}
for some scalar $c$ and normalization factor $\widetilde{\alpha}\geq1$. If $c=0$, then the affine dilation reduces to the augmented matrix $\frac{1}{\alpha}\begin{bmatrix}
    A & -\ket{b}
\end{bmatrix}$ used as the input model by the beyond-$\kappa$ solver of~\cite{li2025new}. However, we can also choose $c>0$. Then the affine dilated matrix becomes invertible and the solution $x = A^{-1}\ket{b}$ can be extracted from the top-right entry of the inverse matrix
\begin{equation}
    \widetilde{A}^{-1}
    =\left(\frac{1}{\widetilde{\alpha}}\begin{bmatrix}
        A & -\ket{b}\\
        0 & c
    \end{bmatrix}\right)^{-1}
    =\widetilde{\alpha}\begin{bmatrix}
        A^{-1} & \frac{1}{c}A^{-1}\ket{b}\\
        0 & \frac{1}{c}
    \end{bmatrix}.
\end{equation}
This means the affine dilation model can be paired with any quantum linear system solver, allowing the query complexity to be optimized over different choices of $A$, $\ket{b}$, $c$, and $\widetilde{\alpha}$.

Let us show that, even with the worst-case estimate, it never incurs significant overhead converting from the standard input model to the affine dilation model. To this end, we set $c=1$ and decompose the affine dilation into
\begin{equation}
    \begin{bmatrix}
        A & -\ket{b}\\
        0 & 1
    \end{bmatrix}
    =\begin{bmatrix}
        A & 0\\
        0 & 1
    \end{bmatrix}
    +\frac{1}{2}\begin{bmatrix}
        0 & -\ket{b}\\
        -\bra{b} & 0
    \end{bmatrix}
    +\frac{1}{2}\begin{bmatrix}
        0 & -\ket{b}\\
        \bra{b} & 0
    \end{bmatrix},
\end{equation}
where each summand can be block encoded respectively as
\begin{equation}
\begin{aligned}
    \begin{bmatrix}
        A & 0\\
        0 & 1
    \end{bmatrix}
    &=\begin{bmatrix}
        G_1^\dagger & 0\\
        0 & \bra{0}
    \end{bmatrix}
    \begin{bmatrix}
        O_A & 0\\
        0 & I
    \end{bmatrix}
    \begin{bmatrix}
        G_0 & 0\\
        0 & \ket{0}
    \end{bmatrix},\\
    \begin{bmatrix}
        0 & -\ket{b}\\
        -\bra{b} & 0
    \end{bmatrix}
    &=\begin{bmatrix}
        I & 0\\
        0 & \bra{0}
    \end{bmatrix}
    \begin{bmatrix}
        0 & -O_b\\
        -O_b^\dagger & 0
    \end{bmatrix}
    \begin{bmatrix}
        I & 0\\
        0 & \ket{0}
    \end{bmatrix},\\
    \begin{bmatrix}
        0 & -\ket{b}\\
        \bra{b} & 0
    \end{bmatrix}
    &=\begin{bmatrix}
        I & 0\\
        0 & \bra{0}
    \end{bmatrix}
    \begin{bmatrix}
        0 & -O_b\\
        O_b^\dagger & 0
    \end{bmatrix}
    \begin{bmatrix}
        I & 0\\
        0 & \ket{0}
    \end{bmatrix}.
\end{aligned}
\end{equation}
Taking their linear combination~\cite[Lemma 52]{Gilyen2018singular} then yields a block encoding of the affine dilation matrix $\widetilde{A}=\frac{1}{2}\left[\begin{smallmatrix}
        A & -\ket{b}\\
        0 & 1
    \end{smallmatrix}\right]$ with normalization factor 
    $\widetilde{\alpha} =  2$.
This implies that the condition number increases to at most
\begin{equation}
    \norm{\widetilde{A}^{-1}}
    =
    \norm{2\begin{bmatrix}
        A^{-1} & A^{-1}\ket{b}\\
        0 & 1
    \end{bmatrix}}
    \leq\norm{2\begin{bmatrix}
        A^{-1} & 0\\
        0 & 1
    \end{bmatrix}}
    +\norm{2\begin{bmatrix}
        0 & A^{-1}\ket{b}\\
        0 & 0
    \end{bmatrix}}
    \leq4\kappa,
\end{equation}
while the 
norm of the solution encoded into the upper-right block rescales to $\norm{\widetilde{\alpha}A^{-1}\ket{b}}=2\norm{x}$.
Since the cost of conventional quantum linear system solvers is governed by the condition number and the solution norm, switching from the standard input model to the affine dilation model thus increases the runtime by at most a constant factor. 
An exception to this conclusion may occur if $\operatorname{\mathbf{Cost}}\left(O_b\right)$ is asymptotically larger than $ \operatorname{\mathbf{Cost}}\left(O_A\right)$. In this case, quantum linear system solvers  that make asymptotically fewer calls to the initial vector than to the coefficient matrix may have asymptotically better cost in the standard input model than in the affine dilation model, since  block-encoding the dilated coefficient matrix $\widetilde{A}$ via linear combination requires querying not only $O_A$ (cheap) but also $O_b$ (expensive). As shown in \tab{compare}, this category of quantum linear system solver includes the VTAA method, but excludes all methods with strictly optimal query complexity to $O_A$. 

Li~\cite{li2025new} showed that for sparse linear systems, one can adjust the row norms in the affine dilation model, substantially reducing the query complexity. In general, searching over $A$, $\ket{b}$, $c$, and $\widetilde{\alpha}$ to minimize the query complexity is a nontrivial task. Accordingly, we primarily state our results in the standard input model, though we also analyze the runtime in the affine dilation model in~\append{affine}. A complete exploration of the affine dilation model---and the potential improvements it may yield---is an interesting direction for future work.

\subsection{Quantum linear system problem}
\label{sec:input_problem}
We now introduce notation and terminology for the quantum linear system problem. Specifically, consider the linear system $Ax = \ket{b}$ which has the solution
\begin{equation}
    x=A^{-1}\ket{b}
\end{equation}
for an invertible square matrix $A$. Then the goal of the quantum linear system problem is to produce the normalized solution vector 
\begin{equation}
    \ket{x}=\frac{x}{\norm{x}}=\frac{A^{-1}\ket{b}}{\norm{A^{-1}\ket{b}}}
\end{equation}
as a quantum state. Since both $A^{-1}$ and $\ket{b}$ appear linearly in the numerator and denominator of the solution state, we assume without loss of generality that the coefficient matrix satisfies $\norm{A} \leq 1$ and the initial vector has unit Euclidean norm $\norm{\ket{b}}=1$. We shall use the Dirac notation only for unit vectors.

We consider the abstract setting where a quantum linear system solver is described by a unitary $U$, acting jointly on the ancilla and system register, with the behavior
\begin{equation}
    \norm{\left(\bra{0}\otimes I\right)U\ket{00}}^2>\frac{1}{2},\qquad
    \norm{\frac{\left(\bra{0}\otimes I\right)U\ket{00}}{\norm{\left(\bra{0}\otimes I\right)U\ket{00}}}-\frac{A^{-1}\ket{b}}{\norm{A^{-1}\ket{b}}}}= 
    \operatorname{\mathbf{O}}(\epsilon). 
    \label{eq:QLSP_formal_specs}
\end{equation}
Here, the first condition guarantees that the solver succeeds with probability strictly greater than $\frac{1}{2}$, while the second ensures that the normalized output state approximates the desired solution state to 
within a fixed constant factor of
the target accuracy \begin{revised}$0<\epsilon<1$\end{revised}.
We focus primarily on the standard input model where $A$ is block encoded by an oracle $O_A$ and $\ket{b}$ is prepared by a unitary $O_b$. The complexity of the solver is then measured in terms of the number of queries to $O_A$ and $O_b$ made by $U$. 
When making non-asymptotic comparisons between two solvers that both produce solutions to $\operatorname{\mathbf{O}}(\epsilon)$ error as in \eq{QLSP_formal_specs}---but with different prefactors in front of the $\epsilon$---care must be taken to rescale $\epsilon$ accordingly. 

As discussed in~\sec{intro_solver}, the spectral condition number is a key parameter determining the runtime of quantum linear system solvers, which is mathematically defined by $\norm{A}\norm{A^{-1}}$. However, the precise value of the condition number is often inaccessible in practice; instead, we work with a known upper bound
\begin{equation}
    \kappa\geq\norm{A^{-1}}
    =\max_{\norm{\ket{b}}=1}\norm{A^{-1}\ket{b}}.
\end{equation}
As this involves a maximization over all initial states $\ket{b}$, the condition number $\kappa$ serves as a worst-case cost measure for quantum linear system solvers. The main goal of this paper is to develop beyond-$\kappa$ solvers whose query complexity does not scale with the condition number.

Another parameter of importance is the solution norm $\norm{x} = \norm{A^{-1}\ket{b}}$. Since $1 \leq \norm{x} \leq \norm{A^{-1}}$, this quantity reduces to the condition number in the worst case, but typically provides a more refined cost measure for quantum linear system solvers.
Like the condition number, the exact solution norm is often unavailable in practice. Instead, we assume knowledge of a value $\alpha_x$ that approximates $\norm{x}$ to within a multiplicative factor $\mu \geq 1$, i.e.,
\begin{equation}
    \frac{1}{\mu} \leq \frac{\alpha_x}{\norm{x}} \leq \mu.
\end{equation}
We will often assume that the solution norm is known to constant multiplicative accuracy, i.e., $\mu = \operatorname{\mathbf{O}}(1)$. When such an estimate is not available a priori, we provide solution norm estimation algorithms to obtain one---this estimation only needs to be performed once as a preprocessing step.  Furthermore, we may assume without loss of generality that $\mu \to 1$, as this refinement can be achieved using amplitude estimation 
incurring multiplicative complexity overhead $1/(\mu-1)$~\cite{Dalzell2024shortcut}.

To summarize, a quantum linear system solver is described by a unitary $U(O_A, O_b, \kappa, \alpha_x, \epsilon)$ depending on the input oracles $O_A$ and $O_b$, an upper bound on the condition number $\kappa$, a constant multiplicative estimate of the solution norm $\alpha_x$, and the target accuracy $\epsilon$. With this notation, the query complexity of the solver takes the form
\begin{equation}
    \operatorname{\mathbf{Cost}}(U(O_A, O_b, \kappa, \alpha_x, \epsilon))
    =h_A(\kappa, \alpha_x, \epsilon)\operatorname{\mathbf{Cost}}(O_A)
    +h_b(\kappa, \alpha_x, \epsilon)\operatorname{\mathbf{Cost}}(O_b),
\end{equation}
where $h_A$ and $h_b$ are positive-integer-valued functions denoting the number of calls to the respective oracles.
To represent them in closed form, we introduce $\operatorname{\mathbf{Floor}}(\cdot)$ to denote the largest integer not exceeding a real number and $\operatorname{\mathbf{Ceil}}(\cdot)$ for the smallest integer no less than a number.

The setting above is somewhat restricted for clarity of exposition, but the framework readily accommodates several generalizations. 
First, the solver need not be a unitary---more generally, it can be described by a quantum channel, with the error measured in trace distance rather than Euclidean distance. This generalization is especially relevant when the solver uses classical randomness, in which case it naturally implements a mixed unitary channel.
Second, we have specified the solver's performance through its success probability and output accuracy. It is also possible to characterize the solver instead by its expected behavior and expected query complexity, assuming that the solver is repeated until a successful run is obtained.
Third, we have focused on the standard input model, where $A$ and $\ket{b}$ are accessed through separate oracles. One can also consider the affine dilation model, where $A$ and $\ket{b}$ are block encoded jointly by a single oracle, allowing further refinements of the query complexity.
Fourth, the solver's performance may depend on additional parameters beyond $\kappa$ and $\norm{x}$. In particular, we will introduce a family of parameters $\norm{(A)_{\mathbf{sv}}^{-t}\ket{x}}$ with $0<t<\infty$ based on the singular value transformation of $A$, which upper bounds the query complexity of beyond-$\kappa$ solvers.
Fifth, we have focused on the case where $A$ is an invertible square matrix, but the framework generalizes to non-square and non-invertible matrices. In particular, when the linear system has no solution---i.e., when $\ket{b}\notin\operatorname{\mathbf{Im}}(A)$---a natural output is the least-squares solution, obtained by applying the Moore--Penrose pseudoinverse of $A$. Our truncation-based solvers can be adapted to achieve this.
Finally, we have focused primarily on query complexity. Other resource measures---such as space complexity, time complexity, and query depth---are also of interest and merit further investigation.

%% file: trunc.tex
In this section, we present an algorithmic template to construct beyond-$\kappa$ solvers based on effective truncation of the linear system. We begin with a formal definition of the effective condition number $\keff$ in~\sec{trunc_keff_def} and present a family of upper bounds on $\keff$ in~\sec{trunc_keff_bound}, proving~\thm{keff}. We then describe an algorithm based on effective truncation of the coefficient matrix and analyze its runtime in~\thm{beyondk_strong_trunc} in~\sec{trunc_matrix}.
Next, we \begin{revised}define an oracle in~\sec{trunc_oracle} (used only in the analysis)\end{revised} to realize effective truncation of the initial state.
Finally, we propose and analyze a solver based on effective state truncation in~\thm{beyondk_weak_trunc} in~\sec{trunc_state}.

\subsection{Effective condition number}
\label{sec:trunc_keff_def}
Let $A=\sum_j\sigma_j\ketbra{u_j}{v_j}$ be the singular value decomposition of the coefficient matrix with singular values $\kappa^{-1}\leq\sigma_j\leq1$. Define the singular vector projections $\Pi_{\text{left},\mathcal{S}}=\sum_{\sigma_j\in\mathcal{S}}\ketbra{u_j}{u_j}$ and $\Pi_{\text{right},\mathcal{S}}=\sum_{\sigma_j\in\mathcal{S}}\ketbra{v_j}{v_j}$ for $\mathcal{S}\subseteq\mathbb{R}$ a subset of real numbers. 
Our aim is to choose an effective condition number $\keff$ such that
\begin{equation}
    \frac{\norm{\Pi_{\text{right},\left[0,\keff^{-1}\right)}A^{-1}\ket{b}}}{\norm{A^{-1}\ket{b}}}\leq\epsilon,\qquad
    \frac{\norm{\Pi_{\text{right},\left[0,\keff^{-1}\right]}A^{-1}\ket{b}}}{\norm{A^{-1}\ket{b}}}>\epsilon.
\end{equation}
Intuitively, we use $\keff^{-1}$ to truncate an $\epsilon$-fraction of weight from the solution vector $x=A^{-1}\ket{b}$, starting from the smallest singular values. In cases where multiple $\kappa_{\mathrm{eff}}$ qualify, we shall select the smallest to minimize the cost of our beyond-$\kappa$ solvers (corresponding to the largest possible $\keff^{-1}$).
This defines the effective condition number as a function $\kappa_{\mathrm{eff}}(A, \ket{b}, \epsilon)$; we suppress the arguments when clear from context.

Although the effective condition number can be defined and analyzed directly, we find it more natural to work with cumulative distribution functions and their generalized inverses. Specifically, for any non-decreasing function $f:\mathbb{R}\rightarrow\mathbb{R}$, we denote its one-sided limits by
\begin{equation}
f(x-) = \lim_{\varepsilon \searrow 0} f(x - \varepsilon), \qquad f(x+) = \lim_{\varepsilon \searrow 0} f(x + \varepsilon).
\end{equation}
Since $f$ is non-decreasing, both limits exist and satisfy $f(x-) \leq f(x) \leq f(x+)$ for all $x$.
We then say that $x$ is a \emph{$y$-quantile} of $f$ if $f(x-) \leq y \leq f(x+)$.
In general, the $y$-quantiles are not uniquely determined when $f$ is non-invertible. 
We will use the largest possible value to define the generalized inverse $f^+$.

\begin{definition}[Generalized inverse]\label{defn:gen-inv}
For any non-decreasing function $f:\mathbb{R}\rightarrow\mathbb{R}$,
its \emph{generalized inverse} is defined by
\[
f^+(y) = \sup\{x \in \mathbb{R} : f(x-) \leq y \leq f(x+)\},
\]
with the convention $\sup \emptyset = -\infty$.
\end{definition}

\begin{figure}[t]
\centering
\begin{subfigure}[b]{0.9\textwidth}
\centering
\includegraphics{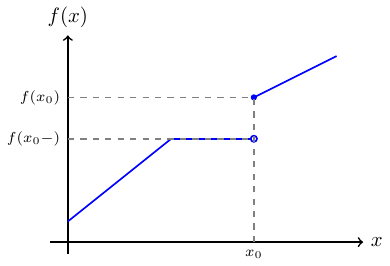}
\caption{A non-decreasing function $f$ with three behaviors: strictly increasing, plateau, and jump (right-continuous at $x_0$).}
\label{fig:F}
\end{subfigure}

\vspace{1em}

\begin{subfigure}[b]{0.48\textwidth}
\centering
\includegraphics{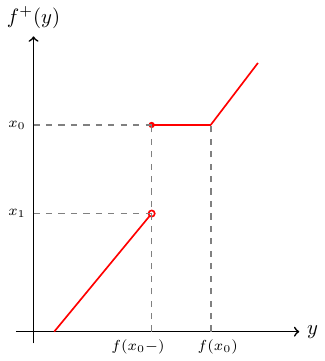}
\caption{The right-continuous generalized inverse $f^+$.}
\label{fig:Fplus}
\end{subfigure}
\hfill
\begin{subfigure}[b]{0.48\textwidth}
\centering
\includegraphics{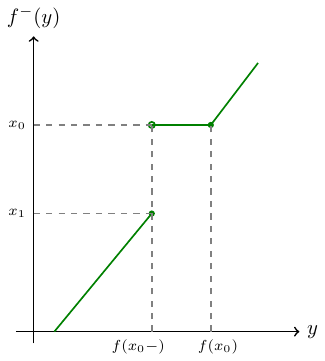}
\caption{The left-continuous generalized inverse $f^-$.}
\label{fig:Fminus}
\end{subfigure}

\caption{A non-decreasing function and its generalized inverses. Plateaus of $f$ become jumps, and jumps of $f$ become plateaus. The only difference between $f^+$ and $f^-$ is the one-sided continuity at discontinuities.}
\label{fig:quantile_all}
\end{figure}

Similarly, one can choose the smallest $y$-quantile and write $f^-(y) = \inf\{x : f(x-) \leq y \leq f(x+)\}$, which is more commonly used in probability. These two definitions are conceptually related though not exactly identical.
For a detailed treatment of this topic, we refer the reader to the excellent note of~\cite{wacker2023please}; see also~\fig{quantile_all}. For our analysis, we need the following properties of the generalized inverse.

\begin{proposition}[Equivalent definitions of $f^+$]\label{prop:equiv-def}
\begin{revised}
Let $y$ satisfy $\inf_{x}f(x)<y<\sup_{x}f(x)$, and let
\end{revised}
$g\colon \mathbb{R} \to \mathbb{R}$ be any function satisfying $f(x-) \leq g(x) \leq f(x+)$ for all $x$. Then the following four quantities are all equal to $f^+(y)$:
\begin{enumerate}
    \item[(a)] $\sup\{x \in \mathbb{R} : f(x) \leq y\}$,
    \item[(b)] $\sup\{x \in \mathbb{R} : g(x) \leq y\}$,
    \item[(c)] $\inf\{x \in \mathbb{R} : f(x) > y\}$,
    \item[(d)] $\inf\{x \in \mathbb{R} : g(x) > y\}$.
\end{enumerate}
In particular, \begin{revised}$f^+(y)$ is finite and\end{revised} the supremum in~\defn{gen-inv} is attained, so $f^+(y)$ is itself the largest $y$-quantile:
\begin{equation}
    f(f^+(y)-) \leq y \leq f(f^+(y)+).
\end{equation}
\end{proposition}
\begin{proof}
We will proceed in steps.

\textbf{Step 1: (a) = (c) and (b) = (d).} For any function $h$, the sets $\{x : h(x) \leq y\}$ and $\{x : h(x) > y\}$ partition $\mathbb{R}$. When $h$ is non-decreasing (or sandwiched between non-decreasing limits), the first set is a left ray and the second is a right ray. Therefore $\sup\{x : h(x) \leq y\} = \inf\{x : h(x) > y\}$. This gives (a) = (c) (taking $h = f$) and (b) = (d) (taking $h = g$). The equivalence (a) = (c) is also discussed in \cite[Remark~1]{wacker2023please}.

\textbf{Step 2: (a) = (b).} Since $f$ itself satisfies $f(x-) \leq f(x) \leq f(x+)$, it is a valid choice of $g$.
Now for any $g$ where $f(x-) \leq g(x) \leq f(x+)$, we have the set inclusions
$
\{x : f(x+) \leq y\} \subseteq \{x : g(x) \leq y\} \subseteq \{x : f(x-) \leq y\}$,
and therefore
$
\sup\{x : f(x+) \leq y\} \leq \sup\{x : g(x) \leq y\} \leq \sup\{x : f(x-) \leq y\}$.
It suffices to show that the two extremes are equal. Let $x \in \{x : f(x-) \leq y\}$. For every $\varepsilon > 0$, we have $f((x - \varepsilon)+) \leq f(x-) \leq y$ (since $f$ is non-decreasing), so $x - \varepsilon \in \{x : f(x+) \leq y\}$. Hence $\sup\{x : f(x+) \leq y\} \geq x - \varepsilon$ for all $\varepsilon > 0$, giving $\sup\{x : f(x+) \leq y\} \geq x$. Since this holds for all $x$ with $f(x-) \leq y$, we conclude
$
\sup\{x : f(x+) \leq y\} \geq \sup\{x : f(x-) \leq y\}$.
Combined with the reverse inequality from the set inclusion, the two extremes are equal, and hence (a) = (b). This also recovers \cite[Lemma~1(n)]{wacker2023please} as a special case.

Let $q_0$ denote the common value of (a)--(d).
\begin{revised}Since $y>\inf_xf(x)$, the set $\{x : f(x)\leq y\}$ is nonempty, and since $y<\sup_xf(x)$, it is bounded above by any $x$ with $f(x)>y$. Hence $q_0$ is finite.\end{revised}

\textbf{Step 3: $f^+(y) \leq q_0$.} The set in \defn{gen-inv} satisfies
$
\{x : f(x-) \leq y \leq f(x+)\} \subseteq \{x : f(x-) \leq y\}$,
since the former imposes two conditions while the latter imposes only one. Taking $g(x) = f(x-)$ in (b), we have $q_0 = \sup\{x : f(x-) \leq y\}$. By the set inclusion, $f^+(y) \leq q_0$.

\textbf{Step 4: $q_0 \leq f^+(y)$.} We show that $f(q_0-) \leq y \leq f(q_0+)$, i.e., $q_0$ is a $y$-quantile, which gives $q_0 \leq \sup\{q : f(q-) \leq y \leq f(q+)\} = f^+(y)$.
\emph{Left inequality:} $f(q_0 -) \leq y$ is \cite[Lemma~1(g).(8)]{wacker2023please}.
\emph{Right inequality:} Suppose for contradiction that $f(q_0 +) < y$. Then there exists $\delta > 0$ such that $f(q_0 + \delta) < y$ (by monotonicity and the definition of the right limit). In particular $f(q_0 + \delta) \leq y$, so by \cite[Lemma~1(g).(5)]{wacker2023please}, we obtain $q_0 \geq q_0 + \delta$, a contradiction. Hence $y \leq f(q_0 +)$.

In conclusion, Steps~3 and~4 give $f^+(y) = q_0$. 
\begin{revised}Since $q_0$ is finite,\end{revised} Step~4 shows $q_0$ is a $y$-quantile that dominates all others (by Step~3), so the supremum in~\defn{gen-inv} is attained.
\end{proof}

The quantile relation $f(f^+(y)-) \leq y \leq f(f^+(y)+)$ cannot be refined in general as both inequalities are saturated when $f$ is an invertible continuous function. However, it can be further strengthened when $f$ is a step function. Specifically, given $n$ discrete points $x_1<x_2<\cdots<x_n$ with associated weights $w_k>0$, define the right-continuous step function $f(x)=\sum_{x_k\leq x}w_k$. Then it holds $f(f^+(y)-) \leq y < f(f^+(y))$.

\begin{corollary}\label{cor:discrete}
Given $n$ real numbers $x_1<x_2<\cdots<x_n$ with associated weights $w_k>0$, define the step function $f(x)=\sum_{x_k\leq x}w_k$. Then,
\begin{equation}
    f(x-) = \sum_{x_k < x} w_k, \qquad f(x+) = f(x) = \sum_{x_k \leq x} w_k.
\end{equation}
Moreover, for all $0<y<\sum_{k=1}^nw_k$, $f^+(y)$ is a $y$-quantile with the right inequality always being strict:
\begin{equation}
    f(f^+(y)-) \leq y < f(f^+(y)).
\end{equation}
\end{corollary}

\begin{proof}
Note that when the sum is finite, the limit commutes with the sum. Specifically,
$
f(x+) = \lim_{\varepsilon \searrow 0} \sum_{x_k \leq x+\varepsilon} w_k = \sum_{x_k \leq x} w_k = f(x)$,
$
f(x-) = \lim_{\varepsilon \searrow 0} \sum_{x_k \leq x-\varepsilon} w_k = \sum_{x_k < x} w_k$,
where both equalities hold because $x_k \leq x \pm \varepsilon$ for all sufficiently small $\varepsilon > 0$ if and only if $x_k \leq x$ (resp.\ $x_k < x$), since the points $\{x_k\}$ are finitely many and isolated.

Since $f$ is right-continuous, \prop{equiv-def} gives $f(f^+(y)-) \leq y \leq f(f^+(y))$. It remains to show that the right inequality is strict.

Let $q = f^+(y) = \sup\{x : f(x) \leq y\}$. Since $f$ is a step function with jumps only at $\{x_k\}$, the point $q=x_j$ must be one of the jump locations. Indeed, if $q$ were not a jump location, then $f$ would be constant in a neighborhood of $q$, and the supremum would extend further to the right, contradicting the definition of $q$.

Now suppose for contradiction that $f(x_j) \leq y$. Since $f$ is a step function, it is constant on $[x_j, x_{j+1})$ for the next jump point $x_{j+1} > x_j$. Then $f(x) \leq y$ for all $x \in [x_j, x_{j+1})$, which gives $\sup\{x : f(x) \leq y\} \geq x_{j+1} > x_j$, contradicting $f^+(y) = x_j$. Therefore $f(f^+(y)) = f(x_j) > y$.
\end{proof}

Suppose that the initial vector expands under the left singular vector basis as $\ket{b}=\sum_jw_j\ket{u_j}$.
We now consider the function
\begin{equation}
    f(\sigma)=\frac{\norm{\Pi_{\text{right},\left[0,\sigma\right]}A^{-1}\ket{b}}^2}{\norm{A^{-1}\ket{b}}^2}
    =\frac{\sum_{\sigma_j\leq \sigma}\sigma_j^{-2}\abs{w_j}^2}{\sum_{k}\sigma_k^{-2}\abs{w_k}^2}.
\end{equation}
This function can be seen as the cumulative distribution function of the discrete weights $\frac{\sigma_j^{-2}\abs{w_j}^2}{\sum_{k}\sigma_k^{-2}\abs{w_k}^2}$. 
Applying the above analysis \begin{revised}with $y=\epsilon^2$ for $0<\epsilon<1$\end{revised}, we obtain
\begin{equation}
\begin{aligned}
    \frac{\sum_{\sigma_j<f^{+}(\epsilon^2)}\sigma_j^{-2}\abs{w_j}^2}{\sum_{k}\sigma_k^{-2}\abs{w_k}^2}
    =f\left(f^{+}(\epsilon^2)-\right)
    \leq\epsilon^2,\qquad
    \frac{\sum_{\sigma_j\leq f^{+}(\epsilon^2)}\sigma_j^{-2}\abs{w_j}^2}{\sum_{k}\sigma_k^{-2}\abs{w_k}^2}
    =f\left(f^{+}(\epsilon^2)\right)
    >\epsilon^2.
\end{aligned}
\end{equation}
Finally, let us define
\begin{equation}
\begin{aligned}
    \keff=\frac{1}{f^{+}(\epsilon^2)}
    =\sup\left\{\sigma^{-1}\ \bigg|\ \frac{\sum_{\sigma_j\leq \sigma}\sigma_j^{-2}\abs{w_j}^2}{\sum_{k}\sigma_k^{-2}\abs{w_k}^2}>\epsilon^2\right\}
    =\inf\left\{\sigma^{-1}\ \bigg|\ \frac{\sum_{\sigma_j\leq \sigma}\sigma_j^{-2}\abs{w_j}^2}{\sum_{k}\sigma_k^{-2}\abs{w_k}^2}\leq\epsilon^2\right\}.
\end{aligned}
\end{equation}
This then gives the desired inequalities
\begin{equation}
\begin{aligned}
    \frac{\sum_{\sigma_j<\keff^{-1}}\sigma_j^{-2}\abs{w_j}^2}{\sum_{k}\sigma_k^{-2}\abs{w_k}^2}
    \leq\epsilon^2,\qquad
    \frac{\sum_{\sigma_j\leq \keff^{-1}}\sigma_j^{-2}\abs{w_j}^2}{\sum_{k}\sigma_k^{-2}\abs{w_k}^2}
    >\epsilon^2.
\end{aligned}
\end{equation}

\subsection{Bounds on the effective condition number}
\label{sec:trunc_keff_bound}
The effective condition number $\kappa_{\mathrm{eff}}$ from the previous subsection serves as a threshold for truncating the linear system, enabling quantum linear system solving beyond the condition number. However, evaluating $\kappa_{\mathrm{eff}}$ directly can be computationally difficult. We now present a family of upper bounds on $\kappa_{\mathrm{eff}}$ that can be used to bound the runtime of truncation-based beyond-$\kappa$ solvers.

Specifically, for any tunable parameter $t>0$, we have 
\begin{equation}
    \sum_j\sigma_j^{-2-2t}\abs{w_j}^2
    \geq\sum_{\sigma_j\leq\keff^{-1}}\sigma_j^{-2-2t}\abs{w_j}^2
    \geq\keff^{2t}\sum_{\sigma_j\leq\keff^{-1}}\sigma_j^{-2}\abs{w_j}^2
    \geq\keff^{2t}\epsilon^2\sum_{k}\sigma_k^{-2}\abs{w_k}^2,
\end{equation}
which gives a family of upper bounds indexed by $t$:
\begin{equation}
    \keff\leq\left(\frac{\sum_j\sigma_j^{-2-2t}\abs{w_j}^2}{\sum_{k}\sigma_k^{-2}\abs{w_k}^2\epsilon^2}\right)^{\frac{1}{2t}}.
\end{equation}
The hierarchy of these bounds can be understood through the following weighted power mean inequality.

\begin{lemma}[Weighted power mean inequality {\cite[Problem 8.3]{steele2004cauchy}~\cite{AoPS}}]
Let $a_j>0$ and $p_j>0$ be positive real numbers such that $\sum_{j}p_j=1$ is $\ell_1$-normalized. For any $0<s<t<\infty$, the weighted power mean
\begin{equation}
    \left(\sum_j p_ja_j^s\right)^{\frac{1}{s}}
    \leq\left(\sum_j p_ja_j^t\right)^{\frac{1}{t}}
\end{equation}
increases monotonically with the exponent.
Moreover, we have the limiting behavior
\begin{equation}
    \lim_{t\rightarrow\infty}\left(\sum_jp_ja_j^t\right)^{\frac{1}{t}}
    =\max_ja_j.
\end{equation}
\end{lemma}

To apply this inequality, we rewrite our bounds on $\keff$ as
\begin{equation}
    \keff
    \leq\frac{\left(\sum_jp_ja_j^{2t}\right)^{\frac{1}{2t}}}
    {\epsilon^{\frac{1}{t}}},\qquad
    a_j=\sigma_j^{-1},\qquad
    p_j=\frac{\sigma_j^{-2}\abs{w_j}^2}{\sum_k\sigma_k^{-2}\abs{w_k}^2}.
\end{equation}
Here, we have $a_j>0$ and we can assume without loss of generality that $p_j>0$ for all $j$ as well, for otherwise we simply drop the zero terms from the summand. 
By the power mean inequality, we see that the numerator monotonically increases to
\begin{equation}
    \left(\sum_jp_ja_j^{2t}\right)^{\frac{1}{2t}}\nearrow\max_j\sigma_j^{-1}.
\end{equation}
This is at most the condition number $\kappa$, with equality when the initial state $\ket{b}$ has nonzero support on the left singular vector of $A$ with the smallest singular value.
On the other hand, the denominator monotonically increases to
\begin{equation}
    \epsilon^{\frac{1}{t}}\nearrow 1.
\end{equation}

Recall that $A=\sum_j\sigma_j\ketbra{u_j}{v_j}$, $
    \ket{b}=\sum_jw_j\ket{u_j}$, and $
    x=A^{-1}\ket{b}=\sum_j\sigma_j^{-1}w_j\ket{v_j}$.
Using the singular value transformation with an odd function extension~\cite{Gilyen2018singular},
\begin{equation}
    \left(\sum_j\sigma_j\ketbra{u_j}{v_j}\right)_{\mathbf{sv}}^{-t}
    =\sum_j\sigma_j^{-t}\ketbra{u_j}{v_j},
\end{equation}
we can then represent the bounds compactly as
\begin{equation}
    \keff\leq\left(\frac{\norm{(A)_{\mathbf{sv}}^{-t}x}}
    {\norm{x}\epsilon}\right)^{\frac{1}{t}}.
\end{equation}
For positive integers $t$, the singular value transformation reduces to the standard matrix inversion as:
\begin{equation}
    \keff\leq\begin{cases}
        \left(\frac{\norm{A^{-1\dagger}(A^\dagger A)^{-(t-1)/2}x}}
{\norm{x}\epsilon}\right)^{1/t},\quad&t\text{ is odd},\\
        \left(\frac{\norm{(A^\dagger A)^{-t/2}x}}
{\norm{x}\epsilon}\right)^{1/t},\quad&t\text{ is even}.
    \end{cases}
\end{equation}
These matrix inversions may then be analyzed using standard linear algebraic techniques.

In particular, when $t=1$, we can bound the effective condition number by
\begin{equation}
    \keff\leq\frac{\norm{A^{-1\dagger}x}}{\norm{x}\epsilon}.
\end{equation}
This yields beyond-$\kappa$ solvers whose runtime is comparable to that of the solver from \sec{gap} based on filtering with effective gap.
However, other choices of $t$ may yield solvers with even lower query complexity. For instance, we can choose
\begin{equation}
    t=\frac{\ln\left(\frac{1}{\epsilon}\right)}{\ln\ln\left(\frac{1}{\epsilon}\right)}
\end{equation}
such that
\begin{equation}
    \left(\frac{1}{\epsilon}\right)^{\frac{1}{t}}
    =\ln\left(\frac{1}{\epsilon}\right).
\end{equation}
The solver thus achieves a logarithmic dependence on the inverse accuracy without increasing the vector norm factor beyond the condition number. We collect our findings on the effective condition number in the following theorem.

\begin{theorem}[Effective condition number]
\label{thm:keff}
Let $A=\sum_j\sigma_j\ketbra{u_j}{v_j}$ be the singular value decomposition of matrix $A$ with singular values $0<\kappa^{-1}\leq\sigma_j\leq1$. Define the singular vector projections $\Pi_{\text{left},\mathcal{S}}=\sum_{\sigma_j\in\mathcal{S}}\ketbra{u_j}{u_j}$ and $\Pi_{\text{right},\mathcal{S}}=\sum_{\sigma_j\in\mathcal{S}}\ketbra{v_j}{v_j}$ for $\mathcal{S}\subseteq\mathbb{R}$ a subset of real numbers. Suppose that $\ket{b}$ is a normalized state expanded under the left singular vector basis as $\ket{b}=\sum_jw_j\ket{u_j}$, $x=A^{-1}\ket{b}=\sum_j\sigma_j^{-1}w_j\ket{v_j}$ and $\ket{x}=\frac{x}{\norm{x}}$. 

For any \begin{revised}$0<\epsilon<1$\end{revised}, the effective condition number
\begin{equation}
    \keff
    =\sup\left\{\sigma^{-1}\ \big|\ \norm{\Pi_{\text{right},\left[0,\sigma\right]}\ket{x}}>\epsilon\right\}
    =\inf\left\{\sigma^{-1}\ \big|\ \norm{\Pi_{\text{right},\left[0,\sigma\right]}\ket{x}}\leq\epsilon\right\}
\end{equation}
satisfies
\begin{equation}
    \norm{\Pi_{\text{right},\left[0,\keff^{-1}\right)}\ket{x}}\leq\epsilon,\qquad
    \norm{\Pi_{\text{right},\left[0,\keff^{-1}\right]}\ket{x}}>\epsilon.
\end{equation}

Moreover, for any $t>0$,
\begin{equation}
    \keff\leq\left(\frac{\norm{(A)_{\mathbf{sv}}^{-t}\ket{x}}}
    {\epsilon}\right)^{\frac{1}{t}},
\end{equation}
where $\left(\sum_j\sigma_j\ketbra{u_j}{v_j}\right)_{\mathbf{sv}}^{-t}
    =\sum_j\sigma_j^{-t}\ketbra{u_j}{v_j}$ is the singular value transformation, which reduces to
\begin{equation}
    \keff\leq\begin{cases}
        \left(\frac{\norm{A^{-1\dagger}(A^\dagger A)^{-(t-1)/2}\ket{x}}}
{\epsilon}\right)^{1/t},\quad&t\text{ is odd},\\
        \left(\frac{\norm{(A^\dagger A)^{-t/2}\ket{x}}}
{\epsilon}\right)^{1/t},\quad&t\text{ is even},
    \end{cases}
\end{equation}
when $t$ is a positive integer.
For $t=1$, this gives
\begin{equation}
    \keff\leq\frac{\norm{A^{-1\dagger}\ket{x}}}{\epsilon},
\end{equation}
while $t=\operatorname{\mathbf{\Theta}}\left(\frac{\log\left(\frac{1}{\epsilon}\right)}{\log\log\left(\frac{1}{\epsilon}\right)}\right)$ yields bounds with logarithmic dependence on the inverse accuracy.
As $t$ increases, the vector norm $\norm{(A)_{\mathbf{sv}}^{-t}\ket{x}}^{1/t}$ increases monotonically below $\kappa$ while the inverse error $(1/\epsilon)^{1/t}$ decreases to unity.
\end{theorem}
\begin{proof}
The claims follow from the analysis preceding the theorem statement applied to the normalized state $\ket{x}=\frac{x}{\norm{x}}$, together with the observations that
\begin{equation}
\begin{aligned}
    \norm{x}^2&=\norm{A^{-1}\ket{b}}^2=\sum_{j}\sigma_j^{-2}\abs{w_j}^2,\\
    \norm{(A)_{\mathbf{sv}}^{-t}x}^2&=\norm{\left(A^\dagger\right)_{\mathbf{sv}}^{-t-1}\ket{b}}^2=\sum_j\sigma_j^{-2-2t}\abs{w_j}^2,\\
    \norm{\Pi_{\text{right},\left[0,\sigma\right]}x}^2&=\norm{\Pi_{\text{right},\left[0,\sigma\right]}A^{-1}\ket{b}}^2=\sum_{\sigma_j\leq \sigma}\sigma_j^{-2}\abs{w_j}^2,\\
    \norm{\Pi_{\text{right},\left[0,\sigma\right)}x}^2&=\norm{\Pi_{\text{right},\left[0,\sigma\right)}A^{-1}\ket{b}}^2=\sum_{\sigma_j< \sigma}\sigma_j^{-2}\abs{w_j}^2.
\end{aligned}
\end{equation}
\end{proof}

\subsection{Analysis of beyond-\texorpdfstring{$\kappa$}{k} solver based on effective matrix truncation}
\label{sec:trunc_matrix}
We are now ready to present and analyze our beyond-$\kappa$ quantum linear system solver based on an effective matrix truncation. The output of the solver will approximate a quantum state proportional to the truncated solution $\Pi_{\text{right},\left[\keff^{-1},1\right]}x$. The following corollary confirms that this truncated solution, after normalization, is close to the true solution state.

\begin{corollary}
With the same notation and assumptions as~\thm{keff},
\begin{equation}
    \norm{\frac{\Pi_{\text{right},\left[\alpha_{\keff}^{-1},1\right]}x}{\norm{\Pi_{\text{right},\left[\alpha_{\keff}^{-1},1\right]}x}}
    -\ket{x}}
    \leq2\epsilon
\end{equation}
for any upper bound on the effective condition number $\alpha_{\keff}\geq\keff$.
\end{corollary}
\begin{proof}
Applying the estimate
\begin{equation}
    \norm{\frac{y}{\norm{y}}-\frac{x}{\norm{x}}}
    \leq\norm{\frac{y}{\norm{y}}-\frac{y}{\norm{x}}}
    +\norm{\frac{y}{\norm{x}}-\frac{x}{\norm{x}}}
    \leq\frac{2\norm{y-x}}{\norm{x}}
\end{equation}
with $y=\Pi_{\text{right},\left[\alpha_{\keff}^{-1},1\right]}x$, we obtain
\begin{equation}
    \norm{\frac{\Pi_{\text{right},\left[\alpha_{\keff}^{-1},1\right]}x}{\norm{\Pi_{\text{right},\left[\alpha_{\keff}^{-1},1\right]}x}}
    -\frac{x}{\norm{x}}}
    \leq\frac{2\norm{\Pi_{\text{right},\left[\alpha_{\keff}^{-1},1\right]}x-x}}{\norm{x}}
    =\frac{2\norm{\Pi_{\text{right},\left[0,\alpha_{\keff}^{-1}\right)}x}}{\norm{x}}.
\end{equation}
Since $\alpha_{\keff}^{-1}\leq\keff^{-1}$, the associated projections are partially ordered as $\Pi_{\text{right},\left[0,\alpha_{\keff}^{-1}\right)}\leq\Pi_{\text{right},\left[0,\keff^{-1}\right)}$, which implies
\begin{equation}
    \norm{\Pi_{\text{right},\left[0,\alpha_{\keff}^{-1}\right)}x}
    \leq\norm{\Pi_{\text{right},\left[0,\keff^{-1}\right)}x}.
\end{equation}
The claimed inequality then follows from~\thm{keff} with the normalization $\ket{x}=\frac{x}{\norm{x}}$.
\end{proof}

As noted in \sec{intro_trunc}, we do \emph{not} perform the singular vector projection, but instead invoke a quantum linear system solver with effective condition number $\kappa_{\mathrm{eff}}$. 
Specifically, we denote the solver by the unitary $U\left(O_A,O_b,\kappa,\epsilon\right)$, where $\kappa$ is an upper bound on the condition number and $\epsilon$ is the target accuracy. The solver then acts on the ancilla and system register jointly and has the behavior
\begin{equation}
\begin{aligned}
    \norm{\left(\bra{0}\otimes I\right)U\left(O_A,O_b,\kappa,\epsilon\right)\ket{00}}^2
    &>\frac{1}{2},\\
    \norm{\frac{\left(\bra{0}\otimes I\right)U\left(O_A,O_b,\kappa,\epsilon\right)\ket{00}}{\norm{\left(\bra{0}\otimes I\right)U\left(O_A,O_b,\kappa,\epsilon\right)\ket{00}}}-\ket{x}}
    &=\operatorname{\mathbf{O}}(\epsilon),
\end{aligned}
\end{equation}
guaranteed to succeed with a constant probability and produce a quantum state close to the normalized solution state.
Since $\kappa_{\mathrm{eff}}$ is not an upper bound on the true condition number, the solver's output is no longer well defined in the usual sense. Instead, we require it to correctly output the projected solution state supported on terms whose singular values exceed $\kappa_{\mathrm{eff}}^{-1}$.

\begin{definition}[Strong truncation property]
\label{defn:strong_trunc}
Let $A$ be a matrix with $\norm{A}\leq1$ block encoded by $O_A$, and $\ket{b}$ be a normalized quantum state prepared by $O_b$.
Assume that $A$ is invertible and denote $x=A^{-1}\ket{b}$ and $\ket{x}=\frac{x}{\norm{x}}$.
Suppose that $U\left(O_A,O_b,\kappa,\epsilon\right)$ is a quantum linear system solver such that with success probability $>\frac{1}{2}$,
\begin{equation}
\begin{aligned}
    \norm{\frac{\left(\bra{0}\otimes I\right)U\left(O_A,O_b,\kappa,\epsilon\right)\ket{00}}{\norm{\left(\bra{0}\otimes I\right)U\left(O_A,O_b,\kappa,\epsilon\right)\ket{00}}}-\ket{x}}
    &=\operatorname{\mathbf{O}}(\epsilon),
\end{aligned}
\end{equation}
provided that $\kappa\geq\norm{A^{-1}}$ and $\epsilon>0$.

We say that the solver satisfies the \emph{strong truncation property} if with success probability $>\frac{1}{2}$,
\begin{equation}
\begin{aligned}
    \norm{\frac{\left(\bra{0}\otimes I\right)U\left(O_A,O_b,\alpha_{\keff},\epsilon\right)\ket{00}}
    {\norm{\left(\bra{0}\otimes I\right)U\left(O_A,O_b,\alpha_{\keff},\epsilon\right)\ket{00}}}
    -\frac{\Pi_{\text{right},\left[\alpha_{\keff}^{-1},1\right]}x}{\norm{\Pi_{\text{right},\left[\alpha_{\keff}^{-1},1\right]}x}}}
    &=\operatorname{\mathbf{O}}(\epsilon),
\end{aligned}
\end{equation}
provided that $\alpha_{\keff}\geq\keff$ with the effective condition number $\keff=\keff(\epsilon)$ defined by~\thm{keff}, where $\Pi_{\text{right},\left[\alpha_{\keff}^{-1},1\right]}$ is the right singular vector projection of $A$ associated with singular values from $\left[\alpha_{\keff}^{-1},1\right]$.
\end{definition}

We can now construct a beyond-$\kappa$ solver by invoking a quantum linear system algorithm with the strong truncation property, using a suitable effective condition number.

\begin{theorem}[Beyond-$\kappa$ solver based on effective matrix truncation]
\label{thm:beyondk_strong_trunc}
Let $A$ be a matrix with $\norm{A}\leq1$ block encoded by $O_A$, and $\ket{b}$ be a normalized quantum state prepared by $O_b$.
Assume that $A$ is invertible and denote $x=A^{-1}\ket{b}$ and $\ket{x}=\frac{x}{\norm{x}}$.

Let $U$ be a quantum linear system solver satisfying the strong truncation property from~\defn{strong_trunc}. For any \begin{revised}$0<\epsilon<1$\end{revised}, define the effective condition number $\keff=\keff(\epsilon)$ according to~\thm{keff}.
Then, with success probability $>\frac{1}{2}$,
\begin{equation}
\begin{aligned}
    \norm{\frac{\left(\bra{0}\otimes I\right)U\left(O_A,O_b,\alpha_{\keff},\epsilon\right)\ket{00}}{\norm{\left(\bra{0}\otimes I\right)U\left(O_A,O_b,\alpha_{\keff},\epsilon\right)\ket{00}}}-\ket{x}}
    &=\operatorname{\mathbf{O}}(\epsilon),
\end{aligned}
\end{equation}
provided that $\alpha_{\keff}\geq\keff$.
\end{theorem}
\begin{proof}
By the strong truncation property, the output of the solver satisfies
\begin{equation}
\begin{aligned}
    \norm{\left(\bra{0}\otimes I\right)
    U\left(O_A,O_b,\alpha_{\keff},\epsilon\right)\ket{00}}^2
    &>\frac{1}{2},\\
    \norm{\frac{\left(\bra{0}\otimes I\right)U\left(O_A,O_b,\alpha_{\keff},\epsilon\right)\ket{00}}
    {\norm{\left(\bra{0}\otimes I\right)U\left(O_A,O_b,\alpha_{\keff},\epsilon\right)\ket{00}}}
    -\frac{\Pi_{\text{right},\left[\alpha_{\keff}^{-1},1\right]}x}{\norm{\Pi_{\text{right},\left[\alpha_{\keff}^{-1},1\right]}x}}}
    &=\operatorname{\mathbf{O}}(\epsilon),
\end{aligned}
\end{equation}
Moreover, the choice of effective condition number guarantees that
\begin{equation}
    \norm{\frac{\Pi_{\text{right},\left[\alpha_{\keff}^{-1},1\right]}x}{\norm{\Pi_{\text{right},\left[\alpha_{\keff}^{-1},1\right]}x}}
    -\ket{x}}
    \leq2\epsilon.
\end{equation}
The claimed error bound then follows from the triangle inequality.
\end{proof}

\thm{beyondk_strong_trunc} applies broadly, converting any conventional quantum linear system solver with the strong truncation property into a beyond-$\kappa$ solver via truncation.
In particular, this applies to VTAA-based solvers, although a direct proof of this claim is somewhat challenging. 
The difficulty is that, when configured with the effective condition number $\kappa_{\mathrm{eff}}$ and the objective of preparing a state proportional to $\Pi_{\text{right}\left[\keff^{-1},1\right]}\ket{x}$, 
the solver aims to
apply matrix inversion on components of the initial vector corresponding to singular values above $\kappa_{\mathrm{eff}}^{-1}$ while suppressing those below $\kappa_{\mathrm{eff}}^{-1}$. 
However, both components are amplified in VTAA, transforming $\ket{b}=\sum_jw_j\ket{u_j}$ into a state proportional to
\begin{equation}
    \frac{1}{\operatorname{\mathbf{\Theta}}(\alpha_{\keff})}\ket{0}\sum_{\sigma_j\geq\alpha_{\keff}^{-1}}\sigma_j^{-1}w_j\ket{v_j}
    +\ket{1}\sum_{\sigma_j<\alpha_{\keff}^{-1}}w_j\ket{v_j},
\end{equation}
where we have omitted all irrelevant ancilla registers used by VTAA, and used singular vectors as opposed to eigenvectors of the walk operator~\cite[Section 4.1]{Low2026quantumlinearsystem} for presentational purpose.
It may be possible to show directly that a sufficiently good approximation of $\ket{x}$ can still be extracted from the output of truncated VTAA; however, this nuance may lead to an unconventional amplification schedule that is cumbersome to analyze.
To circumvent this, we take a different route: we construct a beyond-$\kappa$ solver based on an effective truncation of the initial state, which admits a cleaner, more general analysis that applies to VTAA and implies the above matrix truncation result as a corollary.

\subsection{Oracle for effective state truncation}
\label{sec:trunc_oracle}

We now develop a beyond-$\kappa$ solver based on an effective truncation of the initial state. 
\begin{revised}
In the analysis, we compare the solver with a hypothetical one that replaces every instance of oracle $O_b$ with an oracle $O_{b_{\text{eff}}}$ preparing a truncated initial state. This oracle depends on the singular value decomposition of $A$ and is not implemented by the algorithm; it serves only as a device for the error analysis.
\end{revised} The construction of this truncated state preparation oracle relies on the following rotation within the 2D subspace $\operatorname{\mathbf{Span}}(\ket{b},\ket{b_{\text{eff}}})$ spanned by $\ket{b}$ and $\ket{b_{\text{eff}}}$.

\begin{lemma}[Palais matrix {\cite[2.1.P16]{horn2012matrix}~\cite{Kudo+2024}}]
Let $\ket{b_0}$ and $\ket{b_1}$ be linearly independent unit vectors with $\braket{b_0}{b_1}\in\mathbb{R}$ and define $\ket{w}=\frac{\ket{b_0}+\ket{b_1}}{\norm{\ket{b_0}+\ket{b_1}}}$. Then, the Palais matrix 
\begin{equation}
    V=I-2\ketbra{w}{w}+2\ketbra{b_1}{b_0}
    =\left(I-2\ketbra{w}{w}\right)\left(I-2\ketbra{b_0}{b_0}\right)
\end{equation}
satisfies:
\begin{enumerate}
    \item $V\ket{b_0}=\ket{b_1}$, $V\big\vert_{\operatorname{\mathbf{Span}}^{\bot}\left(\ket{b_0},\ket{b_1}\right)}=I$; and
    \item $VV^\dagger=V^\dagger V=I$, $\operatorname{\mathbf{Det}}(V)=+1$.
\end{enumerate}
Conversely, the above constraints uniquely characterize the matrix $V$.
\end{lemma}
\begin{proof}
We begin with the uniqueness part. Suppose matrices $V$ and $W$ both satisfy the above constraints. Then, $\operatorname{\mathbf{Span}}\left(\ket{b_0},\ket{b_1}\right)$ and $\operatorname{\mathbf{Span}}^{\bot}\left(\ket{b_0},\ket{b_1}\right)$ are invariant under both $W$ and $V$, whereas $W^\dagger V$ satisfies:
\begin{enumerate}
    \item $(W^\dagger V)\ket{b_0}=\ket{b_0}$, $(W^\dagger V)\big\vert_{\operatorname{\mathbf{Span}}^{\bot}\left(\ket{b_0},\ket{b_1}\right)}=I$; and
    \item $(W^\dagger V)(W^\dagger V)^\dagger=(W^\dagger V)^\dagger (W^\dagger V)=I$, $\operatorname{\mathbf{Det}}(W^\dagger V)=+1$.
\end{enumerate}
This means that $W^\dagger V$ preserves $\operatorname{\mathbf{Span}}\left(\ket{b_0},\ket{b_1}\right)$, restricted to which it is unitary, has determinant $+1$, and one eigenvalue $+1$. Its other eigenvalue must therefore also be $+1$. We conclude that $W^\dagger V=I$ is the identity operator, which yields $W=V$.

Since $\ket{b_0}$ and $\ket{b_1}$ are linearly independent, $\ket{b_0}+\ket{b_1}\neq0$, so $\ket{w}$ is well defined. Then the claimed decomposition of the Palais matrix follows from a direct calculation
\begin{equation}
\begin{aligned}
    \left(I-2\ketbra{w}{w}\right)\left(I-2\ketbra{b_0}{b_0}\right)
    &=I-2\ketbra{w}{w}-2\ketbra{b_0}{b_0}+4\braket{w}{b_0}\ketbra{w}{b_0}\\
    &=I-2\ketbra{w}{w}-2\ketbra{b_0}{b_0}
    +4\frac{1+\braket{b_1}{b_0}}{2+\braket{b_0}{b_1}+\braket{b_1}{b_0}}
    \left(\ket{b_0}+\ket{b_1}\right)\bra{b_0}\\
    &=I-2\ketbra{w}{w}+2\ketbra{b_1}{b_0}.
\end{aligned}
\end{equation}
This implies
\begin{equation}
    V\ket{b_0}=-\left(I-2\ketbra{w}{w}\right)\ket{b_0}
    =-\ket{b_0}+2\frac{1+\braket{b_0}{b_1}}{2+\braket{b_0}{b_1}+\braket{b_1}{b_0}}\left(\ket{b_0}+\ket{b_1}\right)
    =\ket{b_1},
\end{equation}
while the remaining claims are straightforward to verify.
\end{proof}

When $\ket{b_0}\approx\ket{b_1}$, we expect that the Palais matrix rotating $\ket{b_0}$ to $\ket{b_1}$ is close to identity. Explicitly,
\begin{equation}
\begin{aligned}
    \norm{V-I}
    &=\norm{-2\ketbra{w}{w}+2\ketbra{b_1}{b_0}}
    =\norm{2\begin{bmatrix}
        -\ket{w} & \ket{b_1}
    \end{bmatrix}
    \begin{bmatrix}
        \bra{w}\\
        \bra{b_0}
    \end{bmatrix}}\\
    &=\norm{
    2\begin{bmatrix}
        \bra{w}\\
        \bra{b_0}
    \end{bmatrix}
    \begin{bmatrix}
        -\ket{w} & \ket{b_1}
    \end{bmatrix}}
    =\norm{2\begin{bmatrix}
        -1 & \braket{w}{b_1}\\
        -\braket{b_0}{w} & \braket{b_0}{b_1}
    \end{bmatrix}}\\
    &=\abs{\braket{b_0}{b_1}-1\pm i\sqrt{1-\braket{b_0}{b_1}^2}}
    =\sqrt{2-2\braket{b_0}{b_1}},
\end{aligned}
\end{equation}
where the third equality uses the fact that $V-I$ is unitarily diagonalizable with rank $2$~\cite[Example 1.3.23]{horn2012matrix}, and the fifth equality holds as $2\left[\begin{smallmatrix}
        -1 & \braket{w}{b_1}\\
        -\braket{b_0}{w} & \braket{b_0}{b_1}
    \end{smallmatrix}\right]$ has eigenvalues $\braket{b_0}{b_1}-1\pm i\sqrt{1-\braket{b_0}{b_1}^2}$.
When applied to the oracles for state preparation, this gives the following distance bound.

\begin{corollary}
With the same notation and assumptions as~\thm{keff}, define
\begin{equation}
    \ket{b_{\text{eff}}}
    =\frac{\Pi_{\text{left},\left[\alpha_{\keff}^{-1},1\right]}\ket{b}}{\norm{\Pi_{\text{left},\left[\alpha_{\keff}^{-1},1\right]}\ket{b}}}
\end{equation}
and let $O_{b_{\text{eff}}}=VO_b$, where $O_b$ is the unitary preparing $\ket{b}$ and $V$ is the unique unitary of determinant $+1$ 
that rotates $\ket{b}$ to $\ket{b_{\text{eff}}}$ and acts trivially on $\operatorname{\mathbf{Span}}^\bot\left(\ket{b},\ket{b_{\text{eff}}}\right)$. Then,
\begin{equation}
    O_{b_{\text{eff}}}\ket{0}=\ket{b_{\text{eff}}},\qquad
    \norm{O_{b_{\text{eff}}}-O_b}
    \leq\sqrt{2-2\sqrt{1-\frac{\epsilon^2\norm{x}^2}{\alpha_{\keff}^2}}}
    =\operatorname{\mathbf{O}}\left(\frac{\epsilon\norm{x}}{\alpha_{\keff}}\right)
\end{equation}
for any $\alpha_{\keff}\geq\keff$.
\end{corollary}
\begin{proof}
When $\ket{b_{\text{eff}}}=\ket{b}$, we adopt the convention that $O_{b_{\text{eff}}}=O_b$ so the claim is trivial. Otherwise, we use~\thm{keff} to bound the distance as
\begin{equation}
\begin{aligned}
    \norm{O_{b_{\text{eff}}}-O_b}
    &=\sqrt{2-2\frac{\bra{b}\Pi_{\text{left},\left[\alpha_{\keff}^{-1},1\right]}\ket{b}}{\norm{\Pi_{\text{left},\left[\alpha_{\keff}^{-1},1\right]}\ket{b}}}}
    =\sqrt{2-2\norm{\Pi_{\text{left},\left[\alpha_{\keff}^{-1},1\right]}\ket{b}}}\\
    &=\sqrt{2-2\sqrt{1-\sum_{\sigma_j<\alpha_{\keff}^{-1}}\abs{w_j}^2}}
    \leq\sqrt{2-2\sqrt{1-\frac{1}{\alpha_{\keff}^2}\sum_{\sigma_j<\alpha_{\keff}^{-1}}\sigma_j^{-2}\abs{w_j}^2}}\\
    &\leq \sqrt{2-2\sqrt{1-\frac{1}{\alpha_{\keff}^2}\sum_{\sigma_j<\keff^{-1}}\sigma_j^{-2}\abs{w_j}^2}} = \sqrt{2-2\sqrt{1-\frac{\norm{x}^2}{\alpha_{\keff}^2}\norm{\Pi_{\text{right},[0,\keff^{-1})}\ket{x}}^2}}\\
    &\leq\sqrt{2-2\sqrt{1-\frac{\epsilon^2\norm{x}^2}{\alpha_{\keff}^2}}}
    =\operatorname{\mathbf{O}}\left(\frac{\epsilon\norm{x}}{\alpha_{\keff}}\right).
\end{aligned}
\end{equation}
where \thm{keff} was used for the inequality in the final line.
\end{proof}

\subsection{Analysis of beyond-\texorpdfstring{$\kappa$}{k} solver based on effective state truncation}
\label{sec:trunc_state}
We are now ready to present and analyze our beyond-$\kappa$ quantum linear system solver based on an effective state truncation.
As mentioned in~\sec{intro_trunc}, we do~\emph{not} perform the state truncation in the actual implementation of the algorithm. Instead, we run a quantum linear system solver configured with condition number $\keff$, while retaining the original input oracles $O_A$ and $O_b$ as well as the constant multiplicative approximation of the solution norm $\norm{x}$. Suppose that the solver has an optimal query complexity of initial state preparation, i.e., it makes $\operatorname{\mathbf{O}}\left(\frac{\keff}{\norm{x}}\right)$ queries to $O_b$. Then replacing every occurrence of $O_b$ by $O_{b_{\text{eff}}}$ introduces an error at most $\operatorname{\mathbf{O}}(\epsilon)$.
Up to this error, our solver is thus effectively configured with condition number $\kappa_{\mathrm{eff}}$ and oracles $O_A$ and $O_{b_{\text{eff}}}$.

As $\keff$ is not an upper bound on the true condition number, this is an unconventional application of quantum linear system solvers. We require that the solver outputs the correct solution state, given the promise that the initial vector has been truncated a priori.

\begin{definition}[Weak truncation property with optimal state preparation]
\label{defn:weak_trunc}
Let $A$ be a matrix with $\norm{A}\leq1$ block encoded by $O_A$, and $\ket{b}$ be a normalized quantum state prepared by $O_b$.
Assume that $A$ is invertible and denote $x=A^{-1}\ket{b}$ and $\ket{x}=\frac{x}{\norm{x}}$.
Suppose that $U\left(O_A,O_b,\kappa,\alpha_x,\epsilon\right)$ is a quantum linear system solver such that with success probability $>\frac{1}{2}$,
\begin{equation}
\begin{aligned}
    \norm{\frac{\left(\bra{0}\otimes I\right)U\left(O_A,O_b,\kappa,\alpha_x,\epsilon\right)\ket{00}}{\norm{\left(\bra{0}\otimes I\right)U\left(O_A,O_b,\kappa,\alpha_x,\epsilon\right)\ket{00}}}-\ket{x}}
    &=\operatorname{\mathbf{O}}(\epsilon),
\end{aligned}
\end{equation}
provided that $\kappa\geq\norm{A^{-1}}$, $\alpha_x$ is a constant multiplicative approximation of $\norm{x}$ and $\epsilon>0$.

We say that the solver satisfies the \emph{weak truncation property} if with success probability $>\frac{1}{2}$,
\begin{equation}
\begin{aligned}
    \norm{\frac{\left(\bra{0}\otimes I\right)U\left(O_A,O_b,\alpha_{\keff},\alpha_x,\epsilon\right)\ket{00}}
    {\norm{\left(\bra{0}\otimes I\right)U\left(O_A,O_b,\alpha_{\keff},\alpha_x,\epsilon\right)\ket{00}}}
    -\ket{x}}
    &=\operatorname{\mathbf{O}}(\epsilon),
\end{aligned}
\end{equation}
provided that $\alpha_{\keff}\geq\keff$ with the effective condition number $\keff=\keff(\epsilon)$ defined by~\thm{keff}, and $\ket{b}\in\operatorname{\mathbf{Im}}\left(\Pi_{\text{left},\left[\alpha_{\keff}^{-1},1\right]}\right)$ where $\Pi_{\text{left},\left[\alpha_{\keff}^{-1},1\right]}$ is the left singular vector projection of $A$ associated with singular values from $\left[\alpha_{\keff}^{-1},1\right]$.

Additionally, we say that the solver has an optimal query complexity of initial state preparation if $U\left(O_A,O_b,\alpha_{\keff},\alpha_x,\epsilon\right)$ makes $\operatorname{\mathbf{O}}\left(\frac{\alpha_{\keff}}{\norm{x}}\right)$ queries to $O_b$.
\end{definition}

It follows directly from the definition that the strong truncation property implies the weak truncation property. This is because, when given the promise $\ket{b}\in\operatorname{\mathbf{Im}}\left(\Pi_{\text{left},\left[\alpha_{\keff}^{-1},1\right]}\right)$, we can say that
\begin{equation}
    \frac{\Pi_{\text{right},\left[\alpha_{\keff}^{-1},1\right]}x}{\norm{\Pi_{\text{right},\left[\alpha_{\keff}^{-1},1\right]}x}}
    =\frac{\Pi_{\text{right},\left[\alpha_{\keff}^{-1},1\right]}A^{-1}\ket{b}}{\norm{\Pi_{\text{right},\left[\alpha_{\keff}^{-1},1\right]}A^{-1}\ket{b}}}
    =\frac{A^{-1}\Pi_{\text{left},\left[\alpha_{\keff}^{-1},1\right]}\ket{b}}{\norm{A^{-1}\Pi_{\text{left},\left[\alpha_{\keff}^{-1},1\right]}\ket{b}}}
    =\frac{A^{-1}\ket{b}}{\norm{A^{-1}\ket{b}}}
    =\frac{x}{\norm{x}}
    =\ket{x}.
\end{equation}
Conversely, if the solver has optimal initial state preparation complexity, then the weak truncation property implies the strong truncation property.
This gives rise to the following algorithmic template for constructing beyond-$\kappa$ solvers.

\begin{theorem}[Beyond-$\kappa$ solver based on effective state truncation]
\label{thm:beyondk_weak_trunc}
Let $A$ be a matrix with $\norm{A}\leq1$ block encoded by $O_A$, and $\ket{b}$ be a normalized quantum state prepared by $O_b$.
Assume that $A$ is invertible and denote $x=A^{-1}\ket{b}$ and $\ket{x}=\frac{x}{\norm{x}}$.

Let $U$ be a quantum linear system solver satisfying the weak truncation property from~\defn{weak_trunc} with an optimal query complexity of initial state preparation. For any \begin{revised}$0<\epsilon<1$\end{revised}, define the effective condition number $\keff=\keff(\epsilon)$ according to~\thm{keff}.
Then, with success probability $>\frac{1}{2}$,
\begin{equation}
\begin{aligned}
    \norm{\frac{\left(\bra{0}\otimes I\right)U\left(O_A,O_b,\alpha_{\keff},\alpha_x,\epsilon\right)\ket{00}}
    {\norm{\left(\bra{0}\otimes I\right)U\left(O_A,O_b,\alpha_{\keff},\alpha_x,\epsilon\right)\ket{00}}}
    -\ket{x}}
    &=\operatorname{\mathbf{O}}(\epsilon),
\end{aligned}
\end{equation}
provided that $\alpha_{\keff}\geq\keff$, $\alpha_x$ is a constant multiplicative approximation of $\norm{x}$, and $\epsilon>0$ is sufficiently small.
\end{theorem}
\begin{proof}
Due to the optimality of initial state preparation, the solver $U\left(O_A,O_b,\alpha_{\keff},\alpha_x,\epsilon\right)$ makes $\operatorname{\mathbf{O}}\left(\frac{\alpha_{\keff}}{\norm{x}}\right)$ queries to $O_b$. Our analysis from the previous subsection then implies that 
the error incurred by replacing $O_b$ with $O_{b_{\text{eff}}}$ is at most $\operatorname{\mathbf{O}}(\epsilon)$, that is
\begin{equation}
    \norm{U\left(O_A,O_{b_{\text{eff}}},\alpha_{\keff},\alpha_x,\epsilon\right)
    -U\left(O_A,O_b,\alpha_{\keff},\alpha_x,\epsilon\right)}
    =\operatorname{\mathbf{O}}\left(\frac{\alpha_{\keff}}{\norm{x}}\right)
    \operatorname{\mathbf{O}}\left(\frac{\epsilon\norm{x}}{\alpha_{\keff}}\right)
    =\operatorname{\mathbf{O}}(\epsilon).
\end{equation}

Note that the effective initial state is supported in $O_{b_{\text{eff}}}\ket{0}
    =\ket{b_{\text{eff}}}\in\operatorname{\mathbf{Im}}\left(\Pi_{\text{left},\left[\alpha_{\keff}^{-1},1\right]}\right)$, and the corresponding solution norm satisfies 
    $\sqrt{1-\epsilon^2}\norm{x}\leq\norm{A^{-1}\ket{b_{\text{eff}}}}
    =\norm{\Pi_{\text{right},\left[\alpha_{\keff}^{-1},1\right]}x}
    \leq\norm{x}$.
Hence, $\alpha_x$ is still a constant multiplicative approximation of  
$\lVert A^{-1}\ket{b_{\text{eff}}}\rVert$ 
and by the weak truncation property,
\begin{equation}
\begin{aligned}
    \norm{\left(\bra{0}\otimes I\right)
    U\left(O_A,O_{b_{\text{eff}}},\alpha_{\keff},\alpha_x,\epsilon\right)\ket{00}}^2
    &>\frac{1}{2},\\
    \norm{\frac{\left(\bra{0}\otimes I\right)U\left(O_A,O_{b_{\text{eff}}},\alpha_{\keff},\alpha_x,\epsilon\right)\ket{00}}
    {\norm{\left(\bra{0}\otimes I\right)U\left(O_A,O_{b_{\text{eff}}},\alpha_{\keff},\alpha_x,\epsilon\right)\ket{00}}}
    -\frac{A^{-1}\ket{b_{\text{eff}}}}{\norm{A^{-1}\ket{b_{\text{eff}}}}}}
    &=\operatorname{\mathbf{O}}(\epsilon).
\end{aligned}
\end{equation}
Moreover, the choice of effective condition number guarantees that
\begin{equation}
    \norm{\frac{A^{-1}\ket{b_{\text{eff}}}}{\norm{A^{-1}\ket{b_{\text{eff}}}}}
    -\ket{x}}
    =\norm{\frac{\Pi_{\text{right},\left[\alpha_{\keff}^{-1},1\right]}x}{\norm{\Pi_{\text{right},\left[\alpha_{\keff}^{-1},1\right]}x}}
    -\frac{x}{\norm{x}}}
    \leq2\epsilon.
\end{equation}

Let us choose $\epsilon>0$ sufficiently small so that
\begin{equation}
    \norm{\left(\bra{0}\otimes I\right)
    U\left(O_A,O_{b},\alpha_{\keff},\alpha_x,\epsilon\right)\ket{00}}^2
    >\frac{1}{2},
\end{equation}
fulfilling the requirement on the success probability.
Meanwhile,
\begin{equation}
    \norm{\frac{\left(\bra{0}\otimes I\right)U\left(O_A,O_{b_{\text{eff}}},\alpha_{\keff},\alpha_x,\epsilon\right)\ket{00}}
    {\norm{\left(\bra{0}\otimes I\right)U\left(O_A,O_{b_{\text{eff}}},\alpha_{\keff},\alpha_x,\epsilon\right)\ket{00}}}
    -\frac{\left(\bra{0}\otimes I\right)U\left(O_A,O_{b},\alpha_{\keff},\alpha_x,\epsilon\right)\ket{00}}
    {\norm{\left(\bra{0}\otimes I\right)U\left(O_A,O_{b},\alpha_{\keff},\alpha_x,\epsilon\right)\ket{00}}}}
    =\operatorname{\mathbf{O}}(\epsilon).
\end{equation}
The claimed error bound now follows from the triangle inequality.
\end{proof}

Like \thm{beyondk_strong_trunc}, \thm{beyondk_weak_trunc} is broadly applicable: it converts any solver with the weak truncation property and optimal query cost for state preparation into a beyond-$\kappa$ solver. 
However, the weak truncation property has an additional promise on the support of the initial state.
\begin{revised}
Thanks to this promise, the weak truncation property of VTAA follows from the analysis of~\cite{Low2026quantumlinearsystem}. The algorithm of~\cite[Section 4.1]{Low2026quantumlinearsystem} operates on the Hermitian dilation $\ketbra{0}{1}\otimes A+\ketbra{1}{0}\otimes A^\dagger$ with initial state $\ket{0}\ket{b}$, whose eigenvectors $\frac{\ket{0}\ket{u_j}\pm\ket{1}\ket{v_j}}{\sqrt{2}}$ correspond to eigenvalues $\pm\sigma_j$. Each step of the algorithm is a function of the associated quantum walk operator, a reflection about the initial state, or a unitary acting only on the clock, flag, and branch-marking registers, followed by a block-encoded matrix inversion controlled by the clock register. By the qubitization result~\cite[Proposition 18]{Low2026quantumlinearsystem}, the subspaces associated with each $\sigma_j$ are invariant under the walk operator, so the direct sum of these subspaces over $\sigma_j\geq\alpha_{\keff}^{-1}$, tensored with the ancilla registers, is invariant under every step of the algorithm. Given $\ket{b}\in\operatorname{\mathbf{Im}}\left(\Pi_{\text{left},\left[\alpha_{\keff}^{-1},1\right]}\right)$, the initial state lies in this subspace, on which all eigenvalues of the dilation have magnitude at least $\alpha_{\keff}^{-1}$. The analysis of~\cite{Low2026quantumlinearsystem} requires the condition number only as a lower bound on these eigenvalues, through the thresholds of gapped phase estimation and the range of the matrix inversion. Therefore, the variable-time stages, the error propagation, the success probability, and the query counts to $O_A$ and $O_b$ in~\cite[Theorem 2]{Low2026quantumlinearsystem} hold with the condition number replaced by $\alpha_{\keff}$.
\end{revised}
We now obtain:

\begin{corollary}[Beyond-$\kappa$ solver based on variable time amplitude amplification]
\label{cor:beyondk_vtaa}
Let $A$ be a matrix with $\norm{A}\leq1$ block encoded by $O_A$, and $\ket{b}$ be a normalized quantum state prepared by $O_b$.
Assume that $A$ is invertible and denote $x=A^{-1}\ket{b}$.

Then, the normalized solution state $\ket{x}=\frac{x}{\norm{x}}$ can be prepared to accuracy $\epsilon>0$ and success probability $>\frac{1}{2}$ with query complexity
\begin{equation}
    \operatorname{\mathbf{O}}\left(\alpha_{\keff}\log\left(\frac{\alpha_{\keff}}{\alpha_x}\right)
    \left(\log\log\left(\frac{\alpha_{\keff}}{\alpha_x}\right)+\log\left(\frac{1}{\epsilon}\right)\right)\operatorname{\mathbf{Cost}}\left(O_A\right)
    +\frac{\alpha_{\keff}}{\alpha_x}\operatorname{\mathbf{Cost}}\left(O_b\right)\right),
\end{equation}
where $\alpha_{\keff}\geq\keff$ is an upper bound on the effective condition number $\keff=\keff(\epsilon)$ defined by~\thm{keff}, and $\alpha_x$ is a constant multiplicative approximation of the solution norm, i.e., $\frac{1}{\mu} \leq \frac{\alpha_x}{\norm{x}} \leq \mu$ for some constant $\mu \geq 1$.

Otherwise, the solution norm $\norm{x}$ can be estimated to a constant multiplicative accuracy with query complexity
\begin{equation}
    \operatorname{\mathbf{O}}\left(\alpha_{\keff}\log^2\left(\frac{\alpha_{\keff}}{\alpha_x}\right)
    \log\log^2\left(\frac{\alpha_{\keff}}{\alpha_x}\right)\operatorname{\mathbf{Cost}}\left(O_A\right)
    +\frac{\alpha_{\keff}}{\alpha_x}\operatorname{\mathbf{Cost}}\left(O_b\right)\right),
\end{equation}
where $\alpha_x\leq\norm{x}$ is a lower bound on the solution norm, and $\alpha_{\keff}\geq\keff(\operatorname{\pmb{\Theta}}(1))$ is an upper bound on the effective condition number.
\end{corollary}
We re-emphasize that the complexity of preparing the solution state $\ket{x}$ has an implicit $\epsilon$ dependence through $\keff$. On the other hand, for estimating the solution norm $\norm{x}$ to constant multiplicative error, one may take $\epsilon = \operatorname{\pmb{\Theta}}(1)$ when computing $\keff$, and there is no implicit $\epsilon$ dependence. 
Additionally, our estimator has a cost that scales with any known lower bound $\alpha_x$ on the solution norm. This may be improved to the exact value of $\norm{x}$ using techniques from~\cite[Theorem 3]{brassard2002quantum} and~\cite[Section 3.3]{SimpleSearch23}. However, we leave the detailed analysis for future work.

\begin{revised}
The weak truncation property also clarifies to what extent the condition number barrier stems from the analysis rather than the algorithm. If $\ket{b}$ is exactly orthogonal to the left singular vectors of $A$ with singular values below $\alpha^{-1}$, then a solver with the weak truncation property, such as VTAA, already produces the correct solution state when configured with $\alpha$ in place of $\kappa$; in this case, the worst-case bound in terms of $\kappa$ simply overestimates the cost of an unmodified algorithm. In general, however, $\ket{b}$ has small but nonzero overlap with this subspace, which is amplified by the matrix inversion. Determining when such components can be discarded without exceeding the target accuracy is precisely the role of the effective condition number in~\thm{keff}, and~\thm{beyondk_weak_trunc} turns any solver with the weak truncation property into a beyond-$\kappa$ solver.
\end{revised}

%% file: gap.tex
In this section, we present a beyond-$\kappa$ solver based on eigenstate filtering with effective gap. We begin by discussing the input model in~\sec{gap_decomposition} and a constrained orthogonal decomposition induced by the given linear system. We then introduce the effective gap lemma in~\sec{gap_effective}. Next, we present \begin{revised}an eigenstate filtering based on the Dirichlet kernel in~\sec{gap_filtering}\end{revised}. In~\sec{gap_solver}, we propose the filtering-based beyond-$\kappa$ solver and analyze the constant prefactor of its query complexity. Finally, we develop a solution norm estimation algorithm in~\sec{gap_est} with a query complexity independent of the condition number, establishing~\thm{beyondk_filtering}.

\subsection{Block encoding and constrained orthogonal decomposition}
\label{sec:gap_decomposition}
Let $Ax=\ket{b}$ be a system of linear equations. 
Recall that the filtering-based solver operates on an augmented matrix of the form $\begin{bmatrix}
    A & -\ket{b}
\end{bmatrix}$.

Starting from the standard input model---where $O_A$ block encodes $A$ with $\norm{A} \leq 1$ and $O_b$ prepares $\ket{b}$---one can set $c = 0$ in the analysis of \sec{input_affine} to obtain a block encoding of $\frac{1}{2}\begin{bmatrix} A & -\ket{b} \end{bmatrix}$ with normalization factor $2$. However, the normalization factor can be improved to $\sqrt{2}$ through the following construction
\begin{equation}
    \frac{1}{\sqrt{2}}\begin{bmatrix}
        A & -\ket{b}
    \end{bmatrix}
    =\frac{1}{\sqrt{2}}\begin{bmatrix}
        \bra{0}\otimes I & -\bra{0}\otimes I
    \end{bmatrix}
    \begin{bmatrix}
        O_A & 0\\
        0 & I\otimes O_b
    \end{bmatrix}
    \begin{bmatrix}
        \ket{0}\otimes I & 0\\
        0 & \ket{0}\otimes\ket{0}
    \end{bmatrix}.
\end{equation}
We then denote this block encoding in terms of an overlap of isometries as
\begin{equation}
    \frac{1}{\sqrt{2}}\begin{bmatrix}
        A & -\ket{b}
    \end{bmatrix}
    =G_1^\dagger G_0,
\end{equation}
where
\begin{equation}
    G_1^\dagger=\frac{1}{\sqrt{2}}\begin{bmatrix}
        \left(\bra{0}\otimes I\right)O_A & -\bra{0}\otimes O_b
    \end{bmatrix},\qquad
    G_0=\begin{bmatrix}
        \ket{0}\otimes I & 0\\
        0 & \ket{0}\otimes \ket{0}
    \end{bmatrix}
\end{equation}
are isometries satisfying $G_0^\dagger G_0=\left[\begin{smallmatrix}
    I & 0\\
    0 & 1
\end{smallmatrix}\right]$ and $G_1^\dagger G_1=I$.
Within this block encoding of the augmented matrix, the isometry $G_1$ makes $1$ controlled query to $O_A$ and $O_b$ in the standard input model, while $G_0$ requires no query to implement.

There are two projections $\Pi_0=G_0G_0^\dagger$ and $\Pi_1=G_1G_1^\dagger$ associated naturally with the block encoding. As $\Pi_0$ and $\Pi_1$ do not commute in general, these operators are not simultaneously diagonalizable and the orthogonal decomposition $\operatorname{\mathbf{Im}}(\Pi_0) = \left(\operatorname{\mathbf{Im}}(\Pi_0) \cap \operatorname{\mathbf{Ker}}(\Pi_1)\right) \obot \left(\operatorname{\mathbf{Im}}(\Pi_0) \cap \operatorname{\mathbf{Im}}(\Pi_1)\right)$ need not hold. Nevertheless, we can utilize the following pair of decompositions.

\begin{lemma}[Constrained orthogonal decomposition]
Any orthogonal projections $\Pi_0$ and $\Pi_1$ on the same space $\mathcal{H}$ induce the following pair of orthogonal decompositions
\begin{equation}
\begin{aligned}
    \mathcal{H}
    &=\operatorname{\mathbf{Ker}}\left(\Pi_1\Pi_0\right)
    \obot\operatorname{\mathbf{Im}}\left(\Pi_0\Pi_1\right),\\
    \operatorname{\mathbf{Im}}\left(\Pi_0\right)
    &=\left(\operatorname{\mathbf{Im}}(\Pi_0)\cap\operatorname{\mathbf{Ker}}(\Pi_1)\right)
    \obot\left(\operatorname{\mathbf{Im}}(\Pi_0)\cap\operatorname{\mathbf{Ker}}(\Pi_1)\right)_{\operatorname{\mathbf{Im}}(\Pi_0)}^{\bot},
\end{aligned}
\end{equation}
where $\obot$ denotes the orthogonal direct sum and $(\cdot)_{\operatorname{\mathbf{Im}}(\Pi_0)}^{\bot}=\operatorname{\mathbf{Im}}(\Pi_0)\cap(\cdot)^{\bot}$ is the orthogonal complement within $\operatorname{\mathbf{Im}}(\Pi_0)$. Moreover, we have the space containment
\begin{equation}
\begin{aligned}
    \operatorname{\mathbf{Im}}(\Pi_0)\cap\operatorname{\mathbf{Ker}}(\Pi_1)
    &\subseteq\operatorname{\mathbf{Ker}}\left(\Pi_1\Pi_0\right),\\
    \operatorname{\mathbf{Im}}\left(\Pi_0\Pi_1\right)
    &\subseteq\left(\operatorname{\mathbf{Im}}(\Pi_0)\cap\operatorname{\mathbf{Ker}}(\Pi_1)\right)_{\operatorname{\mathbf{Im}}(\Pi_0)}^{\bot}.
\end{aligned}
\end{equation}
\end{lemma}
\begin{proof}
The second decomposition is a standard orthogonal decomposition restricted to the subspace $\operatorname{\mathbf{Im}}(\Pi_0)$, whereas the first decomposition follows from
\begin{equation}
    \mathcal{H}=\operatorname{\mathbf{Ker}}\left(\Pi_1\Pi_0\right)\obot\left(\operatorname{\mathbf{Ker}}\left(\Pi_1\Pi_0\right)\right)^\bot
    =\operatorname{\mathbf{Ker}}\left(\Pi_1\Pi_0\right)\obot\operatorname{\mathbf{Im}}\left(\Pi_1\Pi_0\right)^\dagger
    =\operatorname{\mathbf{Ker}}\left(\Pi_1\Pi_0\right)
    \obot\operatorname{\mathbf{Im}}\left(\Pi_0\Pi_1\right).
\end{equation}
The first containment can be verified directly as
\begin{equation}
    \ket{\psi}\in\operatorname{\mathbf{Im}}(\Pi_0)\cap\operatorname{\mathbf{Ker}}(\Pi_1)
    \quad \Rightarrow\quad 
    \Pi_1\Pi_0\ket{\psi}=\Pi_1\ket{\psi}=0
    \quad \Rightarrow\quad 
    \ket{\psi}\in\operatorname{\mathbf{Ker}}\left(\Pi_1\Pi_0\right),
\end{equation}
whereas the second containment follows from
\begin{equation}
    \operatorname{\mathbf{Im}}\left(\Pi_0\Pi_1\right)
    =\left(\operatorname{\mathbf{Ker}}\left(\Pi_1\Pi_0\right)\right)^\bot
    \subseteq\left(\operatorname{\mathbf{Im}}(\Pi_0)\cap\operatorname{\mathbf{Ker}}(\Pi_1)\right)^{\bot}
\end{equation}
together with the trivial fact that $\operatorname{\mathbf{Im}}\left(\Pi_0\Pi_1\right)\subseteq\operatorname{\mathbf{Im}}\left(\Pi_0\right)$.
\end{proof}

Taking the smaller subspace from each containment then yields the following constrained orthogonal decomposition
\begin{equation}
    \left(\operatorname{\mathbf{Im}}(\Pi_0)\cap\operatorname{\mathbf{Ker}}(\Pi_1)\right)
    \obot\operatorname{\mathbf{Im}}\left(\Pi_0\Pi_1\right).
\end{equation}
Here, a generic vector from $\operatorname{\mathbf{Im}}(\Pi_0)$ has the form $G_0\begin{bmatrix}
    y\\
    a
\end{bmatrix}$ for some $y$ and $a$, partitioned conformally to the augmented matrix. For it to also lie in $\operatorname{\mathbf{Ker}}(\Pi_1)$, we require that $G_1^\dagger G_0\begin{bmatrix}
    y\\
    a
\end{bmatrix}=\frac{1}{\sqrt{2}}\begin{bmatrix}
        A & -\ket{b}
    \end{bmatrix}\begin{bmatrix}
    y\\
    a
\end{bmatrix}=0$, which is satisfied by $G_0\begin{bmatrix}
    x\\
    1
\end{bmatrix}$.

On the other hand, a generic vector from $\operatorname{\mathbf{Im}}\left(\Pi_0\Pi_1\right)$ has the form $G_0G_0^\dagger G_1z
=G_0\frac{1}{\sqrt{2}}\begin{bmatrix}
    A^\dagger z\\
    -\bra{b}z
\end{bmatrix}$ for some vector $z$. 
Suppose that we initialize the algorithm with $G_0\begin{bmatrix}
    0\\
    1
\end{bmatrix}$. To find its constrained orthogonal decomposition, we then choose $A^\dagger z=x$ to match the first component which yields $z = A^{-1\dagger} x$. This analysis can be slightly generalized to give:

\begin{corollary}
\label{cor:gap_block}
Let $A$ be a matrix with $\norm{A}\leq1$ block encoded by $O_A$, and $\ket{b}$ be a normalized quantum state prepared by $O_b$.
Assume that $A$ is invertible and denote $x=A^{-1}\ket{b}$.

For any $\beta>0$, the augmented matrix $\frac{1}{\sqrt{\beta^2+1}}\begin{bmatrix}
        \beta A & -\ket{b}
    \end{bmatrix}
    =G_1^\dagger G_0$ can be block encoded as an overlap of isometries with
\begin{equation}
\label{eq:def_g}
    G_1^\dagger=\frac{1}{\sqrt{\beta^2+1}}\begin{bmatrix}
        \beta\left(\bra{0}\otimes I\right)O_A & -\bra{0}\otimes O_b
    \end{bmatrix},\qquad
    G_0=\begin{bmatrix}
        \ket{0}\otimes I & 0\\
        0 & \ket{0}\otimes \ket{0}
    \end{bmatrix}
\end{equation}
    using $1$ controlled query to $O_A$ and $O_b$. Moreover,
\begin{equation}
\begin{aligned}
    G_0\begin{bmatrix}
        0\\
        1
    \end{bmatrix}
    &=\frac{\beta}{\sqrt{\norm{x}^2+\beta^2}}
    \frac{1}{\sqrt{\norm{x}^2+\beta^2}}
    G_0\begin{bmatrix}
        x\\\beta
    \end{bmatrix}
    -\frac{\norm{x}}{\sqrt{\norm{x}^2+\beta^2}}\frac{\beta}{\norm{x}\sqrt{\norm{x}^2+\beta^2}}G_0\begin{bmatrix}
        x\\
        -\frac{\norm{x}^2}{\beta}
    \end{bmatrix}\\
    &=\frac{\beta}{\sqrt{\norm{x}^2+\beta^2}}
    \frac{1}{\sqrt{\norm{x}^2+\beta^2}}G_0\begin{bmatrix}
        x\\
        \beta
    \end{bmatrix}
    -\frac{\norm{x}}{\sqrt{\norm{x}^2+\beta^2}}
    \Pi_0 \Pi_1
    \frac{\sqrt{\beta^2+1}}{\norm{x}\sqrt{\norm{x}^2+\beta^2}}G_1A^{-1\dagger}x\\
    &\in\left(\operatorname{\mathbf{Im}}(\Pi_0)\cap\operatorname{\mathbf{Ker}}(\Pi_1)\right)
    \obot\operatorname{\mathbf{Im}}\left(\Pi_0 \Pi_1\right),
\end{aligned}
\end{equation}
where $\Pi_0=G_0G_0^\dagger$ and $\Pi_1=G_1G_1^\dagger$, and $\obot$ denotes the orthogonal direct sum.
\end{corollary}

\subsection{Effective gap lemma}
\label{sec:gap_effective}
In the above decomposition, the first term encodes the solution $x=A^{-1}\ket{b}$ and is contained in $\operatorname{\mathbf{Ker}}\left(G_1G_1^\dagger G_0G_0^\dagger\right)$, whereas the second term is enclosed by $\operatorname{\mathbf{Im}}\left(G_0G_0^\dagger G_1G_1^\dagger\right)$. Using the fact that the smallest nonzero singular value of $G_1^\dagger G_0$ is gapped from zero by $\operatorname{\mathbf{\Omega}}\left(\frac{1}{\kappa}\right)$, one can apply eigenstate filtering with the underlying operator $G_1^\dagger G_0$ and extract the first term, although doing so incurs a query cost scaling linearly in the condition number $\kappa$~\cite[Appendix C]{li2025new}.

In contrast, we now present the following adaption of the effective gap lemma, which can be applied with the constrained orthogonal decomposition to give a filtering with query cost independent of the condition number.

\begin{lemma}[Effective gap lemma {~\cite[Lemma 4.2]{Lee_effectivegap}}]
\label{lem:eff_gap}
Let $\Pi_0$ and $\Pi_1$ be orthogonal projections on the same space, and define the unitary $W=\left(2\Pi_0-I\right)\left(I-2\Pi_1\right)$.
For any phase function $f:\{e^{i\theta}\ |\ \theta\in[-\pi,\pi]\}\rightarrow[-1,1]$, we have:

\begin{enumerate}
    \item for any $y$, $\norm{f(W)y}\leq\max_{\theta\in[-\pi,\pi]}\abs{f\big(e^{i\theta}\big)}\norm{y}$;
    \item for any $y\in\operatorname{\mathbf{Im}}(\Pi_0)\cap\operatorname{\mathbf{Ker}}(\Pi_1)$, $f(W)y=f(1)y$; and
    \item for any $y=\Pi_0\Pi_1 z\in\operatorname{\mathbf{Im}}\left(\Pi_0\Pi_1\right)$, $\norm{f(W)y}
    \leq\max_{\theta\in[-\pi,\pi]}\abs{f\big(e^{i\theta}\big)
    \begin{revised}
    \sin\left(\frac{\theta}{2}\right)
    \end{revised}
    }\norm{\Pi_1 z}$.
\end{enumerate}
\end{lemma}
\begin{proof}
The first two claims follow directly from the definition.
We present a proof of the third claim using Hermitian qubitization. Specifically, we define isometries $G_0$ and $G_1$ such that $G_0G_0^\dagger=\Pi_0$ and $G_1G_1^\dagger=\Pi_1$. Then,
\begin{enumerate}
    \item $G_0^\dagger\left(I-2G_1G_1^\dagger\right)G_0$ is Hermitian; and
    \item $G_0^\dagger\left(I-2G_1G_1^\dagger\right)^2G_0=I$.
\end{enumerate}
This satisfies all the assumptions of Hermitian qubitization~\cite[Appendix C]{Low2026quantumlinearsystem}.

Suppose that
\begin{equation}
    G_0^\dagger\left(I-2G_1G_1^\dagger\right)G_0
    =\bigobot_{\lambda_u=\pm 1}\begin{bmatrix}
        \lambda_u
    \end{bmatrix}_{\ket{\phi_u}}
    \bigobot_{-1<\lambda_u<1}\begin{bmatrix}
        \lambda_u
    \end{bmatrix}_{\ket{\phi_u}}.
\end{equation}
Here, $\ket{\phi_u}$ are eigenvectors of $G_0^\dagger\left(I-2G_1G_1^\dagger\right)G_0$ which form an orthonormal basis of the underlying space, and $\begin{bmatrix}
    \cdot
\end{bmatrix}_{\ket{\phi_u}}$ denotes the matrix representation under the basis vector $\ket{\phi_u}$.
Then, we have the matrix representations:
\begin{equation}
\begin{aligned}
    I-2G_1G_1^\dagger&=\bigobot_{\lambda_u=\pm 1}\begin{bmatrix}
        \lambda_u
    \end{bmatrix}_{\ket{\phi_{u,0}}}
    \bigobot_{-1<\lambda_u<1}\begin{bmatrix}
        \lambda_u & \sqrt{1-\lambda_u^2}\\
        \sqrt{1-\lambda_u^2} & -\lambda_u
    \end{bmatrix}_{\ket{\phi_{u,0}},\ket{\phi_{u,1}}}
    \bigobot_{\xi_v=\pm1}\begin{bmatrix}
        \xi_v
    \end{bmatrix}_{\ket{\phi_{v,0}}},\\
    G_0G_0^\dagger&=\bigobot_{\lambda_u=\pm 1}\begin{bmatrix}
        1
    \end{bmatrix}_{\ket{\phi_{u,0}}}
    \bigobot_{-1<\lambda_u<1}\begin{bmatrix}
        1 & 0\\
        0 & 0
    \end{bmatrix}_{\ket{\phi_{u,0}},\ket{\phi_{u,1}}}
    \bigobot_{\xi_v=\pm1}\begin{bmatrix}
        0
    \end{bmatrix}_{\ket{\phi_{v,0}}},\\
    W&=\bigobot_{\lambda_u=\pm 1}\begin{bmatrix}
        \lambda_u
    \end{bmatrix}_{\ket{\phi_{u,0}}}
    \bigobot_{-1<\lambda_u<1}\begin{bmatrix}
        \lambda_u & \sqrt{1-\lambda_u^2}\\
        -\sqrt{1-\lambda_u^2} & \lambda_u
    \end{bmatrix}_{\ket{\phi_{u,0}},\ket{\phi_{u,1}}}
    \bigobot_{\xi_v=\pm1}\begin{bmatrix}
        -\xi_v
    \end{bmatrix}_{\ket{\phi_{v,0}}}.
\end{aligned}
\end{equation}
Here, $     \ket{\phi_{u,0}}=G_0\ket{\phi_u}$ and $
    \ket{\phi_{u,1}}=\frac{\left(I-2G_1G_1^\dagger\right)G_0\ket{\phi_u}-\lambda_uG_0\ket{\phi_u}}{\sqrt{1-\lambda_u^2}}$
provide an orthonormal basis for the 1D/2D subspace, and $\ket{\phi_{v,0}}$ is an orthonormal basis for the (uninteresting) orthogonal complement. Within a 2D subspace, $I-2G_1G_1^\dagger$ has eigenpairs
\begin{equation}
    +1:\ket{\phi_{u,ZX+}}=\frac{1}{\sqrt{2}}\begin{bmatrix}
        \sqrt{1+\lambda_u}\\
        \sqrt{1-\lambda_u}
    \end{bmatrix}_{\ket{\phi_{u,0}},\ket{\phi_{u,1}}},\quad
    -1:\ket{\phi_{u,ZX-}}=\frac{1}{\sqrt{2}}\begin{bmatrix}
        \sqrt{1-\lambda_u}\\
        -\sqrt{1+\lambda_u}
    \end{bmatrix}_{\ket{\phi_{u,0}},\ket{\phi_{u,1}}},
\end{equation}
whereas $W=\left(2G_0G_0^\dagger-I\right)\left(I-2G_1G_1^\dagger\right)$ has eigenpairs
\begin{equation}
    e^{i\arccos(\lambda_u)}:\ket{\phi_{u,Y+}}=\frac{1}{\sqrt{2}}\begin{bmatrix}
        1\\
        i
    \end{bmatrix}_{\ket{\phi_{u,0}},\ket{\phi_{u,1}}},\quad
    e^{-i\arccos(\lambda_u)}:\ket{\phi_{u,Y-}}=\frac{1}{\sqrt{2}}\begin{bmatrix}
        1\\
        -i
    \end{bmatrix}_{\ket{\phi_{u,0}},\ket{\phi_{u,1}}}.
\end{equation}

Now suppose we start with an arbitrary $+1$ eigenvector of $G_1G_1^\dagger$ (or equivalently, $-1$ eigenvector of $I-2G_1G_1^\dagger$):
\begin{equation}
    w=\sum_{\lambda_u=-1}\alpha_u\ket{\phi_{u,0}}
    +\sum_{-1<\lambda_u<1}\alpha_u\frac{\sqrt{1-\lambda_u}\ket{\phi_{u,0}}-\sqrt{1+\lambda_u}\ket{\phi_{u,1}}}{\sqrt{2}}
    +\sum_{\xi_v=-1}\alpha_v\ket{\phi_{v,0}}.
\end{equation}
Then, applying $G_0G_0^\dagger$ mathematically transforms the above vector into
\begin{equation}
    G_0G_0^\dagger w
    =\sum_{\lambda_u=-1}\alpha_u\ket{\phi_{u,0}}
    +\sum_{-1<\lambda_u<1}\alpha_u\frac{\sqrt{1-\lambda_u}\ket{\phi_{u,0}}}{\sqrt{2}}.
\end{equation}
Finally, we apply $f(W)$ onto each eigensubspace with eigenvalue $e^{\pm i\arccos(\lambda_u)}$, producing
\begin{equation}
    f(W)G_0G_0^\dagger w
    =\sum_{-1\leq\lambda_u<1}\alpha_u\frac{\sqrt{1-\lambda_u}}{\sqrt{2}}
    \left(\frac{f\left(e^{i\arccos(\lambda_u)}\right)\ket{\phi_{u,Y+}}+
    f\left(e^{-i\arccos(\lambda_u)}\right)\ket{\phi_{u,Y-}}}{\sqrt{2}}\right).
\end{equation}
It is clear that the resulting (unnormalized) vector has length
\begin{equation}
\begin{aligned}
    \norm{f(W)G_0G_0^\dagger w}
    &=\norm{\sum_{-1\leq\lambda_u<1}\alpha_u\frac{\sqrt{1-\lambda_u}}{\sqrt{2}}
    \left(\frac{f\left(e^{i\arccos(\lambda_u)}\right)\ket{\phi_{u,Y+}}+
    f\left(e^{-i\arccos(\lambda_u)}\right)\ket{\phi_{u,Y-}}}{\sqrt{2}}\right)}\\
    &=\sqrt{\sum_{-1\leq\lambda_u<1}\abs{\alpha_u}^2\frac{1-\lambda_u}{2}
    \left(\frac{\abs{f\left(e^{i\arccos(\lambda_u)}\right)}^2}{2}
    +\frac{\abs{f\left(e^{-i\arccos(\lambda_u)}\right)}^2}{2}\right)}\\
    &\leq\sqrt{\sum_{-1\leq\lambda_u<1}\abs{\alpha_u}^2}
    \sqrt{\max_{\theta\in[-\pi,\pi]}\frac{1-\cos(\theta)}{2}\abs{f\left(e^{i\theta}\right)}^2}
    \leq\norm{w}\max_{\theta\in[-\pi,\pi]}\abs{
    \begin{revised}
    \sin\left(\frac{\theta}{2}\right)
    \end{revised}
    f\big(e^{i\theta}\big)}.
\end{aligned}
\end{equation}
The proof is now complete by setting $w=\Pi_1z$.
\end{proof}

\subsection{Eigenstate filtering over unit circle}
\label{sec:gap_filtering}
\begin{revised}
In light of the constrained orthogonal decomposition, our goal is to choose a filter function $f$ on the unit circle that preserves the eigenvalue $+1$ of $W$ and suppresses the component in $\operatorname{\mathbf{Im}}\left(\Pi_0\Pi_1\right)$. By~\lem{eff_gap}, it suffices to require
\begin{equation}
    f(1)=1,\qquad
    \max_{\theta\in[-\pi,\pi]}\abs{f\big(e^{i\theta}\big)}=1,
\end{equation}
and to make the quantity
\begin{equation}
\label{eq:filter_merit}
    g\left(f\right)
    =\max_{\theta\in[-\pi,\pi]}\abs{f\big(e^{i\theta}\big)\sin\left(\frac{\theta}{2}\right)}
\end{equation}
as small as possible. The weight $\sin\left(\frac{\theta}{2}\right)$ vanishes at $\theta=0$, so the filter does not need a spectral gap, and the maximum in~\eq{filter_merit} is taken over the entire unit circle.

Eigenstate filtering over the unit circle was also considered in~\cite[Page 30]{li2025discrimination} using Dolph-Chebyshev polynomials and QSVT. However, that method does not distinguish the eigenphases $\theta=0$ and $\theta=\pi$ of the input unitary, so a filter designed to preserve the eigenvalue $+1$ also preserves the eigenvalue $-1$. A simple fix is to shift the input unitary so that the eigenvalue of interest is mapped to the origin~\cite[Section V]{2021MartynGrand}. Specifically, starting from the unitary $W$, we construct the block encoding
\begin{equation}
\label{eq:shift}
    \frac{W-I}{2}
    =\left(\frac{\bra{0}+\bra{1}}{\sqrt{2}}\otimes I\right)
    \left(\ketbra{0}{0}\otimes W-\ketbra{1}{1}\otimes I\right)
    \left(\frac{\ket{0}+\ket{1}}{\sqrt{2}}\otimes I\right),
\end{equation}
whose eigenvalues $\frac{e^{i\theta_u}-1}{2}$ have magnitude $\abs{\sin\left(\frac{\theta_u}{2}\right)}$. This magnitude vanishes only at $\theta_u=0$ and equals $1$ at $\theta_u=\pi$, so any even polynomial transformation of the singular values of~\eq{shift} separates the two eigenphases. Because the polynomial is even, the phases of the eigenvalues cancel, and the transformation acts on the eigenbasis of $W$. 

In this subsection, we take a different route based on the Dirichlet kernel, which we implement directly as a linear combination of powers of $W$. This avoids the shift altogether. Compared with filtering using Dolph-Chebyshev polynomials with a circular gap, it also removes a logarithmic factor in the inverse accuracy from the query cost and halves the constant prefactor.

We achieve this using the Dirichlet kernel. For a positive odd integer $d=2m+1$ and a unitary $W=\sum_ue^{i\theta_u}\ketbra{\phi_u}{\phi_u}$, define
\begin{equation}
\label{eq:dirichlet}
    \operatorname{\mathbf{D}}_d(W)
    =\frac{1}{d}\sum_{k=-m}^{m}W^k
    =\sum_u\frac{\sin\left(\frac{d\theta_u}{2}\right)}{d\sin\left(\frac{\theta_u}{2}\right)}\ketbra{\phi_u}{\phi_u},
\end{equation}
where the second equality follows from the geometric sum $\frac{1}{d}\sum_{k=-m}^{m}e^{ik\theta}=\frac{\sin\left(\frac{d\theta}{2}\right)}{d\sin\left(\frac{\theta}{2}\right)}$, and the ratio is understood as $1$ at $\theta=0$.

\begin{lemma}[Dirichlet kernel]
\label{lem:dirichlet}
For any positive odd integer $d$, the Dirichlet kernel $\operatorname{\mathbf{D}}_d$ defined in~\eq{dirichlet} is real valued and satisfies
\begin{enumerate}
    \item $\operatorname{\mathbf{D}}_d(1)=1$;
    \item $\max_{\theta\in[-\pi,\pi]}\abs{\operatorname{\mathbf{D}}_d\big(e^{i\theta}\big)}=1$; and
    \item $g\left(\operatorname{\mathbf{D}}_d\right)=\max_{\theta\in[-\pi,\pi]}\abs{\operatorname{\mathbf{D}}_d\big(e^{i\theta}\big)\sin\left(\frac{\theta}{2}\right)}=\frac{1}{d}$.
\end{enumerate}
\end{lemma}
\begin{proof}
All three claims follow from the identity $\operatorname{\mathbf{D}}_d\big(e^{i\theta}\big)\sin\left(\frac{\theta}{2}\right)=\frac{1}{d}\sin\left(\frac{d\theta}{2}\right)$ together with the inequality $\abs{\sin(d\varphi)}\leq\min\left\{1,d\abs{\sin(\varphi)}\right\}$. The maximum in the third claim is attained at $\theta=\frac{\pi}{d}$.
\end{proof}



Interestingly, the Dirichlet kernel is optimal for this task among all Laurent polynomials in $W$ of the same degree.

\begin{proposition}[Optimality of Dirichlet kernel]
\label{prop:dirichlet_opt}
Let $d=2m+1$ be a positive odd integer, and let $f\big(e^{i\theta}\big)=\sum_{k=-m}^{m}c_ke^{ik\theta}$ with $f(1)=1$. Then $g\left(f\right)\geq\frac{1}{d}$.
\end{proposition}
\begin{proof}
The function $h(\varphi)=f\big(e^{2i\varphi}\big)\sin(\varphi)$ is a trigonometric polynomial of degree at most $2m+1=d$ satisfying $h'(0)=f(1)=1$. Bernstein's inequality for trigonometric polynomials~\cite{Kalmykov21} then gives
\begin{equation}
    1=\abs{h'(0)}\leq d\max_{\varphi\in[-\pi,\pi]}\abs{h(\varphi)}.
\end{equation}
The claim follows since $\abs{h(\varphi+\pi)}=\abs{h(\varphi)}$, so the maximum is attained for $\varphi\in\left[-\frac{\pi}{2},\frac{\pi}{2}\right]$, where $\theta=2\varphi$ ranges over $[-\pi,\pi]$.
\end{proof}

The Dirichlet kernel is a linear combination of powers of $W$ with equal coefficients, so we implement it directly using a linear combination of unitaries.

\begin{proposition}[Dirichlet filtering by linear combination of unitaries]
\label{prop:dirichlet_lcu}
Let $G_0$ and $G_1$ be isometries which define $W=\left(2G_0G_0^\dagger-I\right)\left(I-2G_1G_1^\dagger\right)$, and let $d=2m+1$ be a positive odd integer. Then the operator $\operatorname{\mathbf{D}}_d(W)$ introduced in~\eq{dirichlet} can be block encoded with normalization factor $1$ using $m$ queries to each of $G_1$ and $G_1^\dagger$, and at most $2m$ queries to each of $G_0$ and $G_0^\dagger$.
\end{proposition}
\begin{proof}
Let $\operatorname{\mathbf{PREP}}$ be a unitary on an index register with $\operatorname{\mathbf{PREP}}\ket{0}=\frac{1}{\sqrt{d}}\sum_{k=-m}^{m}\ket{k}$, and let
\begin{equation}
    \operatorname{\mathbf{SEL}}=\sum_{k=-m}^{m}\ketbra{k}{k}\otimes W^k.
\end{equation}
Then $\left(\bra{0}\operatorname{\mathbf{PREP}}^\dagger\otimes I\right)\operatorname{\mathbf{SEL}}\left(\operatorname{\mathbf{PREP}}\ket{0}\otimes I\right)=\frac{1}{d}\sum_{k=-m}^{m}W^k=\operatorname{\mathbf{D}}_d(W)$, which is a block encoding with normalization factor $\sum_{k=-m}^{m}\frac{1}{d}=1$.

To implement $\operatorname{\mathbf{SEL}}$, note that both reflections are Hermitian, so
\begin{equation}
    W=\left(2G_0G_0^\dagger-I\right)\left(I-2G_1G_1^\dagger\right),\qquad
    W^\dagger=\left(I-2G_1G_1^\dagger\right)\left(2G_0G_0^\dagger-I\right).
\end{equation}
For $t=1,\ldots,m$, define
\begin{equation}
\begin{aligned}
    L_t
    &=\left(\sum_{k\geq t}\ketbra{k}{k}\otimes\left(2G_0G_0^\dagger-I\right)+\sum_{k<t}\ketbra{k}{k}\otimes I\right)
    \left(\sum_{\abs{k}\geq t}\ketbra{k}{k}\otimes\left(I-2G_1G_1^\dagger\right)+\sum_{\abs{k}<t}\ketbra{k}{k}\otimes I\right)\\
    &\qquad\cdot\left(\sum_{k\leq -t}\ketbra{k}{k}\otimes\left(2G_0G_0^\dagger-I\right)+\sum_{k>-t}\ketbra{k}{k}\otimes I\right).
\end{aligned}
\end{equation}
On the branch $\ket{k}$, the operator $L_t$ acts as $W$ if $k\geq t$, as $W^\dagger$ if $k\leq -t$, and as $I$ if $\abs{k}<t$. Therefore, the branch $\ket{k}$ of the product $L_m\cdots L_1$ receives $W$ from each layer $t=1,\ldots,k$ if $k>0$, receives $W^\dagger$ from each layer $t=1,\ldots,-k$ if $k<0$, and receives $I$ from all other layers. This gives $L_m\cdots L_1=\operatorname{\mathbf{SEL}}$.

Each layer $L_t$ applies one controlled reflection $I-2G_1G_1^\dagger$ and two controlled reflections $2G_0G_0^\dagger-I$. Writing $G_j=U_j\left(\ket{0}\otimes I\right)$ for a unitary $U_j$, each reflection is the conjugation of a reflection about $\ket{0}$ by $U_j$, so only the reflection about the ancilla $\ket{0}$ needs to be controlled, while $U_j$ and $U_j^\dagger$ are applied without control. Hence each layer makes $1$ query to each of $G_1$ and $G_1^\dagger$, and $2$ queries to each of $G_0$ and $G_0^\dagger$. Summing over the $m$ layers yields the claimed cost.
\end{proof}
\end{revised}

\begin{revised}
In applications of the effective gap lemma~\cite{Lee_effectivegap}, the lemma is typically combined with phase estimation at a precision $\delta$, where the eigenphases with $\abs{\theta}<\delta$ and $\abs{\theta}\geq\delta$ are bounded separately. This incurs only a constant overhead when constant accuracy suffices, but leads to a query cost scaling as $\frac{1}{\epsilon^2}$ when a state must be prepared to accuracy $\epsilon$, as in~\cite{li2025new}. Phase estimation postselected on the zero outcome implements the Dirichlet kernel up to a power of $W$, which acts trivially on the component to be preserved. Hence, by our~\lem{eff_gap} and~\lem{dirichlet}, the same circuit achieves accuracy $\epsilon$ with $\operatorname{\mathbf{O}}\left(\frac{1}{\epsilon}\right)$ queries, and the symmetric implementation of~\prop{dirichlet_lcu} further halves this cost.
\end{revised}

\subsection{Constant-prefactor analysis of beyond-\texorpdfstring{$\kappa$}{k} solver based on filtering}
\label{sec:gap_solver}
\setcounter{equation}{117}
We now present our beyond-$\kappa$ solver based on eigenstate filtering over the unit circle with effective gap. We also estimate the constant prefactor of its query complexity to leading order.

Let $A$ be a matrix with $\norm{A}\leq1$ block encoded by $O_A$, and $\ket{b}$ be a normalized quantum state prepared by $O_b$.
Assume that $A$ is invertible and denote $x=A^{-1}\ket{b}$. Applying~\cor{gap_block}, we obtain a block encoding of the augmented matrix as an overlap of isometries
\begin{equation}
    \frac{1}{\sqrt{\beta^2+1}}\begin{bmatrix}
        \beta A & -\ket{b}
    \end{bmatrix}
    =G_1^\dagger G_0
\end{equation}
for any scalar $\beta>0$, using $1$ query to $O_A$ and $O_b$. Moreover, this induces the constrained orthogonal decomposition
\begin{equation}
\begin{aligned}
    G_0\begin{bmatrix}
        0\\
        1
    \end{bmatrix}
    &=\frac{\beta}{\sqrt{\norm{x}^2+\beta^2}}
    \frac{1}{\sqrt{\norm{x}^2+\beta^2}}
    G_0\begin{bmatrix}
        x\\\beta
    \end{bmatrix}
    -\frac{\norm{x}}{\sqrt{\norm{x}^2+\beta^2}}\frac{\beta}{\norm{x}\sqrt{\norm{x}^2+\beta^2}}G_0\begin{bmatrix}
        x\\
        -\frac{\norm{x}^2}{\beta}
    \end{bmatrix}\\
    &=\frac{\beta}{\sqrt{\norm{x}^2+\beta^2}}
    \frac{1}{\sqrt{\norm{x}^2+\beta^2}}G_0\begin{bmatrix}
        x\\
        \beta
    \end{bmatrix}
    -\frac{\norm{x}}{\sqrt{\norm{x}^2+\beta^2}}
    \Pi_0 \Pi_1
    \frac{\sqrt{\beta^2+1}}{\norm{x}\sqrt{\norm{x}^2+\beta^2}}G_1A^{-1\dagger}x,\\
\end{aligned}
\end{equation}
where $\Pi_0=G_0G_0^\dagger$ and $\Pi_1=G_1G_1^\dagger$.

\begin{revised}
Define the unitary $W=\left(2\Pi_0-I\right)\left(I-2\Pi_1\right)$ and perform the filtering $\operatorname{\mathbf{D}}_d(W)$ of~\prop{dirichlet_lcu} for a positive odd integer $d$ to be chosen momentarily. By the constrained orthogonal decomposition (\cor{gap_block}), we obtain
\begin{equation}
    \operatorname{\mathbf{D}}_d(W)G_0\begin{bmatrix}
        0\\
        1
    \end{bmatrix}
    =\frac{\beta}{\norm{x}^2+\beta^2}G_0\begin{bmatrix}
        x\\
        \beta
    \end{bmatrix}
    -\frac{\sqrt{\beta^2+1}}{\norm{x}^2+\beta^2}\operatorname{\mathbf{D}}_d(W)
    \Pi_0 \Pi_1G_1A^{-1\dagger}x.
\end{equation}
Here, the first term lies in $\operatorname{\mathbf{Im}}(\Pi_0)\cap\operatorname{\mathbf{Ker}}(\Pi_1)$ and is therefore an eigenvector of $W$ with eigenvalue $1$. Since $\operatorname{\mathbf{D}}_d(1)=1$ by~\lem{dirichlet}, this term is left invariant by the filtering. The second term is then bounded using the third item of the effective gap lemma (\lem{eff_gap}) and~\lem{dirichlet} as 
\begin{equation}
    \norm{\frac{\sqrt{\beta^2+1}}{\norm{x}^2+\beta^2}\operatorname{\mathbf{D}}_d(W)
    \Pi_0 \Pi_1G_1A^{-1\dagger}x}
    \leq\frac{\sqrt{\beta^2+1}}{\norm{x}^2+\beta^2}\,g\left(\operatorname{\mathbf{D}}_d\right)\norm{G_1A^{-1\dagger}x}
    =\frac{\sqrt{\beta^2+1}}{\norm{x}^2+\beta^2}\frac{\norm{A^{-1\dagger}x}}{d}.
\end{equation}
We thus obtain the following error bound for the eigenstate filtering
\begin{equation}
\label{eq:dirichlet_error}
    \norm{\operatorname{\mathbf{D}}_d(W)G_0\begin{bmatrix}
        0\\
        1
    \end{bmatrix}
    -\frac{\beta}{\norm{x}^2+\beta^2}
    G_0\begin{bmatrix}
        x\\\beta
    \end{bmatrix}}
    \leq\frac{\sqrt{\beta^2+1}}{\norm{x}^2+\beta^2}\frac{\norm{A^{-1\dagger}x}}{d}.
\end{equation}
Previous approaches to eigenstate filtering over the unit circle~\cite[Section 4.1]{li2025discrimination} use Dolph-Chebyshev polynomials with a circular gap $\delta$ and a precision $\xi$, and divide the error budget between the eigenphases with $\abs{\theta}<\delta$ and those with $\abs{\theta}\geq\delta$. This leads to a query cost with an additional factor of $\log\left(\frac{1}{\epsilon}\right)$. In contrast, the error bound~\eq{dirichlet_error} controls the entire unwanted component through the single quantity $g\left(\operatorname{\mathbf{D}}_d\right)=\frac{1}{d}$, so no splitting of the error budget is needed.

Let us choose $d$ so that the solution vector $x$ has error at most $\norm{x}\epsilon$. Equivalently, we require that
\begin{equation}
    \frac{\sqrt{\beta^2+1}}{\norm{x}^2+\beta^2}\frac{\norm{A^{-1\dagger}x}}{d}
    \leq\frac{\beta\norm{x}}{\norm{x}^2+\beta^2}\epsilon
    \quad\Longleftrightarrow\quad
    d\geq\frac{\sqrt{\beta^2+1}}{\beta}\frac{\norm{A^{-1\dagger}x}}{\norm{x}\epsilon}.
\end{equation}
By~\prop{dirichlet_lcu}, this filtering uses $\frac{d-1}{2}$ queries to each of $G_1$ and $G_1^\dagger$, and at most $d-1$ queries to each of $G_0$ and $G_0^\dagger$. Each $G_1$ or $G_1^\dagger$ requires $1$ controlled application of $O_A$ and $O_b$ whereas $G_0$ can be realized with zero query (\eq{def_g}), giving the cost
\begin{equation}
    d-1
\end{equation}
in the standard input model.

However, the exact values of $\norm{x}$ and $\norm{A^{-1\dagger}x}$ are typically unavailable in practice. Instead, we assume that $\alpha_{A^{-1\dagger}\ket{x}}\geq\norm{A^{-1\dagger}\ket{x}}$ is a norm upper bound and $\frac{1}{\mu}\leq\frac{\alpha_x}{\norm{x}}\leq\mu$ is a $\mu$-multiplicative approximation of the solution norm $\norm{x}$ for some $\mu\geq1$. Setting $\beta=\alpha_x$ and using $\frac{\norm{A^{-1\dagger}x}}{\norm{x}}=\norm{A^{-1\dagger}\ket{x}}\leq\alpha_{A^{-1\dagger}\ket{x}}$, it suffices to choose
\begin{equation}
\label{eq:dirichlet_degree}
    d=2\operatorname{\mathbf{Floor}}\left(\frac{1}{2}\operatorname{\mathbf{Ceil}}\left(\frac{\sqrt{\alpha_x^2+1}\,\alpha_{A^{-1\dagger}\ket{x}}}{\alpha_x\epsilon}\right)\right)+1,
\end{equation}
which does not depend on the unknown solution norm $\norm{x}$, where $2\operatorname{\mathbf{Floor}}\left(\frac{1}{2}\operatorname{\mathbf{Ceil}}\left(\cdot\right)\right)+1$ denotes the smallest odd integer greater than or equal to its argument.

With the above choice of $d$ and $\beta$, the success amplitude is lower bounded by
\begin{equation}
    \frac{\beta\norm{x}(1-\epsilon)}{\norm{x}^2+\beta^2}
    =\frac{\alpha_x\norm{x}(1-\epsilon)}{\norm{x}^2+\alpha_x^2}
    =\frac{(1-\epsilon)}{\frac{\norm{x}}{\alpha_x}+\frac{\alpha_x}{\norm{x}}}
    \geq\frac{1-\epsilon}{\mu+\frac{1}{\mu}}.
\end{equation}
Applying the fixed-point amplitude amplification~\cite[Page 29]{Dalzell2024shortcut}, we obtain the solution state $\ket{x}=\frac{x}{\norm{x}}$ with accuracy $\epsilon$ and failure probability $p_{\text{fail}}$ using
\begin{equation}
    2\operatorname{\mathbf{Floor}}\left(\frac{1}{2}\operatorname{\mathbf{Ceil}}\left(\frac{\mu+\frac{1}{\mu}}{1-\epsilon}\ln\left(\frac{2}{\sqrt{p_{\text{fail}}}}\right)\right)\right)+1
\end{equation}
repetitions of filtering. 

Altogether, to succeed with a probability at least $\frac{1}{2}$, the solver makes
\begin{equation}
    \left(2\operatorname{\mathbf{Floor}}\left(\frac{1}{2}\operatorname{\mathbf{Ceil}}\left(\frac{\mu+\frac{1}{\mu}}{1-\epsilon}\ln\left(2\sqrt{2}\right)\right)\right)+1\right)
    \cdot2\operatorname{\mathbf{Floor}}\left(\frac{1}{2}\operatorname{\mathbf{Ceil}}\left(\frac{\sqrt{\alpha_x^2+1}\,\alpha_{A^{-1\dagger}\ket{x}}}{\alpha_x\epsilon}\right)\right)
\end{equation}
queries to the oracles $O_A$ and $O_b$ (recall from~\eq{def_g} that undoing the isometry $G_0^\dagger$ has zero query cost). In the limit $\mu\rightarrow1$, $\epsilon\rightarrow0$, $\alpha_x\rightarrow\infty$, $\alpha_{A^{-1\dagger}\ket{x}}\rightarrow\infty$, the number of repetitions tends to $2\operatorname{\mathbf{Floor}}\left(\frac{1}{2}\operatorname{\mathbf{Ceil}}\left(3\ln\left(2\right)\right)\right)+1=3$ and the cost of each filtering tends to $\frac{\alpha_{A^{-1\dagger}\ket{x}}}{\epsilon}$, so the leading-order query complexity is
\setcounter{equation}{139}
\begin{equation}
    \frac{3\alpha_{A^{-1\dagger}\ket{x}}}{\epsilon}.
\end{equation}
\end{revised}

\subsection{Beyond-\texorpdfstring{$\kappa$}{k} solution norm estimation}
\label{sec:gap_est}
\setcounter{equation}{140}
The simple beyond-$\kappa$ solver from the previous subsection relies on the assumption that the solution norm $\norm{x}$ is known to a constant relative accuracy. If we do not have this prior knowledge, then we can turn the solver into a solution norm estimation algorithm with essentially the same query complexity up to logarithmic factors.

\begin{revised}
To start, we recall from~\eq{dirichlet_error} that the Dirichlet filtering introduces an error at most $\frac{\sqrt{\beta^2+1}}{\norm{x}^2+\beta^2}\frac{\norm{A^{-1\dagger}x}}{d}$. Unlike before, we now aim to choose $d$ so that, when projected onto the first component, the error in the success amplitude is sufficiently small. Additionally, we perform an exponential search over the rescaling factor $\beta=1,2,4,\ldots$, so $\beta\geq1$ is guaranteed. Using $\sqrt{\beta^2+1}\leq\sqrt{2}\beta$ and $\norm{A^{-1\dagger}x}=\norm{x}\norm{A^{-1\dagger}\ket{x}}\leq\norm{x}\alpha_{A^{-1\dagger}\ket{x}}$, where $\alpha_{A^{-1\dagger}\ket{x}}\geq\norm{A^{-1\dagger}\ket{x}}$ is a known upper bound, we have
\begin{equation}
    \frac{\sqrt{\beta^2+1}}{\norm{x}^2+\beta^2}\frac{\norm{A^{-1\dagger}x}}{d}
    \leq\frac{\sqrt{2}\beta\norm{x}}{\norm{x}^2+\beta^2}\frac{\alpha_{A^{-1\dagger}\ket{x}}}{d}
    \leq\frac{1}{8}\frac{\beta\norm{x}}{\norm{x}^2+\beta^2},
\end{equation}
provided that we choose
\begin{equation}
    d=2\operatorname{\mathbf{Floor}}\left(\frac{1}{2}\operatorname{\mathbf{Ceil}}\left(8\sqrt{2}\,\alpha_{A^{-1\dagger}\ket{x}}\right)\right)+1.
\end{equation}
Then, the filtering has asymptotic query complexity $\operatorname{\mathbf{O}}\left(\alpha_{A^{-1\dagger}\ket{x}}\right)$, independent of $\beta$.

Note that the success amplitude is bounded within the interval
\begin{equation}
    \left[\frac{7}{8}\frac{\beta\norm{x}}{\norm{x}^2+\beta^2},\
    \frac{9}{8}\frac{\beta\norm{x}}{\norm{x}^2+\beta^2}\right]
    \subseteq\left[\frac{7}{8}\frac{\beta\norm{x}}{\norm{x}^2+\beta^2},\
    \frac{9}{8}\frac{\beta}{\norm{x}}\right].
\end{equation}
Let us choose the rescaling factor $\beta_j=2^j$ for $j=0,1,\ldots,j_{\max}$ with $j_{\max}=\operatorname{\mathbf{Ceil}}\left(\log_2\left(\alpha_{A^{-1\dagger}\ket{x}}\right)\right)$. This range covers the solution norm, since $\norm{x}\geq\norm{Ax}=1$ and $\norm{x}^2=\bra{x}A^{-1}\ket{b}\leq\norm{A^{-1\dagger}x}$ give $1\leq\norm{x}\leq\norm{A^{-1\dagger}\ket{x}}\leq\alpha_{A^{-1\dagger}\ket{x}}$. We consider two cases. If $2^j\leq\frac{\norm{x}}{16}$, then the success amplitude is at most
\begin{equation}
    \frac{9}{8}\frac{2^j}{\norm{x}}
    \leq\frac{9}{8}\frac{\frac{\norm{x}}{16}}{\norm{x}}
    =\frac{9}{128}.
\end{equation}
On the other hand, if $\frac{\norm{x}}{2}\leq2^j\leq\norm{x}$, then the success amplitude is at least
\begin{equation}
    \frac{7}{8}\frac{2^j\norm{x}}{\norm{x}^2+4^{j}}
    \geq\frac{7}{8}\frac{\frac{\norm{x}^2}{2}}{\norm{x}^2+\norm{x}^2}
    =\frac{7}{32}
    =\frac{28}{128},
\end{equation}
and such a $j$ exists in the range $0\leq j\leq j_{\max}$.

To distinguish these two cases, we loop through $j=0,1,\ldots,j_{\max}$, and perform amplitude estimation with precision $\frac{1}{16}$ and failure probability $p=\frac{1}{3(j_{\max}+1)}$. By the union bound, with probability at least $1-(j_{\max}+1)p=\frac{2}{3}$, all estimates are within $\frac{1}{16}$ of the corresponding success amplitudes. In this event, the algorithm produces estimates at most $\frac{17}{128}$ for all $j$ such that $2^j\leq\frac{\norm{x}}{16}$, and estimates at least $\frac{20}{128}$ for all $j$ such that $\frac{\norm{x}}{2}\leq2^j\leq\norm{x}$. Suppose the first time we find the amplitude above $\frac{19}{128}$ is at $j_0$. This means
\begin{equation}
    2^{j_0}\leq\norm{x},\qquad
    2^{j_0}>\frac{\norm{x}}{16},
\end{equation}
which gives a constant multiplicative approximation of solution norm:
\begin{equation}
    \frac{1}{4}\leq\frac{\norm{x}}{2^{j_0}\cdot4}<4.
\end{equation}
Amplitude estimation with constant precision and failure probability $p$ makes $\operatorname{\mathbf{O}}\left(\log\left(\frac{1}{p}\right)\right)$ applications of the filtering and its inverse. The total query complexity is then at most
\begin{equation}
    \operatorname{\mathbf{O}}\left(\sum_{j=0}^{j_{\max}}
    \alpha_{A^{-1\dagger}\ket{x}}
    \log\left(3(j_{\max}+1)\right)
    \right)
    =\operatorname{\mathbf{O}}\left(
    \alpha_{A^{-1\dagger}\ket{x}}\log\left(\alpha_{A^{-1\dagger}\ket{x}}\right)\log\log\left(\alpha_{A^{-1\dagger}\ket{x}}\right)
    \right).
\end{equation}
\end{revised}%
The approximation ratio can be refined to $\mu$ arbitrarily close to $1$ using amplitude estimation incurring multiplicative complexity overhead $1/(\mu-1)$~\cite{Dalzell2024shortcut}. We now summarize the filtering-based beyond-$\kappa$ solver and the corresponding solution norm estimator in the following theorem.

\setcounter{equation}{150}
\begin{theorem}[Beyond-$\kappa$ solver based on filtering with effective gap]
\label{thm:beyondk_filtering}
Let $A$ be a matrix with $\norm{A}\leq1$ block encoded by $O_A$, and $\ket{b}$ be a normalized quantum state prepared by $O_b$. Assume that $A$ is invertible and denote $x=A^{-1}\ket{b}$.

Then, the normalized solution state $\ket{x}=\frac{x}{\norm{x}}$ can be prepared to accuracy \begin{revised}$0<\epsilon<1$\end{revised} and success probability $>\frac{1}{2}$ with query complexity
\begin{revised}
\begin{equation}
    \operatorname{\mathbf{O}}\left(\frac{\alpha_{A^{-1\dagger}\ket{x}}}{\epsilon}
    \operatorname{\mathbf{Cost}}\left(O_A\right)
    +\frac{\alpha_{A^{-1\dagger}\ket{x}}}{\epsilon}
    \operatorname{\mathbf{Cost}}\left(O_b\right)
    \right),
\end{equation}
\end{revised}
assuming a norm upper bound $\alpha_{A^{-1\dagger}\ket{x}} \geq \norm{A^{-1\dagger}\ket{x}}$ and a constant multiplicative approximation $\alpha_x$ of the solution norm, i.e., $\frac{1}{\mu} \leq \frac{\alpha_x}{\norm{x}} \leq \mu$ for some constant $\mu \geq 1$.

In the limit where $\mu\rightarrow1$, $\epsilon\rightarrow0$, $\alpha_x\rightarrow\infty$, $\alpha_{A^{-1\dagger}\ket{x}}\rightarrow\infty$,
the solver has leading-order query complexity
\begin{revised}
\begin{equation}
    \frac{3\alpha_{A^{-1\dagger}\ket{x}}}{\epsilon}.
\end{equation}
\end{revised}

Otherwise, the solution norm $\norm{x}$ can be estimated to a constant multiplicative accuracy with query complexity
\begin{revised}
\begin{equation}
\begin{aligned}
    \operatorname{\mathbf{O}}\bigg(&\alpha_{A^{-1\dagger}\ket{x}}\log\left(\alpha_{A^{-1\dagger}\ket{x}}\right)\log\log\left(\alpha_{A^{-1\dagger}\ket{x}}\right)
    \operatorname{\mathbf{Cost}}\left(O_A\right)\\
    &+\alpha_{A^{-1\dagger}\ket{x}}\log\left(\alpha_{A^{-1\dagger}\ket{x}}\right)\log\log\left(\alpha_{A^{-1\dagger}\ket{x}}\right)
    \operatorname{\mathbf{Cost}}\left(O_b\right)\bigg).
\end{aligned}
\end{equation}
\end{revised}
\end{theorem}

%% file: discuss.tex
In this work, we have developed algorithms that solve the quantum linear system problem with query complexity independent of the spectral condition number.
We have proposed an algorithmic template based on effective truncation of the linear system, which converts conventional quantum linear system solvers into beyond-$\kappa$ solvers. Our template applies to any solver satisfying the strong truncation property, or the weak truncation property with optimal query complexity for initial state preparation.
Additionally, we have presented a beyond-$\kappa$ solver using eigenstate filtering with effective gap.
We have shown that the filtering-based solver \begin{revised}using the Dirichlet kernel\end{revised} is particularly simple and achieves a favorable runtime prefactor when the solution norm is known.
When it is not, we have designed a similarly straightforward solution norm estimator---making our result fully self-contained.

\begin{revised}
The two solvers depend on the instance through different parameters, and neither uniformly dominates conventional solvers. Since $\keff\leq\kappa$, the truncation-based solver is never worse than conventional VTAA. The filtering-based solver has cost linear in $\frac{\norm{A^{-1\dagger}\ket{x}}}{\epsilon}$, so it outperforms conventional solvers with complexity $\operatorname{\mathbf{O}}\left(\kappa\log\left(\frac{1}{\epsilon}\right)\right)$ when $\norm{A^{-1\dagger}\ket{x}}$ is smaller than roughly $\kappa\epsilon\log\left(\frac{1}{\epsilon}\right)$, but it can be slower by a factor of order $\frac{1}{\epsilon\log(1/\epsilon)}$ when $\norm{A^{-1\dagger}\ket{x}}\approx\kappa$. Moreover, our algorithms take upper bounds on these instance-dependent parameters as input. They are correct for any valid upper bound, and their cost scales with the supplied bound. The trivial bound $\kappa$ is always available, and it recovers the conventional scaling for the truncation-based solver. In practice, one may therefore run whichever of the conventional, truncation-based, and filtering-based solvers has the smallest cost bound for the available parameter estimates. Designing solvers that adapt automatically to the unknown values of these parameters, or procedures that certify such bounds efficiently, is an interesting direction for future work.
\end{revised}

We have introduced the effective condition number $\kappa_{\mathrm{eff}}$ as a measure of complexity for the truncation-based solvers, and established a family of upper bounds $\keff\leq\left(\frac{\norm{(A)_{\mathbf{sv}}^{-t}\ket{x}}}
{\epsilon}\right)^{\frac{1}{t}}$ in terms of the singular transformation of the coefficient matrix $A$, normalized solution vector $\ket{x}$ and adjustable parameter $0<t<\infty$.
The case where $t$ is a positive integer is especially straightforward, as the transformation reduces to $(A)_{\mathbf{sv}}^{-t}=(A^\dagger A)^{-t/2}$ for even $t$ and $(A)_{\mathbf{sv}}^{-t}=A^{-1\dagger}(A^\dagger A)^{-(t-1)/2}$ for odd $t$ and can be analyzed via standard linear algebraic techniques.
In particular, choosing $t=1$ gives $(A)_{\mathbf{sv}}^{-t}=A^{-1\dagger}$.
This choice introduces a polynomial dependence on $1/\epsilon$, which can be improved to a polylogarithmic scaling by setting $t=\operatorname{\mathbf{\Theta}}\left(\frac{\log\left(\frac{1}{\epsilon}\right)}{\log\log\left(\frac{1}{\epsilon}\right)}\right)$.
Finding efficiently computable bounds on such transformations or alternative tight estimates of $\kappa_{\mathrm{eff}}$ remains an interesting open problem.
Our effective condition number $\kappa_{\mathrm{eff}}$ captures both the strength and the limitation of truncation-based solvers. Its reciprocal is the smallest cutoff for which the truncated linear system can still be solved to accuracy $\epsilon$.
However, it remains plausible that solvers based on different techniques could achieve scaling below $\keff$. An interesting direction is to either construct such solvers or prove a lower bound ruling out their existence.
Along similar lines, the filtering-based solver is much simpler but only achieves a scaling with $\frac{\norm{A^{-1\dagger}\ket{x}}}{\epsilon}\geq\kappa_{\mathrm{eff}}$. Whether this gap can be closed by a simple algorithm is an interesting question for future work.

In analyzing our solvers, we have focused primarily on the standard input model, where the coefficient matrix $A$ and initial state $\ket{b}$ are accessed through separate oracles. However, we have also defined the affine dilation model, which block encodes $A$ and $\ket{b}$ jointly, allowing further improvements in query complexity.
To demonstrate that this improvement can be substantial, we now give a toy example where the separation between the two models diverges as $\kappa \to \infty$. Specifically, consider $A=\left[\begin{smallmatrix}
        \frac{1}{\kappa} & 0\\
        0 & 1
    \end{smallmatrix}\right]$, $\ket{b}=\left[\begin{smallmatrix}
        \frac{1}{\kappa}\\
        \sqrt{1-\frac{1}{\kappa^2}}
    \end{smallmatrix}\right]$, $x=A^{-1}\ket{b}=\left[\begin{smallmatrix}
        1\\
        \sqrt{1-\frac{1}{\kappa^2}}
    \end{smallmatrix}\right]$, $A^{-1\dagger}x=\left[\begin{smallmatrix}
        \kappa\\
        \sqrt{1-\frac{1}{\kappa^2}}
    \end{smallmatrix}\right]$.
In the standard model, we can block encode $A$ with normalization $1$ and prepare the normalized state $\ket{b}$. Then the quantity $\norm{A^{-1\dagger}\ket{x}}$ which determines the runtime of beyond-$\kappa$ solvers scales like
\begin{equation}
    \norm{A^{-1\dagger}\ket{x}}
    =\operatorname{\mathbf{\Theta}}\left(\frac{\sqrt{\kappa^2+1-\frac{1}{\kappa^2}}}{\sqrt{1+1-\frac{1}{\kappa^2}}}\right)
    =\operatorname{\mathbf{\Theta}}(\kappa).
\end{equation}
In contrast, we can define $\widetilde{A}=\frac{1}{\widetilde{\alpha}}
    \left[\begin{smallmatrix}
        MA & -M\ket{b}\\
        0 & 1
    \end{smallmatrix}\right]$, $\ket{\widetilde{b}}=\left[\begin{smallmatrix}
        0\\
        1
    \end{smallmatrix}\right]$, $\widetilde{x}=\widetilde{A}^{-1}\ket{\widetilde{b}}
    =\widetilde{\alpha}\left[\begin{smallmatrix}
        x\\
        1
    \end{smallmatrix}\right]$, and 
    $\widetilde{A}^{-1\dagger}\widetilde{x}
    =\widetilde{\alpha}^2\left[\begin{smallmatrix}
        M^{-1\dagger}A^{-1\dagger}x\\
        \norm{x}^2+1
    \end{smallmatrix}\right]$
in the affine dilation model, where $M=\left[\begin{smallmatrix}
    \kappa & 0\\
    0 & 1
\end{smallmatrix}\right]$ is a diagonal matrix and $\widetilde{\alpha}=\operatorname{\mathbf{\Theta}}(1)$ is the normalization factor for block encoding, giving 
\begin{equation}
    \norm{\widetilde{A}^{-1\dagger}\ket{\widetilde{x}}}
    =\widetilde{\alpha}\frac{\sqrt{\norm{M^{-1\dagger}A^{-1\dagger}x}^2+\left(\norm{x}^2+1\right)^2}}{\sqrt{\norm{x}^2+1}}
    =\operatorname{\mathbf{\Theta}}(1).
\end{equation}
This shows that the affine dilation model can offer an arbitrarily large advantage as the condition number $\kappa$ grows.

The preceding example is in fact a special case of a broader strategy known as preconditioning, which solves the modified linear system $MAx=M\ket{b}$ within the affine dilation model. More generally, our framework opens a new regime for designing preconditioners for quantum linear system solvers: rather than targeting the condition number $\kappa(A)$ of matrix $A$, the purpose of this new preconditioning strategy is to redistribute how $\ket{b}$ aligns with the left singular vectors of $A$, thereby reducing the effective condition number. This can yield substantial complexity improvements even when $\kappa(MA)$ itself remains large, since the gains are captured entirely by the effective condition number.

For sparse linear systems,~\cite{li2025new} constructs a preconditioner using diagonal matrices to balance the row norms.  
Interestingly, the instance-aware parameter $\operatorname{\mathbf{ET}} = \sum_i p_i^2\, d_i$ introduced in~\cite{li2025new} can be recovered within our framework through diagonal preconditioning. Here $p_i = \abs{\bra{i}(AA^\dagger)^{-1}|b\rangle}$ is the $i$-th component of vector $(AA^\dagger)^{-1}|b\rangle$ in absolute value and $d_i = \|a_i\|^2 + |b_i|^2$ is the squared norm of the $i$-th row of the augmented matrix $\begin{bmatrix} A & -|b\rangle \end{bmatrix}$. The sparse access model normalizes each row by its norm, implicitly block encoding the reweighted system $\begin{bmatrix} MA & -M|b\rangle \end{bmatrix}$ with diagonal matrix $M = \operatorname{\mathbf{Diag}}([\ldots,1/\sqrt{d_i},\ldots])$. A direct calculation then gives $(MA)^{-1\dagger}x = M^{-1}(AA^\dagger)^{-1}|b\rangle$, so that $\left\|(MA)^{-1\dagger}x\right\|^2 = \sum_i p_i^2\, d_i = \operatorname{\mathbf{ET}}$ recovers the result of~\cite{li2025new}.  
Moreover, by designing new input models, one can unlock preconditioning techniques beyond those available in the existing framework. To illustrate, consider a setting where $\norm{\ketbra{b}{b}A}$ and $\norm{(I - \ketbra{b}{b})A}$ are substantially different. In this case, one can introduce an input model that block encodes $\ketbra{b}{b}A$ and $(I - \ketbra{b}{b})A$ separately, each with its own normalization factor. 
This effectively preconditions the linear system with respect to the orthogonal decomposition $\operatorname{\mathbf{Im}}(\ketbra{b}{b})\obot\operatorname{\mathbf{Ker}}(\ketbra{b}{b})$\begin{revised}, where $\obot$ denotes the orthogonal direct sum\end{revised}. 
We expect that a systematic exploration of this preconditioning regime, identifying which classical preconditioners can be lifted into suitable input models and what complexity gains they unlock, is a fruitful direction for future work.
As another example,~\cite{Orsucci2021solvingclassesof} showed that for positive-definite matrices, the worst-case runtime lower bound of a quantum linear system solver remains linear in $\kappa$, matching the indefinite case. Nevertheless, they identified broad classes of positive-definite systems where this lower bound can be circumvented, achieving a quadratic improvement to $\sqrt{\kappa}$ via more efficient block-encodings or Cholesky-type preconditioners. It would be interesting to investigate whether our beyond-$\kappa$ techniques can be combined with such structural assumptions to yield further speedups.

Setting aside these technical refinements, we believe the most impactful direction for future work is to identify practical applications where the condition number is prohibitively large for conventional solvers, yet beyond-$\kappa$ solvers remain efficient. Our filtering-based solver is particularly appealing in this regard---its simplicity makes it a prime candidate for compilation into explicit quantum circuits. A concrete resource estimate for running this algorithm on a fault-tolerant quantum computer, even for a small-scale instance, would be an important step toward demonstrating practical quantum advantage for solving linear systems.

%% file: li.tex
The beyond-$\kappa$ solver of Li~\cite{li2025new} was originally proposed for solving sparse linear systems. In this appendix, we show that it can be extended to work in the standard input model, where the coefficient matrix is accessed through a block encoding oracle and the initial vector is prepared by a unitary. 
To simplify the discussion, we use the precise values of vector norms such as $\norm{x}$ and $\norm{A^{-1\dagger}x}$, rather than their known upper and lower bounds.

Specifically, suppose that $A$ is a matrix with $\norm{A}\leq1$ block encoded by $O_A$, and $\ket{b}$ be a normalized quantum state prepared by $O_b$. Assume that $A$ is invertible and define the solution vector by $x=A^{-1}\ket{b}$. Then we know from~\cor{gap_block} that, for any scalar $\beta>0$, the augmented matrix
\begin{equation}
    \frac{1}{\sqrt{\beta^2+1}}\begin{bmatrix}
        \beta A & -\ket{b}
    \end{bmatrix}
    =G_1^\dagger G_0
\end{equation}
can be block encoded in terms of an overlap of isometries using $1$ query to $O_A$ and $O_b$. Moreover, this induces a constrained orthogonal decomposition
\begin{equation}
\begin{aligned}
    G_0\begin{bmatrix}
        0\\
        1
    \end{bmatrix}&=\frac{\beta}{\norm{x}^2+\beta^2}
    G_0\begin{bmatrix}
        x\\
        \beta
    \end{bmatrix}
    -\frac{\sqrt{\beta^2+1}}{\norm{x}^2+\beta^2}
    \Pi_0 \Pi_1
    G_1A^{-1\dagger}x\\
    &\in\left(\operatorname{\mathbf{Im}}(\Pi_0)\cap\operatorname{\mathbf{Ker}}(\Pi_1)\right)
    \obot\operatorname{\mathbf{Im}}\left(\Pi_0 \Pi_1\right),
\end{aligned}
\end{equation}
where $\Pi_0=G_0G_0^\dagger$ and $\Pi_1=G_1G_1^\dagger$.

Here, the goal is to preserve the first term that encodes the solution vector $x$ while suppressing the second term using the effective gap lemma. 
In the main text, we show that this can be achieved \begin{revised}with eigenstate filtering based on the Dirichlet kernel, implemented by a linear combination of unitaries\end{revised}.
In~\cite{li2025new}, this was instead realized using quantum phase estimation. More concretely,~\cite[Lemma 3.4]{li2025new} shows that applying quantum phase estimation produces the normalized state $\frac{1}{\sqrt{\norm{x}^2+\beta^2}}G_0\begin{bmatrix}
    x\\
    \beta
\end{bmatrix}$ with accuracy $\xi$ in trace distance, using
\begin{equation}
    \operatorname{\mathbf{O}}\left(\frac{\norm{x}^2+\beta^2}{\beta^2\xi^2}
    \frac{\sqrt{\beta^2+1}}{\norm{x}^2+\beta^2}\norm{A^{-1\dagger}x}\right)
    =\operatorname{\mathbf{O}}\left(\frac{1}{\beta\xi^2}
    \norm{A^{-1\dagger}x}\right)
\end{equation}
queries to the oracles $O_A$ and $O_b$ and succeeding with probability 
\begin{equation}
    \operatorname{\mathbf{\Theta}}\left(\frac{\beta^2}{\norm{x}^2+\beta^2}\right),
\end{equation}
assuming $\beta=\operatorname{\mathbf{\Omega}}(1)$.

To ensure that the normalized solution $\ket{x}=\frac{x}{\norm{x}}$ is produced with accuracy $\epsilon$, Li chose $\xi=\operatorname{\mathbf{\Theta}}\left(\frac{\norm{x}^2\epsilon}{\norm{x}^2+\beta^2}\right)$.
Then, he boosted the success probability close to unity by repeating the quantum phase estimation
\begin{equation}
    \operatorname{\mathbf{\Theta}}\left(\frac{\norm{x}^2+\beta^2}{\beta^2}\frac{\norm{x}^2+\beta^2}{\norm{x}^2}\right)
    =\operatorname{\mathbf{\Theta}}\left(\frac{\norm{x}^2}{\beta^2}+\frac{\beta^2}{\norm{x}^2}\right)
\end{equation}
times. This leads to a total query complexity of
\begin{equation}
    \operatorname{\mathbf{O}}\left(\left(\frac{\norm{x}^2}{\beta^2}+\frac{\beta^2}{\norm{x}^2}\right)
    \frac{\left(\norm{x}^2+\beta^2\right)^2}{\norm{x}^4\epsilon^2}
    \frac{\norm{A^{-1\dagger}x}}{\beta}
    \right)
    =\operatorname{\mathbf{O}}\left(\left(\frac{\norm{x}^2}{\beta^2}+\frac{\beta^6}{\norm{x}^6}\right)
    \frac{\norm{A^{-1\dagger}x}}{\beta\epsilon^2}
    \right).
\end{equation}

When the solution norm $\norm{x}$ is known, Li set $\beta=\operatorname{\mathbf{\Theta}}(\norm{x})$, simplifying the query cost to
\begin{equation}
    \operatorname{\mathbf{O}}\left(
    \frac{\norm{A^{-1\dagger}x}}{\norm{x}\epsilon^2}
    \right).
\end{equation}
Otherwise, he chose $\beta=1$. Under the convention that $\norm{x}\geq1$, the query complexity then becomes
\begin{equation}
    \operatorname{\mathbf{O}}\left(
    \frac{\norm{A^{-1\dagger}x}\norm{x}^2}{\epsilon^2}
    \right).
\end{equation}

%% file: affine.tex
To facilitate a direct comparison with prior quantum linear system solvers, we have focused on the standard input model when presenting our beyond-$\kappa$ solvers in the main text.
In this appendix, we analyze beyond-$\kappa$ solvers in the affine dilation model, which allows further refinement of the query complexity.

\subsection{Analysis of truncation-based solvers in the affine dilation model}
\label{append:affine_trunc}
In the affine dilation model, the input to a quantum linear system is accessed through the block encoding
\begin{equation}
    \widetilde{A}=\frac{1}{\widetilde{\alpha}}\begin{bmatrix}
        A & -\ket{b}\\
        0 & c
    \end{bmatrix}.
\end{equation}
whereas the dilated initial vector is $\ket{\widetilde{b}}=\begin{bmatrix}
    0\\
    1
\end{bmatrix}$.
Here, we assume that $\ket{b}$ is normalized after a potential rescaling of the block encoding. For truncation-based solvers, we take $c>0$ so that the dilated matrix is invertible. The normalization factor $\widetilde{\alpha}$ then satisfies $\widetilde{\alpha}\geq\max\{\norm{A},c,1\}$.

After matrix inversion, we have 
\begin{equation}
    \widetilde{A}^{-1}=\widetilde{\alpha}\begin{bmatrix}
        A^{-1} & \frac{1}{c}x\\
        0 & \frac{1}{c}
    \end{bmatrix},
\end{equation}
which encodes the dilated solution vector
\begin{equation}
    \widetilde{x}=\widetilde{A}^{-1}\ket{\widetilde{b}}
    =\frac{\widetilde{\alpha}}{c}\begin{bmatrix}
        x\\
        1
    \end{bmatrix}.
\end{equation}
The performance of the truncation-based solvers is still determined by the effective condition number from \thm{keff}. The only difference is that $\kappa_{\mathrm{eff}}$ is now defined in terms of the dilated $\widetilde{A}$ and $\ket{\widetilde{b}}$, and we denote it $\widetilde{\kappa}_{\mathrm{eff}}$.

It is generally challenging to give a closed-form formula for $\widetilde{\kappa}_{\mathrm{eff}}$ analogous to the one in \thm{keff}, since the singular value decomposition of the dilated matrix $\widetilde{A}$ does not relate simply to that of $A$. However, the effective condition number admits the succinct upper bound $\widetilde{\kappa}_{\mathrm{eff}}\leq\frac{\norm{\widetilde{A}^{-1\dagger}\widetilde{x}}}{\norm{\widetilde{x}}\epsilon}$
same as in the standard input model, where $\norm{\widetilde{A}^{-1\dagger}\widetilde{x}}$ can be evaluated as
\begin{equation}
    \norm{\widetilde{A}^{-1\dagger}\widetilde{x}}
    =\widetilde{\alpha}^2\norm{\begin{bmatrix}
        \frac{A^{-1\dagger}x}{c}\\[.25em]
        \frac{\norm{x}^2+1}{c^2}
    \end{bmatrix}}
    =\widetilde{\alpha}^2\sqrt{\frac{\norm{A^{-1\dagger}x}^2}{c^2}+\frac{\big(\norm{x}^2+1\big)^2}{c^4}}.
\end{equation}
This implies that the effective condition number for the dilated system is at most
\begin{equation}
    \widetilde{\kappa}_{\mathrm{eff}}\leq\frac{\widetilde{\alpha}}{\epsilon}\frac{1}{\sqrt{\norm{x}^2+1}}
    \sqrt{\norm{A^{-1\dagger}x}^2+\frac{\big(\norm{x}^2+1\big)^2}{c^2}}.
\end{equation}

\subsection{Analysis of filtering-based solver in the affine dilation model}
\label{append:affine_filtering}
For the filtering-based solver, we instead consider the augmented matrix
\begin{equation}
    \widetilde{A}=\frac{1}{\alpha}
    \begin{bmatrix}
        A & -\ket{b}
    \end{bmatrix}.
\end{equation}
This corresponds to setting $c = 0$ in the generic affine dilation model and removing the redundant zero blocks at the bottom.

Suppose that this is block encoded by an overlap of isometries as $\widetilde{A}=\frac{1}{\alpha}
    \begin{bmatrix}
        A & -\ket{b}
    \end{bmatrix}=G_1^\dagger G_0$.
Then we have the constrained orthogonal decomposition
\begin{equation}
\begin{aligned}
    G_0\begin{bmatrix}
        0\\
        1
    \end{bmatrix}&=\frac{1}{\norm{x}^2+1}
    G_0\begin{bmatrix}
        x\\
        1
    \end{bmatrix}
    -\frac{\alpha}{\norm{x}^2+1}
    \Pi_0 \Pi_1
    G_1A^{-1\dagger}x\\
    &\in\left(\operatorname{\mathbf{Im}}(\Pi_0)\cap\operatorname{\mathbf{Ker}}(\Pi_1)\right)
    \obot\operatorname{\mathbf{Im}}\left(\Pi_0 \Pi_1\right),
\end{aligned}
\end{equation}
where $\Pi_0=G_0G_0^\dagger$ and $\Pi_1=G_1G_1^\dagger$.
\begin{revised}
Performing the filtering $\operatorname{\mathbf{D}}_d(W)$ of~\prop{dirichlet_lcu} with the walk operator $W=\left(2\Pi_0-I\right)\left(I-2\Pi_1\right)$ then gives
\begin{equation}
    \operatorname{\mathbf{D}}_d(W)G_0\begin{bmatrix}
        0\\
        1
    \end{bmatrix}=\frac{1}{\norm{x}^2+1}
    G_0\begin{bmatrix}
        x\\
        1
    \end{bmatrix}
    -\frac{\alpha}{\norm{x}^2+1}
    \operatorname{\mathbf{D}}_d(W)
    \Pi_0 \Pi_1
    G_1A^{-1\dagger}x,
\end{equation}
where, by~\lem{eff_gap} and~\lem{dirichlet},
\begin{equation}
    \norm{\operatorname{\mathbf{D}}_d(W)
    \Pi_0 \Pi_1G_1A^{-1\dagger}x}
    \leq g\left(\operatorname{\mathbf{D}}_d\right)\norm{A^{-1\dagger}x}
    =\frac{\norm{A^{-1\dagger}x}}{d}.
\end{equation}

Assuming that the solution norm $\norm{x}=\norm{A^{-1}\ket{b}}$ is estimated to a constant multiplicative accuracy, we set
\begin{equation}
    d=\operatorname{\mathbf{\Theta}}\left(\frac{\alpha\norm{A^{-1\dagger}x}}{\norm{x}\epsilon}\right).
\end{equation}
Note that $d=\operatorname{\mathbf{\Omega}}(1)$,
since $\epsilon=\operatorname{\mathbf{O}}(1)$ and
\begin{equation}
    \frac{\alpha\norm{A^{-1\dagger}x}}{\norm{x}}
    =\alpha\norm{A^{-1\dagger}\ket{x}}
    \geq\norm{A}\norm{A^{-1\dagger}\ket{x}}
    =\norm{A^\dagger}\norm{A^{-1\dagger}\ket{x}}
    \geq\norm{A^\dagger A^{-1\dagger}\ket{x}}
    =1.
\end{equation}
Finally, we perform $\operatorname{\mathbf{O}}\left(\frac{\norm{x}^2+1}{\norm{x}}\right)$
steps of amplitude amplification, so the total query complexity is
\begin{equation}
    \operatorname{\mathbf{O}}\left(\frac{\alpha}{\epsilon}\frac{\norm{x}^2+1}{\norm{x}^2}\norm{A^{-1\dagger}x}\right).
\end{equation}
\end{revised}

\subsection{Complexity comparison}
\label{append:affine_compare}
We now compare the runtimes of the truncation-based and filtering-based solvers in the affine dilation model. In particular, we show that the query complexity of the former is no worse than that of the latter up to logarithmic factors, as in the standard input model.

To this end, suppose we start with a block encoding of the augmented matrix $\frac{1}{\alpha}\begin{bmatrix}
        A & -\ket{b}
    \end{bmatrix}=G_1^\dagger G_0$. Then, we can construct a block encoding of the affine dilation
\begin{equation}
    \widetilde{A}=\frac{1}{\sqrt{\alpha^2+c^2}}\begin{bmatrix}
        A & -\ket{b}\\
        0 & c
    \end{bmatrix}
\end{equation}
with normalization factor $\sqrt{\alpha^2+c^2}$. Indeed, this can be achieved through the following factorization of the dilated matrix as an overlap of isometries
\begin{equation}
    \frac{1}{\sqrt{\alpha^2+c^2}}\begin{bmatrix}
        A & -\ket{b}\\
        0 & c
    \end{bmatrix}
    =\begin{bmatrix}
        G_1^\dagger & 0\\
        0 & \begin{bmatrix}
            1 & 0 & 0
        \end{bmatrix}
    \end{bmatrix}
    \frac{1}{\sqrt{\alpha^2+c^2}}
    \begin{bmatrix}
        \alpha G_0\\
        c\begin{bmatrix}
            0 & 1\\
            I & 0\\
            0 & 0
        \end{bmatrix}
    \end{bmatrix}.
\end{equation}
Then the effective condition number is bounded by
\begin{equation}
    \widetilde{\kappa}_{\mathrm{eff}}\leq\frac{\sqrt{\alpha^2+c^2}}{\epsilon\sqrt{\norm{x}^2+1}}
    \sqrt{\norm{A^{-1\dagger}x}^2+\frac{\big(\norm{x}^2+1\big)^2}{c^2}}
    \leq\frac{\sqrt{\alpha^2+c^2}}{\epsilon}\frac{\norm{A^{-1\dagger}x}}{\sqrt{\norm{x}^2+1}}
    +\frac{\sqrt{\alpha^2+c^2}}{\epsilon}\frac{\sqrt{\norm{x}^2+1}}{c}.
\end{equation}

Let us choose
\begin{equation}
    c=\alpha\frac{\norm{x}^2+1}{\norm{x}}
    \geq2\alpha,
\end{equation}
so that
\begin{equation}
    \sqrt{\alpha^2+c^2}
    =\operatorname{\mathbf{O}}\left(\alpha\frac{\norm{x}^2+1}{\norm{x}}\right).
\end{equation}
Meanwhile, by our block encoding convention,
\begin{equation}
    1=\norm{\ket{b}}=\norm{Ax}
    \leq\alpha\norm{x}
    \leq\alpha^2\norm{A^{-1\dagger}x},
\end{equation}
which gives
\begin{equation}
    c=\frac{\norm{x}^2+1}{\frac{\norm{x}}{\alpha}}
    \geq\frac{\norm{x}^2+1}{\norm{A^{-1\dagger}x}}.
\end{equation}

Altogether,
\begin{equation}
    \widetilde{\kappa}_{\mathrm{eff}}
    =\operatorname{\mathbf{O}}\left(\frac{\alpha}{\epsilon}\frac{\norm{x}^2+1}{\norm{x}}\frac{\norm{A^{-1\dagger}x}}{\sqrt{\norm{x}^2+1}}\right)
    =\operatorname{\mathbf{O}}\left(\frac{\alpha}{\epsilon}\frac{\norm{x}^2+1}{\norm{x}^2}\norm{A^{-1\dagger}x}\right).
\end{equation}
Hence, with this specific choice of $c$ in the affine dilation model, the runtime of the truncation-based algorithm matches that of the filtering-based method up to logarithmic factors.

%% file: trunc_general.tex
In the main text, we have defined the strong and weak truncation property (\defn{strong_trunc} and \defn{weak_trunc}) for a quantum linear system solver. Informally, these capture the requirement that if one supplies $\alpha$ as the condition number parameter, the solver correctly inverts singular value components above $1/\alpha$ and discards those below. To obtain a beyond-$\kappa$ solver, we have required that the truncation properties hold when $\alpha$ is at least the effective condition number from~\thm{keff}.
More generally, we can define the strong truncation property for an arbitrary $\alpha$ as follows.

\begin{definition}[Generalized strong truncation property]
\label{defn:strong_trunc_gen}
Let $A$ be a matrix with $\norm{A}\leq1$ block encoded by $O_A$, and $\ket{b}$ be a normalized quantum state prepared by $O_b$.
Assume that $A$ is invertible and denote $x=A^{-1}\ket{b}$ and $\ket{x}=\frac{x}{\norm{x}}$.
Suppose that $U\left(O_A,O_b,\kappa,\epsilon\right)$ is a quantum linear system solver such that with success probability $>\frac{1}{2}$,
\begin{equation}
\begin{aligned}
    \norm{\frac{\left(\bra{0}\otimes I\right)U\left(O_A,O_b,\kappa,\epsilon\right)\ket{00}}{\norm{\left(\bra{0}\otimes I\right)U\left(O_A,O_b,\kappa,\epsilon\right)\ket{00}}}-\ket{x}}
    &=\operatorname{\mathbf{O}}(\epsilon),
\end{aligned}
\end{equation}
provided that $\kappa\geq\norm{A^{-1}}$ and $\epsilon>0$.

We say that the solver satisfies the \emph{generalized strong truncation property} if with success probability $>\frac{1}{2}$,
\begin{equation}
\begin{aligned}
    \norm{\frac{\left(\bra{0}\otimes I\right)U\left(O_A,O_b,\alpha,\epsilon\right)\ket{00}}
    {\norm{\left(\bra{0}\otimes I\right)U\left(O_A,O_b,\alpha,\epsilon\right)\ket{00}}}
    -\frac{\Pi_{\text{right},\left[\alpha^{-1},1\right]}x}{\norm{\Pi_{\text{right},\left[\alpha^{-1},1\right]}x}}}
    &=\operatorname{\mathbf{O}}(\epsilon),
\end{aligned}
\end{equation}
provided that $1\leq\alpha\leq\kappa$, where $\Pi_{\text{right},\left[\alpha^{-1},1\right]}$ is the right singular vector projection of $A$ associated with singular values from $\left[\alpha^{-1},1\right]$.
\end{definition}

When configured with $\alpha$, VTAA correctly inverts all singular values above $1/\alpha$. Combining this with a standard amplitude amplification yields the truncated solution state supported on singular values above $1/\alpha$. We thus conclude that VTAA satisfies the generalized strong truncation property with a query overhead from amplification. Similarly, one can show that VTAA satisfies the following generalized version of weak truncation property.

\begin{definition}[Generalized weak truncation property]
\label{defn:weak_trunc_gen}
Let $A$ be a matrix with $\norm{A}\leq1$ block encoded by $O_A$, and $\ket{b}$ be a normalized quantum state prepared by $O_b$.
Assume that $A$ is invertible and denote $x=A^{-1}\ket{b}$ and $\ket{x}=\frac{x}{\norm{x}}$.
Suppose that $U\left(O_A,O_b,\kappa,\epsilon\right)$ is a quantum linear system solver such that with success probability $>\frac{1}{2}$,
\begin{equation}
\begin{aligned}
    \norm{\frac{\left(\bra{0}\otimes I\right)U\left(O_A,O_b,\kappa,\epsilon\right)\ket{00}}{\norm{\left(\bra{0}\otimes I\right)U\left(O_A,O_b,\kappa,\epsilon\right)\ket{00}}}-\ket{x}}
    &=\operatorname{\mathbf{O}}(\epsilon),
\end{aligned}
\end{equation}
provided that $\kappa\geq\norm{A^{-1}}$ and $\epsilon>0$.

We say that the solver satisfies the \emph{generalized weak truncation property} if with success probability $>\frac{1}{2}$,
\begin{equation}
\begin{aligned}
    \norm{\frac{\left(\bra{0}\otimes I\right)U\left(O_A,O_b,\alpha,\epsilon\right)\ket{00}}
    {\norm{\left(\bra{0}\otimes I\right)U\left(O_A,O_b,\alpha,\epsilon\right)\ket{00}}}
    -\ket{x}}
    &=\operatorname{\mathbf{O}}(\epsilon),
\end{aligned}
\end{equation}
provided that $1\leq\alpha\leq\kappa$, and $\ket{b}\in\operatorname{\mathbf{Im}}\left(\Pi_{\text{left},\left[\alpha^{-1},1\right]}\right)$ where $\Pi_{\text{left},\left[\alpha^{-1},1\right]}$ is the left singular vector projection of $A$ associated with singular values from $\left[\alpha^{-1},1\right]$.
\end{definition}